\documentclass[a4paper,cleveref, autoref,english,thm-restate]{lipics-v2021}
\pdfoutput=1 %uncomment to ensure pdflatex processing (mandatatory e.g. to submit to arXiv)
\hideLIPIcs  %uncomment to remove references to LIPIcs series (logo, DOI, ...), e.g. when preparing a pre-final version to be uploaded to arXiv or another public repository
\nolinenumbers

\title{Approximating the Geometric Knapsack Problem \\ in Near-Linear Time and Dynamically} %TODO Please add

\titlerunning{Approximating the geometric knapsack problem in near-linear time and
	dynamically} %TODO optional, please use if title is longer than one line

\author{Moritz Buchem}{Technische Universität München, Munich, Germany}{}{}{}%TODO mandatory, please use full name; only 1 author per \author macro; first two parameters are mandatory, other parameters can be empty. Please provide at least the name of the affiliation and the country. The full address is optional. Use additional curly braces to indicate the correct name splitting when the last name consists of multiple name parts.

\author{Paul Deuker}{Technische Universität München, Munich, Germany}{}{}{}
\author{Andreas Wiese}{Technische Universität München, Munich, Germany}{}{}{}

\authorrunning{M. Buchem, P. Deuker, A. Wiese} %TODO mandatory. First: Use abbreviated first/middle names. Second (only in severe cases): Use first author plus 'et al.'

\Copyright{Moritz Buchem, Paul Deuker, Andreas Wiese} %TODO mandatory, please use full first names. LIPIcs license is "CC-BY";  http://creativecommons.org/licenses/by/3.0/

 %TODO mandatory: Please choose ACM 2012 classifications from https://dl.acm.org/ccs/ccs_flat.cfm 

 %TODO mandatory; please add comma-separated list of keywords

\category{} %optional, e.g. invited paper

\EventEditors{Wolfgang Mulzer and Jeff M. Phillips}
\EventNoEds{2}
\EventLongTitle{40th International Symposium on Computational Geometry (SoCG 2024)}
\EventShortTitle{SoCG 2024}
\EventAcronym{SoCG}
\EventYear{2024}
\EventDate{June 11-14, 2024}
\EventLocation{Athens, Greece}
\EventLogo{socg-logo}
\SeriesVolume{293}
\ArticleNo{XX}     % <-- This will be filled in by the typesetters
\begin{document}
	
	\global\long\def\OPT{\mathrm{OPT}}%
	\global\long\def\I{\mathcal{I}}%
	\global\long\def\H{\mathcal{H}}%
	
	\global\long\def\N{\mathcal{N}^{*}}%
	
	\global\long\def\V{\mathcal{V}}%
	\global\long\def\S{\mathcal{S}}%
	\global\long\def\L{\mathcal{L}}%
	
	\global\long\def\B{\mathcal{B}}%
	
	\global\long\def\Q{\mathcal{Q}}%
	
	\global\long\def\ALG{\mathrm{ALG}}%
	
	\global\long\def\S{\mathcal{S}}%
	\global\long\def\T{\mathcal{T}}%

	\global\long\def\P{\mathcal{P}}%
	\global\long\def\el{\epsilon_{\mathrm{large}}}%
	\global\long\def\es{\epsilon_{\mathrm{small}}}%

\maketitle

%TODO mandatory: add short abstract of the document
\begin{abstract}
One important goal in algorithm design is determining
the best running time for solving a problem (approximately). For some problems, we
know the optimal running time, assuming certain conditional lower
bounds. In this paper, we study the $d$-dimensional geometric knapsack
problem in which we are far from this level of understanding. We are given a set of weighted
$d$-dimensional geometric items like squares, rectangles, or hypercubes
and a knapsack which is a square or a (hyper-)cube. Our goal is to
select a subset of the given items that fit non-overlappingly inside
the knapsack, maximizing the total profit of the packed items. We make a significant step towards determining the best running time for solving these problems approximately by presenting approximation algorithms whose running times are near-linear, i.e.,
$O(n\cdot\mathrm{poly}(\log n))$, for any constant $d$ and any parameter~$\epsilon>0$ (the exponent of $\log n$ depends on $d$ and $1/\epsilon$)

In the case of (hyper)-cubes, we present a $(1+\epsilon)$-approximation algorithm. This improves drastically upon the currently best known algorithm which is a $(1+\epsilon)$-approximation algorithm with a running time of $n^{O_{\epsilon,d}(1)}$ where the exponent of $n$ depends exponentially on $1/\epsilon$
and $d$. In particular, our algorithm is an efficient polynomial time approximation scheme (EPTAS). Moreover, we present a $(2+\epsilon)$-approximation algorithm for rectangles in the setting without rotations and a $(\frac{17}{9}+\epsilon)\approx 1.89$-approximation algorithm if we allow rotations by 90 degrees.
The best known polynomial time algorithms for this setting have approximation ratios of $\frac{17}{9}+\epsilon$
and $1.5+\epsilon$, respectively, and running times in which the exponent of $n$ depends exponentially on $1/\epsilon$.
In addition, we give dynamic algorithms with polylogarithmic query and update times, having the same approximation guarantees as our other algorithms above.

Key to our results is a new family of structured packings which we call \emph{easily guessable packings}. They are flexible enough to guarantee the existence of profitable solutions while providing enough structure so that we can compute these solutions very quickly.

\end{abstract}

\section{Introduction}

\textsc{Knapsack} is a fundamental problem in combinatorial
optimization. We are given a knapsack with a specified capacity $W$ and
a set of $n$ items, each of them characterized by its size $s_i$ and its profit $p_i$.
The goal is to compute a set of items that fits into the knapsack,
maximizing the total profit. \textsc{Knapsack} is very well understood:
% after long line of research {[}???{]}
there is an FPTAS for the
problem with a running time of only $\tilde{O}(n+(1/\epsilon)^{2.2})$
\cite{deng2023approximating} with an asymptotically almost matching conditional lower bound
of $(n+1/\epsilon)^{2-o(1)}$~\cite{cygan2019problems,knnemann_et_al:LIPIcs:2017:7468}.
Even more, there are dynamic algorithms for
\textsc{Knapsack} that can maintain $(1+\epsilon)$-approximate
solutions in polylogarithmic update time whenever an item is inserted to or removed
from the input~\cite{eberle_et_al:LIPIcs.FSTTCS.2021.18,henzinger_et_al:LIPIcs.STACS.2023.36}.
It is an important goal in algorithm design to determine
the best possible running time to solve (or approximate) a problem. Besides \textsc{Knapsack}, there are
also many other problems for which we have (almost) matching upper and lower bounds, e.g.,
computing the Fréchet distance~\cite{BringmannFrechet}, Least Common Subsequence~\cite{AbboudLCS,bringmann2015quadratic}, Negative Triangles~\cite{williams2010subcubic}, or
Graph Diameter~\cite{roditty2013fast}.
%
% see https://www.mpi-inf.mpg.de/fileadmin/inf/d1/teaching/summer16/polycomp/polycomp13.pdf

% \todo{Sollten wir hier betonen, dass $d$ konstant ist?}
A natural generalization of \textsc{Knapsack} is the \emph{$d$-dimensional
	geometric knapsack }problem in which the items are geometric objects
like squares, rectangles, or hypercubes. Like in \textsc{Knapsack,}
we want to select a subset of the given items, but now we require
in addition that they are placed non-overlappingly inside the knapsack,
which we assume to be a square or a (hyper-)cube. The problem is motivated
by many practical applications like placing advertisements on a board
or a website, cutting pieces out of raw material like wood or metal,
or loading cargo into a ship or a truck. Formally, we assume that
we are given an integer $N$ such that our knapsack is an axis-parallel
square (if $d=2$) or a (hyper-)cube (if $d\ge3$) where each edge
has length $N$. Also, we are given a set of items $\I$ where each
item $i\in\I$ is a $d$-dimensional (hyper-)cube or a $d$-dimensional
(hyper-)cuboid with axis-parallel edges and a given profit.

Unfortunately, our understanding of the \emph{$d$-dimensional geometric
	knapsack }problem falls short in comparison to our understanding of \textsc{Knapsack}.
There is a polynomial time $(1+\epsilon)$-approximation algorithm
for each $\epsilon>0$ if all input items are (hyper-)cubes, due to Jansen, Khan, Lira, and Sreenivas~\cite{jansen2022ptas}.
In the running time of this algorithm, the exponent
of $n$ depends exponentially on $d$ and $1/\epsilon$.
However, there is no (conditional) lower bound known that justifies this. From
all we know, it might still be possible to obtain a better running
time of the form $f(\epsilon,d)n^{O(1)}$, for which the exponent
of $n$ depends neither on $\epsilon$ nor on $d$, but it is only
a small constant like 2 or even~1. Note that there are problems for
which we know conditional running time lower bounds that rule this
out, e.g., lower bounds of $\Omega(n^{2})$ or $\Omega(n^{3})$, based
on assumptions like the (Strong) Exponential Time Hypothesis or the
3-SUM conjecture (see, e.g., \cite{AbboudLCS,AbboundBringmannFischer2023,bringman2018multivariate, BringmannFrechet,bringmann2015quadratic} and references therein).
For example, for the Graph Diameter problem there is a lower bound of $\Omega(m^{2})$, with $m$ being the number of edges of the given graph, for computing
a better approximation ratio than $3/2$~\cite{roditty2013fast}.
However, no such lower bounds are known for $d$-dimensional geometric knapsack.

For the special case of squares, i.e., $d=2$, there is a $(1+\epsilon)$-approximation
known with a running time of the form $f(\epsilon)n^{O(1)}$ due to Heydrich and Wiese~\cite{heydrich2019faster}.
However, even in this result the running time is much slower than
linear time since the algorithm uses an initial guessing step with
$\Omega(n)$ options and for each of these option solves several linear
programs of size $\Omega(n)$ each. Furthermore, there is no dynamic
algorithm known for \emph{$d$}-dimensional geometric knapsack, not
even for the special case of squares (i.e., if $d=2$).

If we allow more general shapes than squares, cubes, and hypercubes,
we understand the problem even less. For two-dimensional axis-parallel
rectangles, the best known polynomial time approximation algorithm is
due to Gálvez, Grandoni, Ingala, Heydrich, Khan, and Wiese~\cite{galvez2021approximating}, having an approximation
ratio of~$1.89+\epsilon$. If it is allowed to rotate rectangles by 90 degrees,
then a $(1.5+\epsilon)$-approximation algorithm is known~\cite{galvez2021approximating}. The problem
is not known to be $\mathsf{APX}$-hard, so it may even admit a PTAS.
Also in the mentioned results, the exponent of $n$ in the running
time depends exponentially on $1/\epsilon$. However, we do not know any lower
bound of the needed running time to solve the problem.
Thus, it might well be possible that we can achieve these
or similar results in a running time that is much faster, e.g., $O(n\cdot\mathrm{poly}(\log n))$.
Also, it is open whether a dynamic algorithm exists for the problem.

\subsection{Our contribution}

Our first result is a $(1+\epsilon)$-approximation algorithm for\emph{
}the\emph{ $d$}-dimensional geometric knapsack problem for squares,
cubes, and hypercubes with a running time that is near-linear, i.e.,
$O(n\cdot\mathrm{poly}(\log n))$ for any constant $d$ and $\epsilon>0$, where the exponent of $\log n$ depends on exponentially on $d$ as well as $1/\epsilon$. In particular, this drastically improves the exponent of $n$ in the
running time in~\cite{jansen2022ptas} from a value that is exponential in the dimension
$d$ and $1/\epsilon$ to only 1, which is in particular completely
independent of $d$ and $1/\epsilon$. This even implies that our algorithm is an \emph{efficient} polynomial time approximation scheme (EPTAS) \footnote{Note that
there exists a function $f$ such that
for any $n$ and $k$ we have that
		$(\log n)^k \le f(k) \cdot n^{O(1)}$.}. Thus, up to polylogarithmic
factors, we obtain the fastest possible running time of a PTAS for
any fixed $d$ and $\epsilon$. Note that, for constant $d$, this yields a distinction
to problems for which there are lower bounds of $\Omega(n^{2})$
or $\Omega(n^{3})$ based on (S)ETH or other hypotheses, e.g., \cite{AbboudLCS,AbboundBringmannFischer2023,bringman2018multivariate, BringmannFrechet,bringmann2015quadratic, roditty2013fast}.

For the case of rectangles, we present a $(2+\epsilon)$-approximation
algorithm with a running time of $O(n\cdot\mathrm{poly}(\log n))$
for any constant $\epsilon>0$. If it is allowed to rotate the rectangles
by 90 degrees, we obtain an approximation ratio of $\frac{17}{9}+\epsilon$
with the same running time bound. Thus, our algorithms are much faster
than the best known polynomial time algorithms for the problem~\cite{galvez2021approximating};
in their running times, the exponent of $n$ depends exponentially
on $1/\epsilon$ while our exponent is only 1. Although we need
much less running time, our approximation factors are not much higher than their
ratios of $1.89+\epsilon$ and $1.5+\epsilon$, respectively.

Moreover, we present the first dynamic algorithms for $d$-dimensional
geometric knapsack with hypercubes and  rectangles with and without rotations by 90 degrees. These algorithms maintain solutions with the same approximation guarantees as stated above,
with poly-logarithmic worst-case query and update times. In comparison, note that there are problems
% like Single-Source-Reachability (SSR)
for which there are polynomial conditional lower bounds for
the update and query time for dynamic algorithms~\cite{henzinger2015unifying}.
%
% . To achieve these algorithms, we use data structures that allow for poly-logarithmic update times, i.e., an item can be inserted or removed from the instance in time $O(\mathrm{poly}(\log n))$. We then use the ideas behind the algorithmic results describe above to allow for querying an approximation of the optimal solution value and whether an item is contained in a (fixed) solution with the corresponding approximation factor. Both of these operations can also be executed in time $O(\mathrm{poly}(\log n))$. Note that in comparison, there are
We remark that our algorithms maintain \emph{implicit} solutions in the sense that after each update operation, the answers to all query operations refer to the same fixed approximate solution. Note that after adding or removing a single item (e.g., a very large but very profitable item), it can be necessary to change $\Omega(n)$ items in the current solution in order to maintain a bounded approximation guarantee. Therefore, it is impossible to maintain explicit solutions with our update and query times.
%
% they only indicate how many items of a certain type (depending on the item dimensions and profits) are packed rather than the correct identity of the packed items. This is necessary since an explicit solution may change a lot when inserting or removing an item and, therefore, changing this solution may take time $\Omega(n)$.
%for $d\ge2$; our algorithms maintain solutions with the same
%approximation guarantees as above.
%: $1+\epsilon$ for $d$-dimensional
%squares, cubes, and hypercubes, $2+\epsilon$ for rectangles without
%rotations, and $1.???+\epsilon$ for rectangles with rotations by
%90 degrees.
%In these algorithms, an update operation, i.e., inserting
%or deleting an item, takes time $O(\mathrm{poly}(\log n))$ for any
%constant $d$ and $\epsilon>0$ in the worst-case. Other operations
%are querying an approximation of the optimal solution value and whether
%a certain item is contained in a (fixed) solution $\ALG$ with the
%corresponding approximation factor. Each of these operations takes
%also time $O(\mathrm{poly}(\log n))$ in the worst case. The answers
%to the
%query operations
%% latter queries
%are consistent between two consecutive update
%operations, i.e., they correspond to the same solution $\ALG$. Furthermore,
%we can output $\ALG$ completely in time $O(|\ALG|\mathrm{poly}(\log n))$.
%Note that in comparison, there are problems
%% like Single-Source-Reachability (SSR)
%for which there are polynomial conditional lower bounds for
%the update and query time of any dynamic algorithm~\cite{henzinger2015unifying}.
\subsection{Techniques}
The known algorithms for the\emph{ $d$}-dimensional geometric knapsack
problem for squares, cubes, hypercubes, and rectangles are based on
the existence of structured packings into $O_{\epsilon,d}(1)$ boxes
(i.e., a constant number of boxes for each fixed $\epsilon$ and $d$).
In these algorithms, one first guesses these boxes (i.e., enumerates
all possibilities) which already yields a running time bound of
$n^{O_{\epsilon,d}(1)}$. It is not clear how to make use of these
structured packings without guessing the boxes first.

Instead, we use a different type of structured packings which we call \emph{easily guessable packings}. They are
% In order to overcome this obstacle, we pinpoint the properties needed to be able to guess the boxes in time $O(\mathrm{poly}(\log n))$ and show the existence of a new type of structured packings which offer us a balance between properties that are
\begin{enumerate}[(i)]
	\item flexible enough so that they allow for very profitable solutions, and
	\item structured enough so that we can compute these solutions very fast.
\end{enumerate}
In these packings, each box is specified by some parameters (e.g. height and width) and we can guess \emph{all but at most one} parameter for each box in time $O(\mathrm{poly}(\log n))$. In contrast, in the previous
results, $n^{\Omega(1)}$ time is needed already for one single parameter.
However, for each box, there may still be one parameter whose value
we have not guessed yet. To determine them, we use an important property of our
easily guessable packings. There is a partition of the input items
and a partition of the boxes for which we have not yet guessed all parameters. For each resulting subset of items,
all items of this set can be assigned only to boxes in one specific set of boxes in the partition.
Moreover, all boxes in the latter set have a certain (identical) value for the remaining
(not yet guessed) parameter.
%
% such that items from each set of the partition can be only assigned
% to boxes which have a certain (identical) value for the remaining
% (not yet guessed) parameter.    %\todo{Partition into $O_{\epsilon,d}(1)$ sets?}
This allows us to adapt the indirect guessing
framework from~\cite{heydrich2019faster} to guess the remaining parameter approximately for each box
step by step, while losing only a factor of $1+\epsilon$ in our profit, compared
to guessing it exactly in time $n^{\Omega(1)}$.

Note that for the case of rectangles without rotations there is a threshold of 2 since there are no structured packings with a better ratio using only constantly many
boxes~\cite{galvez2021approximating}. Thus, it is unclear how to improve our the approximation factor for this case to, e.g., $2-\delta$ for~$\delta>0$.

%Therefore, we present a new type of structured packings
%which we can use \emph{without} guessing
%% for which
%% it is \emph{not} necessary to guess
%all boxes at the very
%beginning in time $n^{\Omega_{\epsilon,d}(1)}$. We call them \emph{easily guessable packings}. In these
%packings, each box is specified by some parameters (e.g., height and
%width) and we can guess \emph{all but at most one} parameter easily
%in time $O(\mathrm{poly}(\log n))$. In contrast, in the previous
%results, $n^{\Omega(1)}$ time is needed already for one single parameter.
%However, for each box, there may still be one parameter whose value
%we have not guessed yet. Here we use an important property of our
%easily guessable packings: there is a partition of the input items
%such that items from each set of the partition can be only assigned
%to boxes which have a certain (identical) value for the remaining
%(not yet guessed) parameter.    %\todo{Partition into $O_{\epsilon,d}(1)$ sets?}
%This allows us to use the indirect guessing
%framework from~\cite{heydrich2019faster} to guess the remaining parameter approximately for each box
%step by step, while losing only a factor of $1+\epsilon$ compared
%to guessing it exactly in time $n^{\Omega(1)}$.
%
%In particular, in our easily guessable packings we identify a balance between properties
%that are (i) strong enough such that we can obtain algorithms with near-linear running times, and that are (ii)
%flexible enough to
%yield
%% allow profitable packings whose
%approximation ratios identical or similar to those of the
%best known polynomial time algorithms.

Please note that due to space limitations, most proofs were moved to the appendix.

% \aw{crucial} properties of structured packings that
% are sufficient to obtain algorithms with near-linear running times.
% Also, we show that profitable packings exist with these additional
% properties. In the case of hypercubes, we obtained even exactly the
% same approximation ratio as the previously known result, and for rectangles,
% we obtain similar approximation ratios as the best known polynomial
% time algorithms.

\subsection{Other related work}

Prior to the results for two-dimensional geometric knapsack for rectangles
mentioned above, a polynomial time $(2+\epsilon)$-approximation algorithm
was presented by Jansen and Zhang~\cite{jansen2004rectangle} in which the
exponent of $n$ in the running time is exponential in $1/\epsilon$.
For the special case of unweighted rectangles, the same authors gave
a faster $(2+\epsilon)$-approximation algorithm with a running time
of $O(n^{1/\epsilon+1})$~\cite{jansen2004maximizing}. If
one allows pseudo-polynomial running time instead of polynomial running
time, there is also a $(4/3+\epsilon)$-approximation algorithm in the setting of weighted rectangles known
due to Gálvez, Grandoni, Khan, Ramírez-Romero, and Wiese~\cite{grandoni2021improved}, having a running
time of $(nN)^{O_{\epsilon}(1)}$. Also here, the exponent of $nN$
is exponential in $1/\epsilon$. In addition, there is a $(1+\epsilon)$-approximation
algorithm by Adamaszek and Wiese~\cite{adamaszek2014quasi} with
quasi-polynomial running time for any constant $\epsilon>0$, assuming
quasi-polynomially bounded input data. If the input objects are triangles
that can be freely rotated, there is a polynomial time $O(1)$-approximation
algorithm due to Merino and Wiese~\cite{merino_et_al2020}.

Another way to generalize \textsc{Knapsack} is to allow several knapsacks
into which the items can be packed, possibly with different capacities.
This generalization still admits $(1+\epsilon)$-approximation algorithms
with a running time of $n^{O_{\epsilon}(1)}$~\cite{chekuri2005polynomial,kellerer1999polynomial},
and even with a running time of the form $f(\epsilon)n^{O(1)}$ for
some function $f$~\cite{jansen2010parameterized,jansen2012fast}.
Furthermore, there is a dynamic algorithm known for the problem
with polylogarithmic update time which also achieves
an approximation ratio of $1+\epsilon$~\cite{eberle_et_al:LIPIcs.FSTTCS.2021.18}.
% Note that the number of knapsacks is not assumed to be a constant
% in these results.

\section{\label{sec:hypercubes}Algorithms for $d$-dimensional hypercubes}

In this section we present our PTAS and dynamic algorithm for the geometric knapsack problem with $d$-dimensional hypercubes. Let $\epsilon>0$ and assume w.l.o.g.~that $1/\epsilon\in\mathbb{N}$
and $\epsilon<1/2^{d+2}$; we assume that $d\in \mathbb{N}$ is a constant. We are given a set of $n$ items $\I$ where each item $i \in \I$ is a hypercube with
side length $s_i \in \mathbb{N}$ in each dimension and profit $p_i \in \mathbb{N}$, and an integer $N$ such that the knapsack has side length $N$ (in
each dimension). In the following, we present a simplified version of our algorithm with a running time of $O(n\cdot\mathrm{poly}(\log N))$. In Appendix~\ref{app:cubes} we describe how to improve
its running time to $O(n\cdot\mathrm{poly}(\log n))$ using more involved technical ideas.

In our algorithms, we use a special data structure to store our items~$\I$. This data structure will allow us to quickly update the input and obtain important
information about the items in $\I$. It is defined in the following lemma, which can be proven using standard data structures for range counting/reporting for points in two dimensions (corresponding to side length and profit of each item)~\cite{lee1984computational}. The details can be found in Appendix~\ref{app:data_struc}.

%First, we describe a data structure in which we store
%our items $\I$. The data structure will allow us later to compute
%our solution quickly. Also, it can be updated quickly if an item is
%inserted into $\I$ or deleted from $\I$, which will be important
%later when we describe our dynamic algorithm. The data structure can
%be implemented via data structures for range counting/reporting
%for points in two dimensions (corresponding to side length and profit of each input item)~\cite{lee1984computational}. % and allows the following set of operations.
\begin{lemma}%[\aw{Implied by}~\cite{nekrich2009orthogonal}]
	\label{lem:data-structure}There is a data structure for the items
	$\I$ that allows the following operations:
	\begin{itemize}
		\item Insertion and deletion of an item in time $O(\log^2 n)$.
		\item Given four values $a,b,c,d \in \mathbb{N}$, return the cardinality of the set $\I(a,b,c,d):= \{i \in I: a\leq s_i \leq b, c\leq p_i \leq d\}$ in time $O(\log n)$.
		\item Given four values $a,b,c,d \in \mathbb{N}$,  return the set $\I(a,b,c,d):= \{i \in I: a\leq s_i \leq b, c\leq p_i \leq d\}$ time $O(\log n +|\I(a,b,c,d)|)$. 
		%\item Given four values $a,b,c,d \in \mathbb{N}$, the total profit of all items in the set $\I(a,b,c,d):= \{i \in I: a\leq s_i \leq b, c\leq p_i \leq d\}$ can be reported in time $O(\log n)$.
		%\item Given three values $m, t_1, t_2 \in \mathbb{N}$, the $m$ most profitable squares with $s_i \in (t_1,t_2]$ can be reported in time $O(\log^3 n)$.
		%\item Given two values $p\in\mathbb{R}$ and $j \in \mathbb{N}$, the $j$-th smallest item with profit $p$ can be reported in time $O(\log^3 n)$.
	\end{itemize}
\end{lemma}

In the following, we will refer to this data structure as our \emph{item
	data structure}. Observe that we can insert all given items $\I$
into it in time $O(n\log^2 n)$.
Additionally, we use balanced binary search trees (see e.g.~\cite{guibas1978dichromatic}) to store the set of item side lengths and profits; an item can be inserted, deleted, and queried in time $O(\log n)$ in each of these search trees.% allowing for $O(\log n)$ insertion, deletion and lookups.}

%
% \moritz{Furthermore, we use balanced binary search trees (see e.g.~\cite{guibas1978dichromatic}) to store the set of side lengths allowing for $O(\log n)$ insertion, deletion and lookups.}\awr{can we include this directly in the item data structure?}

\subsection{Easily guessable packing of hypercubes}\label{sec:hypercubes_strucpack}

Our algorithm is based on a structured packing into $O_{\epsilon,d}(1)$ boxes, i.e., $d$-dimensional hypercuboids, such that the profit of
the packed items is at least $(1-O(\epsilon))\OPT$, where $\OPT$ denotes the optimal solution of the given instance. These boxes are packed non-overlappingly inside the knapsack and each item is packed into one of these boxes (see Figure~\ref{fig:packboxes}). Intuitively, in our algorithm we will guess the sizes of these boxes and compute an assignment of items into the boxes. Our boxes are designed such that we can make these guesses quickly and such that it is easy to place all items that were assigned to each box.
% and their placement in the knapsack. Then, we will
% we wish to find this set of boxes and an assignment of items to boxes. We then use the fact that the considered boxes can be packed easily once the set of items to be packed into each box is fixed. \mor{figure needs some adjustments}
\begin{figure}[h!]
	\begin{minipage}{.45\textwidth}
		\centering
		\includegraphics[scale = 0.63, page = 9]{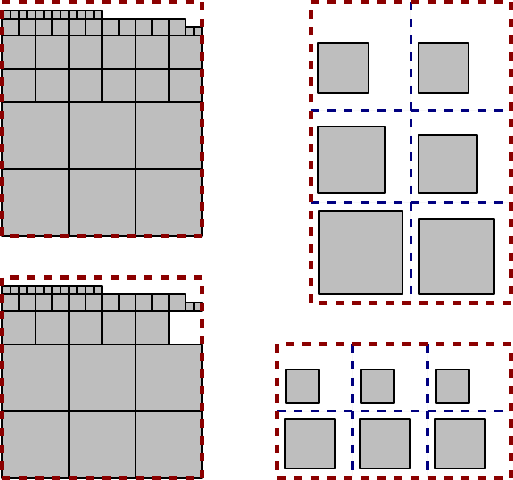}
		\subcaption{Packing of items into single boxes}
		\label{fig:packboxes_boxes}
	\end{minipage}%
	\begin{minipage}{.45\textwidth}
		\centering
		\includegraphics[width=0.8\linewidth,page = 7]{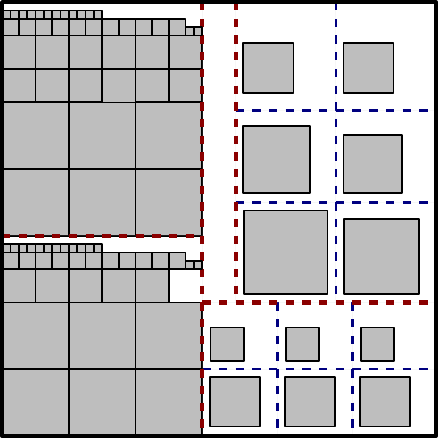}
		\subcaption{Packing of boxes into knapsack} 
		\label{fig:packboxes_full}
	\end{minipage}
	\caption{Visualization of packing using boxes}
	\label{fig:packboxes}
\end{figure}

In~\cite{jansen2022ptas}, a structured
packing was presented that lead to a PTAS with a running time of $n^{O_{\epsilon,d}(1)}$ for our problem. They use two specific types of boxes. In the formal definition of those, we need the volume and surface of a $d$-dimensional box $B$ which we define as $\mathrm{VOL}_d(B) := \prod_{d'=1}^d \ell_{d'}(B)$ and $\mathrm{SURF}_d(B) := 2\sum_{d'=1}^d\mathrm{VOL}_d(B)/\ell_{d'}(B)$, respectively.
\begin{definition}[$\mathcal{V}$-boxes and $\mathcal{N}$-boxes~\cite{jansen2022ptas}]
	\label{def:jansen_boxes}
	Let $B$ be a $d$-dimensional hypercuboid, $\I$ be a set of items packed in it, and let $\hat{s}$ be an upper bound on the side lengths of the items in $\I$. We say that $B$ is
	\begin{itemize}
	 \item a $\mathcal{V}$-box if its side lengths can be written as $n_1\hat{s},n_2\hat{s},\dots,n_d\hat{s}$ where $n_1,n_2,\dots,n_d \in \mathbb{N}_{+}$ and the volume of the packed items is at most $\mathrm{VOL}_d(B)-\hat{s}\frac{\mathrm{SURF}_d(B)}{2}$,
	 \item an $\mathcal{N}$-box if its side lengths can be written as $n_1\hat{s},n_2\hat{s},\dots,n_d\hat{s}$ where $n_1,n_2,\dots,n_d \in \mathbb{N}_{+}$ and the number of packed items is at most $\prod_{i=1}^d n_i$.
	\end{itemize}
\end{definition}

These boxes and the structured packing using them, however, is not enough for our purposes since each parameter must be guessed. It is unclear how to do this faster than
in time~$n^{\Omega(1)}$, already for one single parameter. Therefore, we refine the packing presented in~\cite{jansen2022ptas} by adapting the type of boxes to include additional technical properties.

%each box is either an
%$\N$-box \footnote{The asterisk in the term ``$\N$-box'' refers to the fact that we require an additional
%	property compared to the $\mathcal{N}$-boxes in~\cite{jansen2022ptas}.}
%or an $\S$-box.

%However, for our purposes we need certain additional
%technical properties. Therefore, we refine the packing presented in~\cite{jansen2022ptas}. 
Our first type of boxes are $\N$-boxes which are a refinement of the aforementioned $\mathcal{N}$-boxes. For each $\N$-box $B$ we specify two additional parameters $s_{\min}(B)$ and $s_{\max}(B)$ and require for each item $\in \I$ packed in $B$ that $s_{\min}(B)\le s_i \le  s_{\max}(B)$. These values $s_{\min}(B)$ and $s_{\max}(B)$ will help us in our computation later. Intuitively, the box $B$ is partitioned into a $d$-dimensional grid with spacing $s_{\max}$, such that each item in $B$ is placed inside one of these grid cells (similar as the $\N$-boxes), see Figure~\ref{fig:origpack_Nmain}.
%
% value $s_{\max}(B)$ plays a similar role as the
%
% It has two parameters
%
% Like the $\mathcal{N}$-boxes, an $\N$-box $B$ is
%
% as we also require a lower bound on the side lengths of items packed into such a box. Observe that this property may seem weak but as we will later show it is necessary for our purposes of reducing the running time of the guessing step of our algorithm to $O(\mathrm{poly}(\log n))$.
% %\footnote{The asterisk in the term ``$\N$-box'' refers to the fact that we require an additional property compared to the $\mathcal{N}$-boxes in~\cite{jansen2022ptas}.}
% Intuitively, an $\N$-box is a $d$-dimensional hypercuboid  $B$ that is
% partitioned into a $d$-dimensional grid with spacing $s_{\max}$,
% where $s_{\max}$ is an upper bound for the side length of the largest
% item that is placed inside the box. Moreover, the number of items packed into $B$ is at most the number of cells of the grid described above. Hence, packing these items into $B$ is easily done by placing one item in each cell.
\begin{definition}[$\N$-box]
Let $B$ be a $d$-dimensional hypercuboid with given values $s_{\min}(B)\ge~0$ and
$s_{\max}(B)\ge~0$
and suppose we are given
a packing of items $\I'\subseteq\I$ inside $B$.
% Let $s_{\max}(B)$
% and $s_{\min}(B)$ be such that $s_{\max}(B)\ge\max_{i\in\I'}s_{i}$
% and $s_{\min}(B)\le\min_{i\in\I'}s_{i}$.
We say that $B$ is an $\N$\emph{-box}
if $s_{\min}(B)\le\min_{i\in\I'}s_{i}$ and
$s_{\max}(B)\ge\max_{i\in\I'}s_{i}$
and
if its side lengths can be written as $n_{1}(B)s_{\max}(B),\dots,n_{d}(B)s_{\max}(B)$
with $n_{1}(B),\dots,n_{d}(B)\in\mathbb{N}$ such that~$|\I'|\le\prod_{i=1}^{d}n_{i}(B)$.
\end{definition}
Our second type of boxes are $\S$-boxes. Intuitively, an $\S$-box $B$ contains only items that are very small in each dimension compared to $B$. For technical reasons, we require that the packed items do not use the full volume of $B$, even
if we increase each item size to the next larger power of $1+\epsilon$. This will allow us to compute our solution efficiently since we may
round up item sizes and profits to powers of $1+\epsilon$.

Throughout this paper, for any  $x > 0$ we define $\lceil x\rceil_{1+\epsilon}$
% (and $\lfloor x\rfloor_{1+\epsilon}$)
to be the smallest power of $1+\epsilon$ that is larger
than~$x$ rounded down. Similarly, we define $\lfloor x\rfloor_{1+\epsilon}$ to be the largest power of $1+\epsilon$ that is smaller than $x$ rounded up.
% See Figure~\ref{fig:origpack_Smain} for a visualization of an $\S$-box in two dimensions which is packed using NFDH.

%Intuitively, a $\S$-box $B$ contains only items that are very small
%in each dimension compared to $B$. For technical reasons, we require
%that the the packed items do not use the full volume of $B$, even
%if we round up the item sizes to the next larger power of $1+\epsilon$.
%This will allow us later to select more efficiently the items that
%we pack into our $\S$-boxes, e.g., we can treat items identically
%if their sizes are almost the same. Throughout this paper,
%for any  $x > 0$ we define $\lceil x\rceil_{1+\epsilon}$
%% (and $\lfloor x\rfloor_{1+\epsilon}$)
%to be the smallest power of $1+\epsilon$ that is larger
%than~$x$. Similarly, we define $\lfloor x\rfloor_{1+\epsilon}$ to be the largest power of $1+\epsilon$ that is smaller than $x$.

\begin{definition}[$\S$-boxes]
Let $B$ be a $d$-dimensional hypercuboid with side length $\ell_{d'}(B)\in\mathbb{N}_{0}$
for each $d'\in[d]$ and suppose we are given a packing of items $\I'\subseteq\I$
inside $B$. We say that $B$ is a $\S$\emph{-box} if for each item
$i\in\I'$ we have $s_{i}\leq\epsilon\min_{d'\in[d]}\ell_{d'}(B)$
and additionally $\sum_{i\in\I'}\lceil s_{i}\rceil_{1+\epsilon}^{d}\leq(1-2d\cdot\epsilon)\mathrm{VOL}_d(B)$.
\end{definition}
\begin{figure}
	\begin{minipage}{.4\textwidth}
		\centering
		\includegraphics[width=0.9\linewidth, page = 1]{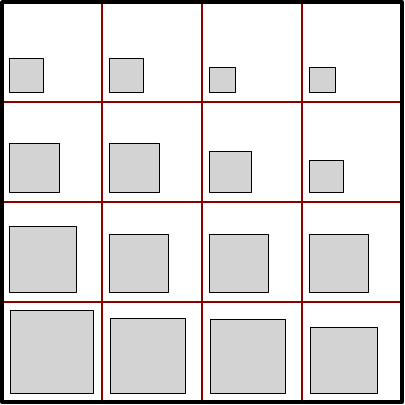}
		\subcaption{Sorted packing of an $\N$-box: each item is packed into a single grid cell.}
		\label{fig:origpack_Nmain}
	\end{minipage}%
	\hspace*{1cm}
	\begin{minipage}{.4\textwidth}
		\centering
		\includegraphics[width=0.9\linewidth,page = 2]{figures/Squares2D.pdf}
		\subcaption{NFDH packing of an $\S$-box: items are packed greedily into shelves.}
		\label{fig:origpack_Smain}
	\end{minipage}
	\caption{Visualization of an $\N$- and an $\S$-box for $d=2$}
	\label{fig:origpackmain}
\end{figure}

%For each $\S$-box $B$ we require that the items in $\I'$ packed inside it do not use the
%full volume of $B$ and that each item in $\I'$ is small compared to $B$. 
For an $\S$-box $B$, we can show that all items in
%\emph{any} set of the items 
$\I'$
% satisfying the above properties
can be packed inside $B$ using the Next-Fit-Decreasing-Height (NFDH) algorithm. NFDH is a greedy algorithm that orders
the items non-increasingly by their sizes and then packs them greedily
in this order into shelves within the box $B$ (see Figure~\ref{fig:origpack_Smain} and
Appendix~\ref{app:cubes} for details).
%\awr{sentence in comments that could go to caption of NFDH-figure}
%In Figure~\ref{fig:origpack_Smain} one can see how the first item packed into each shelf fixes the height of the shelf such that the shelf is then treated as a one dimensional knapsack which is packed greedily.
% We make use of the fact that NFDH always finds a feasible packing of an $\S$-box $B$ with given item set $\I'$.

%Moreover, since we required that the items in $\I'$ do not use the
%full volume of a $\S$-box $B$ and they are small compared to $B$,
%the Next-Fit-Decreasing-Height (NFDH) algorithm always finds a packing
%of such a set $\I'$ into $B$. NFDH is a greedy algorithm that orders
%the items non-increasingly by their sizes and then packs them greedily
%in this order into shelves within the box $B$ (see
%Figure~\ref{fig:origpack_S} and
%Appendix~\ref{app:cubes} for details). \mor{More detail on NFDH?}

%We will need later that this is also true
%if the items have a side length of up to $2\epsilon\ell_{\min}(B)$ (instead of $\epsilon\ell_{\min}(B)$ as in the previous definition)
%% if the items are slightly larger than required in the previous definition
%and their total rounded volume is slightly larger as well.
\begin{lemma}[implied by Lemma 4 in~\cite{harren2009approximation}]
\label{lem:NFDH}Let $B$ be a box and let $\I'$ be a set of items
such that $s_{i}\leq2\epsilon\ell_{\min}(B)$ for each $i\in\I'$
and $\sum_{i\in\I'}\lceil s_{i}\rceil_{1+\epsilon}^{d}\leq\mathrm{VOL}_{d}(B)-2(d-1)\cdot\epsilon\mathrm{VOL}_{d}(B)$.
%\awr{notation $\lceil s_{i}\rceil_{1+\epsilon}^{d}$ introduced?}
Then, NFDH finds a packing of $\I'$ into $B$.
%\marginpar{Note that this is a slightly different condition than the definition
	%before!}
\end{lemma}
We now prove that there is always a $(1+O(\epsilon))$-approximate \emph{easily guessable packing} using $O_{\epsilon,d}(1)$ boxes, each of them being an $\N$- or an $\S$-box. In order to do so we refine the structured packing in~\cite{jansen2022ptas} which uses $\mathcal{N}$- and $\V$-boxes.
The key additional property is that there are $O_\epsilon(1)$ values $k_{1},k_{2},\dots,k_{r}$ such that each $\N$-box $B$ contains
only items whose respective sizes are within the interval $(k_{j-1},k_{j}]$
for some $j$. In particular, we have that $s_{\min}(B)=k_{j-1}$ and $s_{\max}(B)=k_{j}$ and the values $k_{1},k_{2},\dots,k_{r}$ yield a partition
of the items in $\I$ of the form $(k_{j-1},k_{j}]$. This will allows us to apply an indirect guessing framework later to guess the values $k_{1},k_{2},\dots,k_{r}$ quickly.
All remaining parameters of the boxes  can be easily guessed in time $O(\log N)$ each (and there are only $O_{\epsilon,d}(1)$ parameters in total).
% options for each of them, and there are $O_\epsilon(1)$ such parameters in total.

%
% Note, however, that it is not straight-forward to guess the values $k_{1},k_{2},\dots,k_{r}$.
% On the other hand,
%
%
% The main task is to establish property ii) below and to ensure
% that each $\S$-box $B$ contains only items that are very small compared
% to $B$. Intuitively, property ii) states that there are constantly
% many values $k_{1},k_{2},\dots,k_{r}$ such that each $\N$-box contains
% only items whose respective sizes are within the interval $(k_{j},k_{j+1}]$
% for some $j$. This property allows us to guess all but one parameter for each of the boxes. The remaining parameter will then be computed (approximately) using a modification of the indirect guessing framework~\cite{heydrich2019faster}.
%We refine a structured packing in~\cite{jansen2022ptas}, mainly to establish property ii) below and to ensure
%that each $\S$-box $B$ contains only items that are very small compared
%to $B$. Intuitively, property ii) states that there are constantly
%many values $k_{1},k_{2},\dots,k_{r}$ such that each $\N$-box contains
%only items whose respective sizes are within the interval $(k_{j},k_{j+1}]$
%for some $j$.
\begin{lemma}[Near-optimal packing of hypercubes]
\label{lem:struc_hypercubes} For any instance $\I$ of the $d$-dimensional
hypercube knapsack problem and any $\epsilon<1/2^{d+2}$, there exists
a packing with the following properties: 
\begin{enumerate}[i)]
	\item It consists of $\N$- and $\S$-boxes whose total number is bounded
	by a value $C_{\mathrm{boxes}}(\epsilon,d)$ depending only on $\epsilon$
	and $d$. 
	\item There exist values $k_{1},k_{2},\dots,k_{r}\in\mathbb{Z}_{\geq0}$
	with $r\in O_{\epsilon,d}(1)$ such that if $B$ is an $\N$-box there
	exists a value $j_{B}\in\{1,2,\dots,r\}$ such that $s_{\max}(B)=k_{j_{B}}$
	and $s_{\min}(B)=k_{j_{B}-1}$ with $k_{0}:=0$. 
	\item For each $\S$-box $B$ and each $d' \in [d]$ we have that $\ell_{d'}(B) = \lfloor x\rfloor_{1+\epsilon}$ for some $x\in [N]$.
	\item For each $\N$-box $B$ and each $d' \in [d]$ we have that $n_{d'}(B) = \lfloor n'\rfloor_{1+\epsilon}$ for some $n'\in [n]$ if $n_{d'}(B) > 1/\epsilon$.

	\item The total profit of the packing is at least $(1-2^{O(d)}\epsilon)\OPT$.
	% \todo{check constant here! rephrase $O_{d}(1)\epsilon$?}\awr{how about $2^{O(d)}$ ?},
	% where $\OPT$ is the profit of an optimal packing for instance
	% $\I$.
\end{enumerate}
\end{lemma}
\begin{proof}[Proof sketch.]
We start with the structured packing due to Jansen, Khan, Lira and Sreenivas~\cite{jansen2022ptas} which consists of $O_{\epsilon,d}(1)$ many $\mathcal{N}$- and $\V$-boxes and items with a total profit of at least~$(1-2^{d+2}\epsilon)\OPT$. We modify this packing as follows. First, we consider the $\V$-boxes and split each of them into at most $O_{\epsilon,d}(1)$ many $\N$- and $\S$-boxes. We lose at most a factor of $1+\epsilon$ of our profit in this step.
For each resulting $\N$-box $B$, we define
% Repeating this argument for all $\V$-boxes and defining
% values $s_{\min}(B)$ and
$s_{\max}(B)$ to be the maximum size of an item packed into $B$.
% such that for each item $i$ packed in $B$ we have that $s_{\min}(B) \le s_i \le s_{\max}(B)$.
This yields a packing satisfying property i).
% for each box $B$ which is an original $\N$-box, we construct a packing satisfying property i).
In order to satisfy property ii), we first introduce a value $k_j$ for each distinct value $s_{\max}(B)$ of the $\N$-boxes constructed so far. We order the resulting values $k_1,...,k_r$ increasingly.
Then, we split the $\N$-boxes into smaller boxes such that for each resulting box $B$, the range of item sizes of $B$ is contained in $(k_{j-1},k_{j}]$ for some $j$; hence, we define $s_{\min}(B):=k_{j-1}$. We ensure that the number of boxes increases by at most a factor of $O_{\epsilon,d}(1)$ in this step.

% In all operations above we always ensure that for every original box the profit that remains packed into the resulting new boxes is at least a factor $(1+O(\epsilon))$ of the original profit.
Next, we modify the packing to satisfy properties iii) and iv).
For each $\N$-box $B$ and each dimension $d'\in [d]$, we round $n_{d'}(B)$ down to $\lfloor n_{d'}(B) \rfloor_{1+\epsilon}$. This decreases the maximum number of items we can pack into $B$ by at most a factor of $(1+\epsilon)^d=1+O(\epsilon)$; hence, also our profit reduces by at most this factor. Finally, for each $\S$-box $B$ we round down each side length
$\ell_{d'}(B)$ to $\lfloor \ell_{d'}(B)\rfloor_{1+\epsilon}$.
% to a power of $1+\epsilon$.
Since all items in $B$ are small compared to each side length of~$B$, our profit reduces by at most a factor of $1+\epsilon$.
%
% To this end, observe that for an $\N$-box $B$ rounding  $n_{d'}(B)$ to $\lfloor n_{d'}(B) \rfloor_{1+\epsilon}$ implies that the number of items that we can pack into $B$ decreases by at most a factor of $(1+\epsilon)$. Thus, doing this in all dimensions leads to a packing satisfying property iv) while only losing a factor of  $(1+O(\epsilon))$ of the profit. Finally, we use an LP-argument to show that we can round down the side lengths of every $\S$-box to integer powers of $(1+\epsilon)$ and only lose a factor $(1+O(\epsilon))$ of the profit. This leads to our final packing.
\end{proof}

\subsection{\label{subsec:Computing-packing}Computing a packing}
In this section we discuss how to compute our packing. Intuitively, we try to guess the packing due to Lemma~\ref{lem:struc_hypercubes}. We assume that all input items $\I$ are stored in our item data structure (see to Lemma~\ref{lem:data-structure}) and we discard all items with profit less than $\epsilon \frac{p_{\max}}{n}$ where $p_{\max} :=\max_{i\in \I} p_i$. The total profit of the discarded items is at most $n\cdot \epsilon \frac{p_{\max}}{n} \le \epsilon \OPT$.

Let $\B$ denote the set of boxes due to Lemma~\ref{lem:struc_hypercubes}. In our algorithm, we first guess all parameters of the boxes in $\B$ apart from the values $k_1,\dots,k_r$. We then modify the indirect guessing framework introduced by Heydrich and Wiese~\cite{heydrich2019faster} to approximately compute the values $k_1,\dots,k_r$. Our computed values might be imprecise in the sense that they yield a solution whose profit could be by a factor $1+\epsilon$ lower than the solution due to Lemma~\ref{lem:struc_hypercubes}. However, our resulting boxes are guaranteed to fit inside of the knapsack.

% consists of two main steps where we make use of the nature of our easily guessable packings which allow us to guess all but one parameter for each box. In particular, we start by basic quantities of the boxes $\B$ which only exclude the values $k_1,\dots,k_r$. Having guessed these basic quantities, we then modify the indirect guessing framework introduced by Heydrich and Wiese~\cite{heydrich2019faster} to approximately compute values $k_1,\dots,k_r$. In this step we also find partial packings of items into the guessed boxes. Thus, the final packing is obtained by combining these partial packings. We would like to remark again that the computed packings are implicit as they treat items the same if they differ by at most a factor of $(1+\epsilon)$ in both profit and size. This is necessary for our running time bounds. These implicit packings, however, can be used to find explicit packings.

%Assume that all input items $\I$ are stored in our item data structure
%(see to Lemma~\ref{lem:data-structure}).
%Let $\B$ denote the set of boxes due to Lemma~\ref{lem:struc_hypercubes}.
%First, we guess certain quantities of the boxes in $\B$
%% the packing due to Lemma~\ref{lem:struc_hypercubes}
%i.e., we try all combinations of all possibilities for these quantities.
%For each box $B\in\B$ denote
%by $\I(B)$ the items packed into $B$. In the remainder, we discard all items with profit less than $\epsilon \frac{p_{\max}}{n}$, where $p_{\max} :=\max_{i\in \I} p_i$,
%losing a total profit of at most $\epsilon \OPT$.

\subsubsection{Guessing basic quantities}
We start by guessing the number of $\N$- and $\S$-boxes, respectively. Then, for each $\N$-box $B\in\B$ we guess
\begin{itemize}
	\item the value $j_{B}$ indicating that $s_{\max}(B)=k_{j_{B}}$ and $s_{\min}(B)=k_{j_{B}-1}$; note that for $j_{B}$ there are only $O_{\epsilon,d}(1)$ possibilities,
	\item the value $n_{d'}(B)$ for each dimension $d' \in [d]$; for each value $n_{d'}(B)$ there are only $O(\log n)$ possibilities.
\end{itemize}
For each $\S$-box $B\in\B$, we guess $\ell_{d'}(B)$ for each $d'\in[d]$. Note that also here, for each of these values there are only $O(\log N)$ possibilities.
We denote by the \emph{basic quantities} all these guessed values for all boxes in $\B$.

\begin{lemma}\label{lem:hypcub_guessing_basic}
	In time $(\log N)^{O_{\epsilon,d}(1)}$ we can guess all basic quantities.
% 	\begin{itemize}
% 	 \item for each $\N$-box $B$ the value $j_{B}$ and the value $n_{d'}(B)$ for dimension $d' \in [d]$ and
% 	 \item for each $\S$-box $B$  the value $n_{d'}(B)$ for dimension $d' \in [d]$.
% 	\end{itemize}
\end{lemma}

We do not know the values $k_{1},k_{2},\dots,k_{r}$. However, recall that they yield a partition of $\I$ with a set $\I_{j}:=\{i\in\I:s_{i}\in(k_{j-1},k_{j}]\}$
for each $j\in[r]$. We next guess additional quantities which provide us with helpful information for the next step of our algorithm, the indirect guessing framework. First, we guess approximately the profit that each set $\I_{j}$ contributes
to $\OPT$. Formally, for each $j\in[r]$ we guess $\hat{p}(j):=\left\lfloor p(\I_{j}\cap\OPT)\right\rfloor _{1+\epsilon}$
if $p(\I_{j}\cap\OPT)\ge\frac{\epsilon}{r}\OPT$ and $\hat{p}(j):=0$
otherwise.  Since $\OPT\in[p_{\max},n\cdot p_{\max})$, we have that $\hat{p}_{j}\in\{0\}\cup[\frac{\epsilon}{r}p_{\max},n\cdot p_{\max})$; hence, there are
only $O(\log N)$ possibilities for $\hat{p}_{j}$.

\begin{lemma}\label{lem:hypcub_guessing_pjs}
	The value $\hat{p}_j$ for each $j\in[r]$ can be guessed in time $(\log n)^{O_{\epsilon,d}(1)}$ and they satisfy $\sum_{j=1}^{r}\hat{p}(j)\ge(1-O(\epsilon))\OPT$.
\end{lemma}

%where for convenience we define $k_{0}:=0$. We
%guess approximately the profit that each set $\I_{j}$ contributes
%to $\OPT$. Formally, for each $j\in[r]$ we guess $\hat{p}(j):=\left\lfloor p(\I_{j}\cap\OPT)\right\rfloor _{1+\epsilon}$
%if $p(\I_{j}\cap\OPT)\ge\frac{\epsilon}{r}\OPT$ and $\hat{p}(j):=0$
%otherwise. Since $\OPT\in[p_{\max},n\cdot p_{\max})$, we have that $\hat{p}_{j}\in\{0\}\cup[\frac{\epsilon}{r}p_{\max},n\cdot p_{\max})$. This implies that for each $\hat{p}(j)$ there are at most $O_{\epsilon,d}(\log n)$ possibilities.
%% Observe that for $\hat{p}(j)$ there are at most $O_{\epsilon,d}(\log n)$
%%possibilities since $\OPT\in[p_{\max},n\cdot p_{\max})$ and hence
%%$\hat{p}_{j}\in\{0\}\cup[\frac{\epsilon}{r}p_{\max},n\cdot p_{\max})$.
%Also, one can show that $\sum_{j=1}^{r}\hat{p}(j)\ge(1-O(\epsilon))\OPT$.

Observe that each $\N$-box $B\in\B$ can contain only items from
$\I_{j_{B}}$. However, each $\S$-box $B\in\B$ might contain items
from more than one set $\I_{j}$. For each $\S$-box $B\in\B$ and
each set $\I_{j}$, we guess approximately the fraction of the volume
$B$ that is occupied by items from $\I_{j}$. Formally, for each
such pair we define the value $a_{B,j}:=\frac{\sum_{i\in\I(B)\cap\I_{j}}\lceil s_{i}\rceil_{1+\epsilon}^{d}}{\mathrm{VOL}_d(B)}$
and guess the value $\hat{a}_{B,j}:=\left\lceil \frac{a_{B,j}}{\epsilon/r}\right\rceil \epsilon/r$,
i.e., the value $a_{B,j}$ rounded up to the next larger integral
multiple of $\epsilon/r$. Note that for each value $\hat{a}_{B,j}$
there are only $O_{\epsilon,d}(1)$ possibilities, and that there
are only $O_{\epsilon,d}(1)$ such values that we need to guess.
\begin{lemma}\label{lem:hypcub_guessing_abjs}
	For each $\S$-box $B \in \B$ and each $j\in[r]$ the value $\hat{a}_{B,j}$
	can be guessed in time~$(\log n)^{O_{\epsilon,d}(1)}$. Moreover, for each $\S$-box $B \in \B$ we have that $\sum_{j=1}^r \hat{a}_{B,j} \mathrm{VOL}_d(B) \leq (1-2d\cdot\epsilon)\mathrm{VOL}_d(B) + \epsilon \mathrm{VOL}_d(B)$.
\end{lemma}

%We remark that for the correct guesses
%we have that $\sum_{j=1}^r \hat{a}_{B,j} \mathrm{VOL}_d(B) \leq (1-2d\cdot\epsilon)\mathrm{VOL}_d(B) + \epsilon \mathrm{VOL}_d(B)$
%for each $\S$-box $B$ (the additive term $\epsilon \mathrm{VOL}_d(B)$ stems from rounding the values $a_{B,j}$).
%
% the following will hold for each : $\sum_{j=1}^r \hat{a}_{B,j} \mathrm{VOL}_d(B) \leq (1-2d\cdot\epsilon)\mathrm{VOL}_d(B) + \epsilon \mathrm{VOL}_d(B)$. The second additive term is due to the rounded guesses.}
\subsubsection{Indirect guessing}
The next step of our algorithm is to determine the values $k_{1},k_{2},\dots,k_{r}$. Unfortunately, we cannot guess them directly in polylogarithmic time,
since there are $N$ options for each of them. In contrast to the other guessed quantities, it is not sufficient to allow only powers of $1+\epsilon$ for each of them. If we choose a value for some $k_j$ that is only a little bit too large, then our boxes might not fit into the knapsack anymore, since for each $\N$-box $B$ we have that $s_{\max}(B)=k_{j_B}$ and the side length of $B$ in each dimension $d'$ equals $n_{d'}(B) s_{\max}(B)$. On the other hand, if we define some value $k_j$ only a little bit too small, then we might not be able to assign items with enough profit to $B$, e.g., if $B$ contains only one item which does not fit anymore if we choose $k_j$ too small. Hence, we need a different approach to determine the values $k_{1},k_{2},\dots,k_{r}$. To this end, we modify the \emph{indirect guessing framework} introduced in~\cite{heydrich2019faster} to fit our purposes.

The main idea is to compute values $\tilde{k}_0, \tilde{k}_{1},\tilde{k}_{2},\dots,\tilde{k}_{r}$
that we use instead of the values $k_{0},k_{1},k_{2},\dots,k_{r}$.
They yield a partition of $\I$ into sets $\tilde{\I}_{j}:=\{i\in\I:s_{i}\in(\tilde{k}_{j-1},\tilde{k}_{j}]\}$.
Intuitively, for each $j$ we want to pack items from $\tilde{\I}_{j}$
into the space that is used by items in $\I_{j}$ in the packing from
Lemma~\ref{lem:struc_hypercubes}. We will choose the values $\tilde{k}_{1},\tilde{k}_{2},\dots,\tilde{k}_{r}$ such that in this way, we obtain almost the same profit as the packing of Lemma~\ref{lem:struc_hypercubes}.
% To estimate this profit we will make use of the guessed quantities $\hat{p}(j)$ for each $j \in [r]$.
Furthermore, we ensure that $\tilde{k}_{j}\le k_{j}$ for each $j\in[r]$. This implies that the side lengths of the resulting $\N$-boxes are at most the side lengths of the $\N$-boxes due to Lemma~\ref{lem:struc_hypercubes} and, therefore, the guessed $\N$-boxes can be feasibly packed into the knapsack.
%We will choose the values $\tilde{k}_{1},\tilde{k}_{2},\dots,\tilde{k}_{r}$
%such that in this way, we obtain almost the same profit. On the other
%hand, we will ensure that $\tilde{k}_{j}\le k_{j}$ for each $j\in[r]$.
%This crucial for the $\N$-boxes since their sizes will be determined
%by the respective values $\tilde{k}_{j}$. 

Formally, we perform $r$ iterations, one for each value $k_j$. We define $\tilde{k}_{0}:=0$. Suppose
inductively that we have determined $\ell$ values $\tilde{k}_{1},\tilde{k}_{2},\dots,\tilde{k}_{\ell}$
already for some $\ell\in\{0,1,...,r-1\}$ such that $\tilde{k}_{\ell}\le k_{\ell}$.
We want to compute $\tilde{k}_{\ell+1}$ such that  $\tilde{k}_{\ell+1}\le k_{\ell+1}$. To this end, we do a binary search on the set $S:=\{s_{i}:i\in\I\wedge \tilde{k}_{\ell}<s_{i}\}$, using our item data structure. Our first candidate value is the median of $S$
which we can find in time $O(\log n)$ via our binary search tree for the item sizes. For each candidate value $s\in S$, we estimate
up to a factor of $1+\epsilon$ the possible profit due to items in $\tilde{\I}_{\ell+1}$
if we define $\tilde{k}_{\ell+1}:=s$.
We will describe later how we compute such an estimation. The objective is to find the smallest value $s \in S$ such that the estimated profit is at least $(1+\epsilon)^{-1}\hat{p}(j)$. Hence, if for a specific guess $s$ our obtained profit due to $s$ is at least $(1+\epsilon)^{-1}\hat{p}(j)$, we restrict our set $S$ to $\{s_{i}:i\in\I\wedge \tilde{k}_{\ell} \leq s_{i} \leq s\}$ and continue with the next iteration of the binary search.
If, however, the estimated profit due to $s$ is strictly less than $(1+\epsilon)^{-1}\hat{p}(j)$, this implies that $\tilde{k}_{\ell+1} > s$ since otherwise the set of items $\tilde{\I}_{\ell+1}$ is a subset of the items considered for guess $s$ and cannot yield a larger profit. Hence, we restrict our set $S$ to $\{s_{i}:i\in\I\wedge s_{i}>s\}$ and continue with the binary search. We stop when $|S|=1$.

%We work in $r$ iterations. We define $\tilde{k}_{0}:=0$. Suppose
%inductively that we have determined $\ell$ values $\tilde{k}_{1},\tilde{k}_{2},\dots,\tilde{k}_{\ell}$
%already for some $\ell\in\{0,1,...,r-1\}$ such that $\tilde{k}_{\ell}\le k_{\ell}$.
%We want to compute $\tilde{k}_{\ell+1}$. We can assume w.l.o.g.~that
%$k_{\ell+1}$ equals $s_{i}$ for some item $i\in\I$. We do a binary
%search on the set $S:=\{s_{i}:i\in\I\wedge s_{i}>\tilde{k}_{\ell}\}$,
%using our item data structure. For each candidate value $s\in S$,
%we estimate the possible profit due to items in $\tilde{\I}_{\ell+1}$
%if we define $\tilde{k}_{\ell+1}:=s$. We want to find such a value
%$s$ such that the obtained profit from the set $\tilde{\I}_{\ell+1}$
%equals essentially $\hat{p}(\ell+1)$. In the following, we denote by $\B_{\N}$ and $\B_{\S}$, the set of $\N$- and $\S$-boxes in $\B$, respectively.

We denote by $\B_{\N}$ and $\B_{\S}$ the set of $\N$- and $\S$-boxes in $\B$ due to Lemma~\ref{lem:struc_hypercubes}, respectively. Additionally, we denote by $\B_{\N}(\ell+1)$ the set of $\N$-boxes $B$ with $j_{B} = \ell+1$ and
by $\B_{\S}(\ell+1)$ the set of $\S$-boxes $B$ with $\hat{a}_{B,\ell+1}>0$. Note that those are the only boxes that are relevant for the current iteration in which we want to determine $k_{\ell+1}$ (approximately).
We also define $\B(\ell+1) := \B_{\S}(\ell+1) \cup \B_{\N}(\ell+1)$.

We describe now how we estimate the obtained profit for one specific candidate
choice of $s\in S$.
We try to pack items from $\tilde{\I}_{\ell+1}(s):=\left\{ i\in\I:s_{i}\in(\tilde{k}_{\ell},s]\right\} $
into
\begin{itemize}
\item the $\N$-boxes $\B_{\N}(\ell+1)$ and
\item the $\S$-boxes $\B_{\S}(\ell+1)$, where for each $\S$-box $\B_{\S}(\ell+1)$, we use a volume
of at most $\hat{a}_{B,\ell+1}\cdot\mathrm{VOL}(B)$ and ensure that we pack
only items $i\in\tilde{\I}_{\ell+1}(s)$ for which $s_{i}\leq\epsilon\ell_{\min}(B)$.
\end{itemize}
We solve this subproblem approximately via the following integer program
$(\mathrm{IP}(s))$. Intuitively,
we group items such that all items in the same group have the same size and profit, up to a factor of $1+\epsilon$.
Formally, we define
\begin{itemize}
 \item a size class $\Q_{t}=\{i\in \tilde{\I}_{\ell+1}(s):s_{i}\in[(1+\epsilon)^{t},(1+\epsilon)^{t+1})\}$ for each $t\in  \mathcal{T}_Q =  \{\lfloor \log_{1+\epsilon}(\tilde{k}_\ell)\rfloor,\dots,\lceil \log_{1+\epsilon}(s)\rceil\}$; we denote by $\hat{s}(t):=(1+\epsilon)^{t+1}$ the corresponding ``rounded'' size,
 \item a profit class $\P_{t'}=\{i\in \tilde{\I}_{\ell+1}(s):p_{i}\in[(1+\epsilon)^{t'},(1+\epsilon)^{t'+1})\}$ for each  $t' \in \mathcal{T}_P =  \{\lfloor \log_{1+\epsilon} \epsilon \frac{p_{\max}}{n}\rfloor,\dots, \lceil \log_{1+\epsilon} p_{\max}\rceil\}$; we denote by $\hat{p}(t'):=(1+\epsilon)^{t'+1}$ the corresponding ``rounded'' profit and
 \item a set of pairs $\mathcal{T}:=\{(t,t') : t\in \mathcal{T}_Q \wedge t'\in \mathcal{T}_P\}$.
\end{itemize}
%  and a profit class $\P_{t'}=\{i\in \tilde{\I}_{\ell+1}(s):p_{i}\in[(1+\epsilon)^{t'},(1+\epsilon)^{t'+1})\}$ for each  $t' \in \mathcal{T'} =  \{\lfloor \log_{1+\epsilon} \epsilon \frac{p_{\max}}{n}\rfloor,\dots, \lceil \log_{1+\epsilon} p_{\max}\rceil\}$.
% % with $p_{\max}$ being the largest item profit.
% Furthermore, we denote by $\hat{s}(t):=(1+\epsilon)^{t+1}$ and $\hat{p}(t'):=(1+\epsilon)^{t'+1}$ the rounded size and profit of a class, respectively. Finally, let $n_{t,t'}$ denote the number of items of size class $\S_t$ and profit class $\P_{t'}$ for each pair $(t,t')$.
Our subproblem is now equivalent (up to a factor of $1+\epsilon$ in the profit) to choosing how many items from each group are packed into which box, while maximizing the total profit of these items, such that
\begin{itemize}
	\item for each $\N$-box $B \in \B_{\N}(\ell+1)$ the number of items packed into $B$ is at most the number of grid cells denoted by $n(B) := \prod_{d'=1}^d \hat{n}_{d'}(B)$,
	\item for each $\S$-box $B \in \B_{\S}(\ell+1)$ the total volume of items packed into $B$ does not exceed the designated volume for items in $\I_{j}$ reserved in $B$ and
	\item for each pair $(t,t')$ the number of items packed into all boxes is at most the number of available items in the corresponding group.
\end{itemize}

\begin{alignat*}{3}
(\mathrm{IP}(s))\quad& \text{max} 	& \displaystyle \sum_{(t,t') \in \mathcal{T}}\sum_{B \in \B(\ell+1)} x_{t,t',B} p(t')			& 			& \quad & \\
& \text{s.t.} & \displaystyle\sum_{(t,t') \in \mathcal{T}} x_{t,t',B} 	& \leq n(B)	& 		& \forall B \in \B_{\N}(\ell+1) \\
& & \displaystyle\sum_{(t,t') \in \mathcal{T}} x_{t,t',B}s(t)^d	& \leq a_{B,\ell+1}\mathrm{VOL}_d(B)	& 		& \forall B \in \B_{\S}(\ell+1) \\
&				& \displaystyle\sum_{B \in \B(\ell+1)} x_{t,t',B}								& \leq n_{t,t'}	& 		& \forall (t,t') \in \mathcal{T} \\
&				& x_{t,t',B}								& \in \mathbb{N}_{0}& 		&\forall (t,t') \in \mathcal{T}, B \in \B(\ell+1)
\end{alignat*}

We cannot solve $(\mathrm{IP}(s))$ directly; however, we show that we can solve it
approximately, losing only a factor of $1+\epsilon$. We describe now how to do this in time $(\log_{1+\epsilon}(N))^{O_\epsilon(1)}$. We start by guessing the
$|\B(\ell+1)|\cdot |\T|/\epsilon = O_{\epsilon,d}(1)$
 most profitable items in an optimal solution of $(\mathrm{IP}(s))$. To do this, we guess the profit type and size type of each of these items as well as which box they are packed in. For a single item this yields a total amount of $O_{\epsilon,d}((\log_{1+\epsilon}(N))^2)$ many possibilities, and hence at most $(\log_{1+\epsilon}(N))^{O_\epsilon(1)}$ possibilities overall. We adjust $(\mathrm{IP}(s))$ accordingly (i.e., reduce the right-hand sides of our constraints) and solve the LP relaxation of the remaining problem in time $\left( \log_{1+\epsilon}(N)\right)^{O(1)}$, yielding a solution $x^*$. We round it by simply defining $\bar{x}_{t,t',B}:=\lfloor x^*_{t,t',B}\rfloor$ for each $(t,t')\in \T$ and each $B\in \B(\ell+1)$. This yields a solution consisting of the guessed items together with $\bar{x}$, where $\bar{x}$ represents the remaining items in our solution.
 Since we guessed the  $|\B(\ell+1)|\cdot |\T|/\epsilon$ most profitable items before but there are only $|\B(\ell+1)|\cdot |\T|$ variables,
 we solve $(\mathrm{IP}(s))$ up to a factor of $1+\epsilon$.

\begin{lemma}\label{lem:cubes_IP_sol}
There is an algorithm with a running time of $(\log_{1+\epsilon}(N))^{O(1)}$
that computes a $(1+\epsilon)$-approximate solution for \textup{$(\mathrm{IP}(s))$} for each $s$;
we denote by $q(s)$ the value of this solution. For two values $s,s'$
with $s\le s'$ we have that $q(s)\le q(s')$.
\end{lemma}
%Since we solve $(\mathrm{IP}(s))$ up to a factor of $1-\epsilon$,
%we have that $q(k_{\ell+1})\ge (1-\epsilon)\hat{p}(\ell+1)$
%since $\tilde{k}_{\ell}\le k_{\ell}$ and, therefore, $\tilde{\I}_{\ell+1}(k_{\ell})=\left\{ i\in\I:s_{i}\in[\tilde{k}_{\ell},k_{\ell})\right\} \subseteq\left\{ i\in\I:s_{i}\in[k_{\ell},k_{\ell})\right\} =\I_{\ell+1}$.
At the end of our binary search, we define $\tilde{k}_{\ell+1}$ to be the smallest value $s\in S$ for which $q(s)\ge (1-\epsilon)\hat{p}(\ell+1)$. Let $x^*_{\ell+1}$ denote the computed solution to $(\mathrm{IP}(s))$ corresponding
to $s=\tilde{k}_{\ell+1}$. Based on the inductive assumption $\tilde{k}_{\ell}\le k_{\ell}$ and our choice of $\tilde{k}_{\ell+1}$, we can show $\tilde{k}_{\ell+1}\le k_{\ell+1}$. This is crucial as $\tilde{k}_{\ell+1}$ determines the side lengths of an $\N$-box $B$ with $j_B = \ell+1$. Thus, in order to be able to pack our guessed boxes into the knapsack we must guarantee that these side lengths are not larger than the side lengths of the corresponding box in the packing underlying Lemma~\ref{lem:struc_hypercubes}.
\begin{lemma}\label{lem:cubes_induc_kr}
We have that $\tilde{k}_{\ell+1}\le k_{\ell+1}$.
\end{lemma}
After completing the $r$ rounds of our indirect guessing framework, we have obtained values $\tilde{k}_{1},\tilde{k}_{2},\dots,\tilde{k}_{r}$ and $r$
integral solutions to $(\mathrm{IP}(\tilde{k}_{1})),...,(\mathrm{IP}(\tilde{k}_{r}))$.
We can construct a feasible solution $\ALG$ by combining these solutions:
for each $j\in[r]$ and each $\N$-box $B\in\B$ we define $s_{\max}(B):=\tilde{k}_{j_{B}}$
and we assign the items of each set $\tilde{\I}_{j}$ into
the boxes in $\B$ according to our solution to $\mathrm{IP}(\tilde{k}_{j}))$.

For each $j \in \{1,\dots,r\}$, each size class $\Q_t$, and each profit class $\P_{t'}$, there is a certain number of items from the set
$\tilde{\I}_{j} \cap \Q_t \cap \P_{t'}$ that our solution contains; we denote this number by $z_{j,t,t'}$. It is irrelevant which exact items from
this set we pick (up to a factor of $1+\epsilon$ in the profit). Thus, we simply order the items from the set $\tilde{\I}_{j} \cap \Q_t \cap \P_{t'}$ in non-decreasing order of side lengths, select the first $z_{j,t,t'}$ items, and assign them to the corresponding boxes.
% To do so, we first compute a value $z_{j,t,t'}$ for each $j \in \{1,\dots,r\}$, each size class $\S_t$, and each profit class $\P_{t'}$ denoting how many items from the set $\tilde{\I}_{j} \cap \S_t \cap \P_{t'}$ are contained in our solution. As it is irrelevant (up to a factor of $1+\epsilon$) which items we pick, we may pick
% items from the set $\tilde{\I}_{j} \cap \S_t \cap \P_{t'}$ in non-decreasing order of side lengths. For an item $i\in\tilde{\I}_{j}\cap\S_{t}\cap\P_{t'}$. We then assign the correct number of items from each pair $(t,t')$ to each box $B$.
Finally, we pack each $\S$-box $B$ using the NFDH-algorithm (Lemma~\ref{lem:NFDH}) and each $\N$-box $B$ by placing each item into a single cell of the $d$-dimensional grid with side length $s_{\max}(B)$.
% As the translation of the implicit packing into an explicit packing can be done in time $O(n\cdot(\log^2 n))$, we obtain our polynomial time approximation scheme.

%The $\S$-boxes are then packed using the NFDH-algorithm (Lemma~\ref{lem:NFDH}).
%For each $\N$-box $B$, we place a $d$-dimensional grid with side
%length $s_{\max}(B)$ and pack one item into each grid cell.
This yields an algorithm with a running time of $n\cdot(\log^2 n)+(\log N)^{O_{\epsilon,d}(1)}$. In Appendix~\ref{app:cubes} we explain how we can improve it to $n\cdot(\log^2 n)+(\log n)^{O_{\epsilon,d}(1)}$. Key to this is to guess the basic quantities and to solve each integer program $(\mathrm{IP}(s))$ in time $(\log n)^{O_{\epsilon,d}(1)}$, using additional technical improvements such as restricting the range of the considered item sizes and profits
when solving $(\mathrm{IP}(s))$. We would like to remark that the exponent of $\log n$ depends on the value $C_{\mathrm{boxes}}(\epsilon,d)$ in Lemma~\ref{lem:struc_hypercubes} which in turn depends on the (not precisely specified) number of boxes used in the structured packing in~\cite{jansen2022ptas}.

\begin{theorem}\label{thm:cubes_ptas}
There is a $(1+\epsilon)$-approximation algorithm for the $d$-dimensional
hypercube knapsack problem with a running time of $n\cdot(\log^2 n)+(\log n)^{O_{\epsilon,d}(1)}$.
\end{theorem}

\subsection{Dynamic algorithm}
The algorithmic techniques above can be combined with our item data structure to derive a dynamic algorithm for the $d$-dimensional hypercube knapsack problem. Our algorithm supports the following operations:
\begin{enumerate}[(i)]
	\item Insertion and deletion of an item into our data structure,
	\item Output a $(1+\epsilon)$-estimate of the value of the optimal solution, 
	\item Output a $(1+\epsilon)$-approximate solution $\ALG$ and
	\item Query where a given item is contained in $\ALG$.
\end{enumerate}
For operation (i) we simply add or delete an item from our item data structure (see Lemma~\ref{lem:data-structure}) and our balanced binary search trees, which takes time $O(\log^2 n)$.
In order to execute operation (ii), we run our algorithm described above, except for computing the explicit packing of the items in the end. Instead, we simply return the total profit of
our solutions to the integer programs $(\mathrm{IP}(s))$ that correspond the solution that we output at the end.
 This takes time $(\log n)^{O_{\epsilon,d}(1)}$ in total. If a $(1+\epsilon)$-approximate solution $\ALG$ is queried (operation (iii)), we also compute the exact set of items and their packing as described previously. Since their total number is $|\ALG|$, we can compute and output $\ALG$ in time $O(|\ALG|\cdot(\log n))+(\log n)^{O_{\epsilon,d}(1)}$.

 Finally, if it is queried whether a given item $i \in \I$ is in contained in $\ALG$ (operation (iv)), we determine the value  $j \in \{1,\dots,r\}$, the size class $\Q_t$, and the profit class $\P_{t'}$ for which $i\in \tilde{\I}_{j} \cap \Q_t \cap \P_{t'}$. Recall that $z_{j,t,t'}$ denotes the total number of items from this set we select and we select the
 $z_{j,t,t'}$ items in this set of shortest side length. Hence, we output ``$i\in\ALG$'' if $i$ is among the shortest $z_{j,t,t'}$ items in $\tilde{\I}_{j} \cap \Q_t \cap \P_{t'}$, and ``$i\notin\ALG$'' otherwise. This ensures that we give consistent answers between consecutive updates of the set $\I$.
 % between two consecutive updates of the set $\I$ we give consistent answers to queries.
 
%
%  , we do not need to compute the whole solution $\ALG$. However, as one may expect the solution to change a lot when an item is inserted or deleted we make use of how we construct the explicit solution. Recall that we use the indirect guessing framework to compute a value $z_{j,t,t'}$ for each $j \in \{1,\dots,r\}$, each size class $\S_t$, and each profit class $\P_{t'}$ denoting how many items from the set $\tilde{\I}_{j} \cap \S_t \cap \P_{t'}$ are contained in $\ALG$. As it is irrelevant (up to a factor of $1+\epsilon$) which items we pick, we may pick items from the set $\tilde{\I}_{j} \cap \S_t \cap \P_{t'}$ in non-decreasing order of side lengths. For an item $i\in\tilde{\I}_{j}\cap\S_{t}\cap\P_{t'}$,
% we output ``$i\in\ALG$'' if $i$ is among the first $z_{j,t,t'}$
% items, and ``$i\notin\ALG$'' otherwise. This ensures that we give
% consistent answers between two consecutive updates of the set $\I$.
\begin{theorem}\label{thm:dyn_alg_cubes}
There is a dynamic algorithm for the $d$-dimensional hypercube knapsack
problem that supports the following operations:
\begin{enumerate}[(i)]
\item insertion or deletion of an item in time $O(\log^2 n)$,
\item output a $(1+\epsilon)$-estimate of the value of the optimal solution in time $(\log n)^{O_{\epsilon,d}(1)}$,
\item output a $(1+\epsilon)$-approximate solution $\ALG$ in time $O(|\ALG|\cdot(\log n))+(\log n)^{O_{\epsilon,d}(1)}$,
\item query whether an item is contained in $\ALG$ in time $(\log n)^{O_{\epsilon,d}(1)}$.
\end{enumerate}
\end{theorem}

\section{Algorithms for rectangles}\label{sec:rec_main}
In this section we give an overview of our algorithms for two-dimensional knapsack for rectangles (see Appendix~\ref{app:rec} for details).
First, we classify items into four groups: we say that an item is \emph{large} if it is large compared to edge length of the knapsack in both dimensions and \emph{small} if it is small in both dimensions.

An item of relatively large width (height) and relatively small height (width) is referred to as \emph{vertical} (\emph{horizontal}). We construct \emph{easily guessable packing} using four types of boxes.
The first type are \emph{$\L$-boxes} which contain only one large item each. Also, we use \emph{$\H$-boxes} inside which horizontal items are stacked on top of each other (see Figure~\ref{fig:pack_VH}), and correspondingly \emph{$\V$-boxes} for vertical items.
% inside which contain horizontal and vertical the items are intuitively horizontally and vertically stacked, respectively. An $\H$-box contains only horizontal items while a $\V$-box contains only vertical items.
Finally, we use $\S$-boxes which are defined in the same manner as in the case of (hyper-)cubes. 
\begin{figure}
	\centering
	\includegraphics[width=.9\linewidth,page = 3]{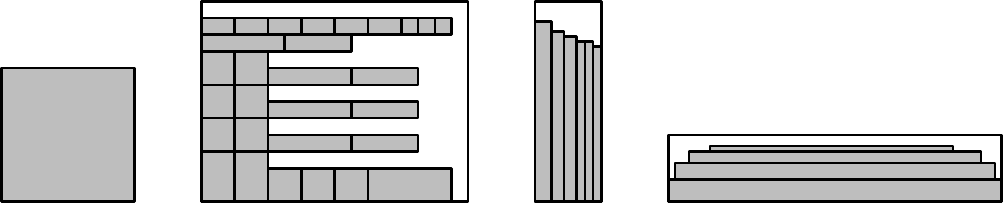}
	\caption{Visualization of packing of $\L$-, $\S$-,$\H$- and $\V$-box, respectively.}
	\label{fig:pack_VH}
\end{figure}
We prove that there always exists a $(2+\epsilon)$-approximate packing using these types of boxes.
\begin{lemma}[Informal]
	There exists an easily guessable packing for packing rectangles into a two-dimensional knapsack with a profit of at least $(1/2-\epsilon)\OPT$.
\end{lemma}
In these packings, we can guess the height of each box (of each type) in $O_\epsilon(1)$ time and the width of each $\S$- and $\V$-box in $O(\mathrm{poly}(\log n))$ time (and additionally some other basic quantities). Then, we apply the indirect guessing framework in order to determine the widths of the $\L$- and $\H$-boxes.
%
% only $\L$-, $\H$-, $\V$-, and $\S$-boxes such that for each box we can guess its height and additional basic quantities  in $\mathrm{poly}(\log n)$ time. Informally, we prove the following result (see Appendix~\ref{app:rec_nr} for the formal results and technical details).
% Based on this structural result, we follow our framework of guessing basic quantities and modifying the indirect guessing framework to suit our problem and design an approximation algorithm as well as a dynamic algorithm with the corresponding approximation guarantees.

\begin{restatable}{theorem}{thmrecnr}
	\label{thm:rec_nr}
	There is a $(2+\epsilon)$-approximation algorithm for the geometric
	knapsack problem for rectangles with a running time of $n\cdot(\log n)^{4}+(\log n)^{O_{\epsilon}(1)}$.
	Also, there is a dynamic $(2+\epsilon)$-approximation algorithm for
	the problem which supports the following operations:
	\begin{itemize}
		\item insert or delete an item in time $O(\log^4 n)$,
		\item output a $(2+\epsilon)$-estimate of the
		value of the optimal solution, or query whether an item is contained
		in $\ALG$, in time $(\log n)^{O_{\epsilon}(1)}$, and
		\item output a $(2+\epsilon)$-approximate solution $\ALG$ in time $O(|\ALG|\cdot(\log n))+(\log n)^{O_{\epsilon}(1)}$.
	\end{itemize}
\end{restatable}

A natural open question is to improve our approximation ratio to $2-\delta$ for some constant~$\delta>0$. This seems difficult since there is provably no corresponding structured packing with only $O_\epsilon(1)$ boxes \cite{galvez2021approximating}. The known polynomial time $(17/9+\epsilon)$-approximation uses an L-shaped container which is packed by a DP with a running time of %$O(n^2N^2)$ or, alternatively,
$n^{\Omega_\epsilon(1)}$~\cite{galvez2021approximating}. It is not clear how to improve this to near-linear running time.
% Also, $\Omega_\epsilon(1)$ ``large'' rectangles are guessed which cannot be handled by our indirect guessing framework.
However, 
if we are allowed to rotate the rectangles by 90 degrees, then it is possible to construct an
% we construct even an
easily guessable packing with $O_\epsilon(1)$ boxes and an approximation ratio of only $17/9+\epsilon$. We use here that we have more freedom to modify the optimal packing by rotating some of its items.

\begin{lemma}[Informal]
	If we are allowed to rotate the input rectangles, there exists an easily guessable packing into a two-dimensional knapsack with a profit of at least $(9/17-\epsilon)\OPT$.
\end{lemma}

On the other hand, it becomes harder to compute a solution that corresonds to our easily guessable packing. The reason is that a horizontal or vertical item can now be assigned to an $\H$- \emph{or} to a $\V$-box. This is particularly problematic since these two types of boxes are not treated symmetrically. Like before, we can guess the height of each box in time $O_\epsilon(1)$. However, for $\V$-boxes this yields a different kind of restriction than for $\H$-boxes.

% If we are allowed to rotate these items, we even present a $17/9+\epsilon$-approximation algorithm and a corresponding dynamic algorithm. Here, we first prove that an \emph{easily guessable packing} exists that allows for this approximation guarantee and is suitable for our algorithmic framework of first guessing basic quantities and then applying an indirect guessing framework. Again, we first guess the height of each box and some additional parameters and use the indirect guessing framework to compute the width of each box. The indirect guessing framework differes from the setting without rotation as we may now treat horizontal and vertical items as the same.

\begin{restatable}{theorem}{theoremrecrot}
	\label{thm:rectangles_rot}
	There is a $(17/9+\epsilon)$-approximation algorithm for the geometric
	knapsack problem for rectangles with rotations with a running time
	of $n\cdot(\log n)^{4}+(\log n)^{O_{\epsilon}(1)}$. Also, there
	is a dynamic $(17/9+\epsilon)$-approximation algorithm for the problem
	which supports the following operations:
	\begin{itemize}
		\item insert or delete an item in time $O(\log^4 n)$,
		\item output a $(17/9+\epsilon)$-estimate of the value of the optimal solution, or query whether an item is contained
		in $\ALG$, in time $(\log n)^{O_{\epsilon}(1)}$,
		\item output a $(17/9+\epsilon)$-approximate solution $\ALG$ in time $|\ALG|\cdot(\log n)^{3}+(\log n)^{O_{\epsilon}(1)}$.
	\end{itemize}
\end{restatable}

\bibliography{arxiv-dynamic-knapsack}

\appendix

\section{Details of Section~\ref{sec:hypercubes}}\label{app:cubes}
In the following we present the technical details underyling our algorithms presented in Section~\ref{sec:hypercubes}.
\subsection{Preliminaries}
Before going into more details, we need to introduce some important notation.
As defined earlier, for each item $i \in \I$, the side length rounded to the next larger power of $1+\epsilon$ is denoted by $\lceil s_i \rceil_{1+\epsilon} := (1+\epsilon)^{\lceil \log_{1+\epsilon}(s_i)\rceil}$.
% Let $\I$ be a set of items.
We denote by $s_{\max}(\I) := \max_{i \in \I}s_i$, $s_{\min}(\I) := \min_{i \in \I}s_i$ the largest and smallest side lengths of the items in $\I$, respectively. Furthermore, the total volume of the items in $\I$ is denoted by $\mathrm{VOL}_d(\I) := \sum_{i \in \I}s_i^d$ and the total volume of the rounded items by  $\lceil \mathrm{VOL}_d(\I)\rceil_{1+\epsilon} := \sum_{i \in \I}\lceil s_i\rceil^d_{1+\epsilon}$. For a $d$-dimensional hypercuboid $B$, we denote its side lengths by $\ell_{d'}(B)$ for $d'=1,\dots,d$ and the smallest side length is $\ell_{\min}(B) := \min_{d'=1,\dots,d}\ell_{d'}(B)$. The volume and surface of $B$ are defined as  $\mathrm{VOL}_d(B) := \prod_{d'=1}^d \ell_{d'}(B)$ and $\mathrm{SURF}_d(B) := 2\sum_{d'=1}^d\mathrm{VOL}_d(B)/\ell_{d'}(B)$. We refer to a $d$-dimensional hypercuboid as a \emph{box}.

Furthermore, we make use of the Next-Fit-Decreasing-Height (NFDH) algorithm when we show the existence of a
% in the derivation of a
structured near-optimal solution, as well as in our algorithm itself. Following {the} existing literature (see e.g.~\cite{bansal2006bin, galvez2021approximating,harren2009approximation,jansen2022ptas}), we describe NFDH in $d$-dimensions in an inductive manner. Suppose that we defined already how to
% know how to
pack items in $d-1$ dimensions using NFDH. Let $B$ be a $d$-dimensional hypercuboid and $\I$ be a set of items sorted in non-increasing order of heights (in the case of hypercubes this is equivalent to the side length). Consider the largest item in $\I$, with side length $s_{\max}(\I)$; this item defines the length of the first shelf of our packing in dimension~$d$, i.e., the first shelf is a sub-box of $B$ with side lengths $\ell_1(B),\dots,\ell_{d-1}(B), s_{\max}(\I)$. We then apply the $d-1$-dimensional NFDH algorithm using the set $\I$ and the first shelf. Afterwards, let $P$ be the set of items packed into the first shelf, then we run $d$-dimensional NFDH with $B'$ being a box of side lengths $\ell_1(B),\dots,\ell_{d-1}(B),\ell_d(B)-s_{\max}(\I)$ and $\I' = \I \setminus P$. We repeat this procedure until no more items can be packed.
Finally, consider the (base) case that we run NFDH for $d=1$.
% To ensure that $d$-dimensional NFDH is well-defined consider the setting in $1$-dimension.
Here, we add items until the next item does not fit. This item will then define the length of the new shelf in the second dimension. In this paper, we assume w.l.o.g. that before packing items into a $d$-dimensional hypercuboid we sort the dimensions such that $\ell_d(B) \leq \dots \leq \ell_1(B)$. The running time of NFDH in $d$-dimensions is in $O(n \log n)$~\cite{bansal2006bin}. See Figure~\ref{fig:origpack_S} for a visualization of a packing constructed by NFDH for $d=2$.

Since the structured packing in~\cite{jansen2022ptas} is the starting point of our new structured packing, we include the result for completeness. The structured packing used in~\cite{jansen2022ptas} is based on $\V$- and $\N$-boxes as defined in Definition~\ref{def:jansen_boxes}. Based on these two types of boxes, the following structural result is given in~\cite{jansen2022ptas}.
\begin{theorem}[Theorem 7 of Jansen et al.~\cite{jansen2022ptas}]\label{thm:dD_cubes_jansen}
	For any instance $\I$ of the $d$-dimensional hypercube knapsack problem and any $\epsilon < 1/2^{d+2}$, there exists a packing with the following properties:
	\begin{enumerate}
		\item It consists of $\mathcal{N}$- and $\mathcal{V}$-boxes whose total number is bounded by a constant $C_{boxes}(d,\epsilon)$, which depends only on $\epsilon$ and $d$. 
		\item The number of items in the packing that are not packed in these boxes is bounded by a constant $C_{large}(d,\epsilon)$, which depends only on $\epsilon$ and $d$.
		\item The total profit of the packing is at least $(1-2^{d+2}\epsilon)OPT(\I)$, where $OPT(\I)$ is the profit of an optimal packing for instance $\I$.
	\end{enumerate}
\end{theorem}

\subsection{Structured packing}

We now derive the structural packing of Lemma~\ref{lem:struc_hypercubes}. Our structured packing uses $\N$- and $\S$-boxes as defined in Section~\ref{sec:hypercubes}. See Figure~\ref{fig:origpack} for a visualization of these boxes and packings.
\begin{figure}[h!]
	\begin{minipage}{.4\textwidth}
		\centering
		\includegraphics[width=.4\linewidth, page = 1]{figures/Squares2D.pdf}
		\subcaption{Sorted packing of $\N$-box}
		\label{fig:origpack_N}
	\end{minipage}%
	\begin{minipage}{.4\textwidth}
		\centering
		\includegraphics[width=.4\linewidth,page = 2]{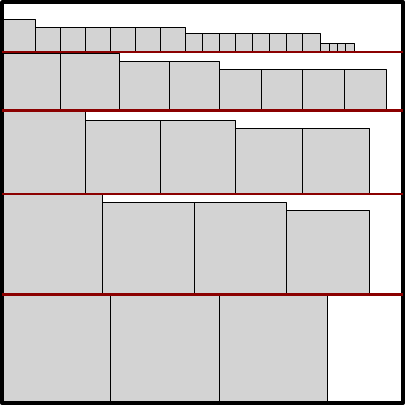}
 		\subcaption{\textit{NFDH} packing of $\S$-box} 
		\label{fig:origpack_S}
	\end{minipage}
	\caption{Visualization of $\N$- and $\S$-box for $d=2$}
	\label{fig:origpack}
\end{figure}

We first prove a few auxilliary results. We start by showing that if all items packed into a box $B$ are small compared the box itself, then using an argumentation via a linear program (LP) we can select a subset of the items packed into $B$, such that for this selected subset of items the box $B$ is an $\S$-box.

% if we keep only this subset as a packing in $B$, then $B$ is a $\S$-box.

\begin{lemma}\label{lem:vstarbox_selection}
	Let $B$ be a $d$-dimensional hypercuboid and $\I$ be the set of items packed into $B$ such that $s_i \leq \epsilon \ell_{\min}(B)$ for all items $i\in \I$. Then, there exists a subset of items $\I'\subseteq \I$ such that $p(\I') \geq (1-O(\epsilon))p(\I)$ and $\sum_{i \in \I'} \lceil s_i \rceil^d_{1+\epsilon} \leq \mathrm{VOL}_d(B)-2\cdot d \cdot \epsilon \mathrm{VOL}_d(B)$ with $\epsilon \in (0,1/2^{d+2})$.
\end{lemma}
\begin{proof}
	To prove this statement, we consider the one dimensional knapsack problem with capacity $\mathrm{VOL}_d(B)-(2d+1)\cdot\epsilon \mathrm{VOL}_d(B)$ and item set $\I$ where we define the item size as $\lceil s_i \rceil_{1+\epsilon}^d$. For this problem, we consider the following LP relaxation:

	\begin{alignat*}{3}
		& \text{minimize} & \sum_{j=1}^{m} &x_ip_i& \\
		& \text{subject to} \quad& \sum_{i \in \I}&x_{i}\lceil s_i \rceil_{1+\epsilon}^d& \leq \mathrm{VOL}_d(B)-(2d+1)\cdot\epsilon \mathrm{VOL}_d(B) \\
		&&& x_{i} \geq 0, &  \forall i \in \I
		\\
		&&& x_{i} \leq 1, &  \forall i \in \I
	\end{alignat*}

	Consider the solution $\hat{x}_i := (1-(2d+1)\cdot\epsilon)(1+\epsilon)^{-d}$ for all $i\in \I(B)$. Clearly, $\hat{x_i} \leq 1$ and $\hat{x_i} \geq 0$ for each $i \in \I$. Furthermore, we have
	\begin{align*}
		\sum_{i \in \I}\hat{x}_{i}\lceil s_i \rceil_{1+\epsilon}^d & = (1-(2d+1)\cdot\epsilon)(1+\epsilon)^{-d}\sum_{i \in \I}\lceil s_i \rceil_{1+\epsilon}^d \\
		& \leq (1-(2d+1)\cdot\epsilon)(1+\epsilon)^{-d}\sum_{i \in \I} (1+\epsilon)^ds_i^d \\
		& = (1-(2d+1)\cdot\epsilon)\sum_{i \in \I} s_i^d \\
		& \leq  (1-(2d+1)\cdot\epsilon)\mathrm{VOL}_d(B).
	\end{align*}
	Here, the last inequality follows from the fact that the total volume of items packed into $B$ is at most the volume of $B$ itself. Thus, $\hat{x}$ is feasible and yields a profit of at least $(1-(2d+1)\cdot\epsilon)(1+\epsilon)^{-d}p(B) \geq (1-O(\epsilon))p(B)$. Let $x^*$ be an optimal extreme point solution. Then, $x^*$ yields a profit of at least $(1-O(\epsilon))p(B)$. Furthermore, due to the rank lemma (see e.g.~\cite{lau2011iterative}), we know that there is at most one fractional variable in the support of $x^*$. We now obtain $\I'$ by taking all items
	$i \in \I$ for which $x^*_i = 1$ and, additionally, the unique item $i' \in \I(B)$ for which $0< x^*_{i'} < 1$ if such an item exists. Using the argumentation above, we know that $p(\I') \geq (1-O_\epsilon(1))p(B)$. Furthermore, as item $i'$ has side length at most $\epsilon \ell_{\min}(B)$, if we round up its side length to the next larger power of $1+\epsilon$, its volume is still bounded by
	% 	its volume after rounding up its side lengths is at most
	\[(1+\epsilon)^d\epsilon^d\mathrm{VOL}_d(B).\]
	Since $\epsilon < 1/2^{d+2}$, we have that 
	\[(1+\epsilon)^d\epsilon^d < \epsilon.\]
	Therefore, the total volume of the rounded items in $\I'$ is at most \[\mathrm{VOL}_d(B)-2d \cdot \epsilon \mathrm{VOL}_d(B). \]
	This concludes the proof. 
\end{proof}

The next two lemmas show that
if a box $B$ is packed using NFDH, then we can split it
% we can split a box $B$ which is packed using NFDH
into $O_{\epsilon,d}(1)$ many $\N$-boxes if the side lengths of the items originally packed into $B$ fall within a certain range.

\begin{lemma}\label{lem:aux_rangebox_single}
	Let $B$ be a $d$-dimensional hypercuboid and $\hat{\I}$ be a set of items such that
	\begin{itemize}
		\item all items in $\hat{\I}$ can be packed into $B$ using $d$-dimensional NFDH and
		\item for all $i \in \hat{\I}$ it holds that $s_i  \in ((1+\epsilon)^\alpha,(1+\epsilon)^{\alpha+1}]$ for a fixed $\epsilon >0$ and $\alpha > 0$.
	\end{itemize}
	Then, $B$ can be transformed into at most $1/\epsilon^d$ many $\N$-boxes such that a subset of items $\hat{\I}' \subseteq \hat{\I}$ can be packed into these $\N$-boxes with $p(\hat{\I}') \geq (1-2\epsilon)^d p (\hat{\I})$.
\end{lemma}
\begin{proof}
	We prove this statement in an inductive manner. First, let $d=1$ and consider the following cases:
	\begin{itemize}
		\item \textbf{Case 1:} Suppose $\ell_1(B)/(1+\epsilon)^\alpha \geq 1/\epsilon$. The number of items packed into $B$ is at most 
		\[\frac{\ell_1(B)}{(1+\epsilon)^\alpha}.\]
		We transform $B$ into a $\N$-box $B_1$ with $s_{\max}(B)= (1+\epsilon)^{\alpha + 1}$ and $n_1= \left\lfloor \frac{\ell_1(B)}{s_{\max}(\hat{\I})} \right\rfloor$. The number of items we can pack into $B_1$ is
		\[n_1 \geq \frac{\ell_1(B)}{(1+\epsilon)^{\alpha+1}} -1 \geq \frac{\ell_1(B)}{(1+\epsilon)^{\alpha}}\left(1-\frac{\epsilon}{1+\epsilon}-\epsilon\right) \geq (1-2\epsilon)\frac{\ell_1(B)}{(1+\epsilon)^{\alpha}}.\]
		Thus, we can pack all but a $2\epsilon$-fraction of the items in $\hat{\I}$ into $B_1$ which means that by packing the most profitable items into $B_1$ we lose a profit of at most $2\epsilon p(\hat{\I})$.
		\item \textbf{Case 2:} Suppose $\ell_1(B)/(1+\epsilon)^\alpha < 1/\epsilon$. Then, the number of items packed into $B$ is $|\hat{\I}| < 1/\epsilon$. We transform $B$ into $|\hat{\I}|$ many $\N$-boxes where each item is packed into its own $\N$-box. Since, we do not lose any items this way we also do not lose any profit.
	\end{itemize}
	Now, assume the statement is true for $d-1$ and let $B$ be a $d$-dimensional hypercuboid. Again, we may distinguish two cases.
	\begin{itemize}
		\item \textbf{Case 1:} Suppose $\ell_d(B)/(1+\epsilon)^\alpha \geq 1/\epsilon$. As we use the sorted variant of NFDH this implies that $\ell_{d'}(B)/(1+\epsilon)^\alpha \geq 1/\epsilon$ for all $d'=1,\dots,d$. We now transform $B$ into a single $\N$-box $B'$ with $s_{\max}(B') = s_{\max}(\hat{\I})$ and $n_{d'}= \lfloor \ell_{d'}(B)/s_{\max}(B') \rfloor$. Observe that for $d'=1,\dots,d$ we have
		\[n_{d'} \geq \frac{\ell_{d'}(B)}{(1+\epsilon)^{\alpha+1}} -1 \geq \frac{\ell_{d'}(B)}{(1+\epsilon)^{\alpha}}\left(1-\frac{\epsilon}{1+\epsilon}-\epsilon\right) \geq (1-2\epsilon)\frac{\ell_{d'}(B)}{(1+\epsilon)^{\alpha}}.\]
		Therefore, the number of items we can pack into $B'$ is at least 
		\[\prod_{d'=1}^d n_{d'} \geq (1-2\epsilon)^d\prod_{d'=1}^d \frac{\ell_{d'}(B)}{(1+\epsilon)^{\alpha}}.\]
		The number of items packed into $B$ is at most
		\[\prod_{d'=1}^d \frac{\ell_{d'}(B)}{(1+\epsilon)^{\alpha}}.\]
		Therefore, we keep a $(1-2\epsilon)^d$-fraction of the items and by packing the items into $B'$ in non-increasing order of profits, we have $p(\hat{\I}') \geq (1-2\epsilon)^dp(\hat{\I})$.
		\item \textbf{Case 2:} Suppose $\ell_d(B)/(1+\epsilon)^\alpha < 1/\epsilon$. Then, we know that NFDH uses less than $1/\epsilon$ shelves in the $d$-th dimension of the packing of $B$. We may consider each of these shelves as a $d-1$ dimensional box, which, by induction can be transformed into at most $1/\epsilon^{d-1}$ many $d-1$-dimensional $\N$-boxes. By setting $n_d = 1$ for each of these, we create at most $1/\epsilon^d$ many $d$-dimensional $\N$-boxes. The guarantee of the remaining profit follows from the fact that the guarantee holds for each of the $\N$-boxes individually.
	\end{itemize}
	This concludes the proof.
\end{proof}
We now prove a more general statement. Consider a box that is packed using NFDH such that the side length of the \emph{smallest} item packed into it is at least an $\epsilon$-fraction of the side length of the \emph{largest} item packed into it. We show that it can be transformed into $O_{\epsilon,d}(1)$ many $\N$-boxes.
%
% In particular, we prove that a box which is packed using NFDH such that the side lengths of the smallest item packed into it is at least an $\epsilon$ of the side length of the largest item packed into it can be transformed into $O_{\epsilon,d}(1)$ many $\N$-boxes.
\begin{lemma}\label{lem:aux_rangebox_multi}
	Let $B$ be a $d$-dimensional hypercuboid and $\hat{\I}$ a set of items such that $\hat{\I}$ can be packed into $B$ using NFDH and $s_{\min}(\hat{\I}) \geq \epsilon s_{\max}(\hat{\I})$. Then, $B$ can be transformed into at most $\left(1/\epsilon \right)^{d+2}$ many $\N$-boxes such that a subset $\hat{\I}' \subseteq \hat{\I}$ of items can be packed into them with $p(\hat{\I}') \geq (1-O(\epsilon))p(\hat{\I})$.
\end{lemma}
\begin{proof}
	Again we prove this statement by induction. First, consider the one-dimensional setting ($d=1$). We group the item sizes into intervals $((1+\epsilon)^\alpha,(1+\epsilon)^{\alpha+1}]$ with $\alpha = \alpha_1,\dots, \alpha_2$. Observe that $\alpha_1 =\lfloor \log_{1+\epsilon}(s_{\min}(\I))\rfloor$ and $\alpha_2 = \lceil \log_{1+\epsilon}(s_{\min}(\I))\rceil$ and, therefore, there are at most $O(1/\epsilon)$ such intervals. Next, we split $B$ into sub-boxes such that each sub-box contains only items from one of these intervals. For each of these sub-boxes we can apply Lemma~\ref{lem:aux_rangebox_single}. This gives a total of at most $O(1/\epsilon^2)$ many $\N$-boxes. Since in the transformation underlying Lemma~\ref{lem:aux_rangebox_single} we lose at most an $\epsilon$-fraction of the profit in each sub-box of $B$, in total we lose at most an $\epsilon$-fraction of the profit packed into $B$.
	
	Suppose now that the statement is true for $d-1$-dimensional hypercuboids and consider a $d$-dimensional hypercuboid $B$. We proceed in a similar fashion as above. Again, we group item sizes into intervals $((1+\epsilon)^\alpha,(1+\epsilon)^{\alpha+1}]$ with $\alpha = \alpha_1,\dots, \alpha_2$. We split $B$ into several sub-boxes as follows. For each $\alpha \in \{\alpha_1 ,...,\alpha_2 \}$, we consider all shelves of $B$ in which every item size is in the interval $((1+\epsilon)^\alpha,(1+\epsilon)^{\alpha+1}]$ and call the sub-box formed by these shelves $B(\alpha)$. Observe that there might be shelves of $B$ which contain items from multiple intervals. Each such shelf forms its own sub-box, we denote by $C(\alpha)$ the sub-box of this type in which the smallest item contained in it has a side length in $((1+\epsilon)^\alpha,(1+\epsilon)^{\alpha+1}]$. The number of sub-boxes we find in this way is at most $O(1/\epsilon)$. We now argue how to treat each of these sub-boxes.
	
	First, consider a sub-box $B(\alpha)$ and let $\I(\alpha)$ be the set of items packed into $B(\alpha)$. By Lemma~\ref{lem:aux_rangebox_single} we know that can be split into at most $O(1/\epsilon^d)$ many $\N$-boxes containing a subset $\I'(\alpha) \subseteq \I(\alpha)$ of items with a total profit of at least $(1-2\epsilon)^dp(\I(\alpha))$. 
	
	Next, consider a sub-box $C(\alpha)$. Since $C(\alpha)$ contains only one shelf in dimension $d$, we can use an inductive argument here by applying the fact that $C(\alpha)$ in $d-1$ dimensions can be split into $O(1/\epsilon^{d+1})$ many $\N$-boxes such that a subset $\I'(C(\alpha)) \subseteq \I(C(\alpha))$ of items can be packed into them with $p(\I'(C(\alpha))) \geq (1-O(\epsilon))p(\I(C(\alpha)))$. These boxes can be transformed into $\N$-boxes in $d$-dimensions by setting $n_d=1$ for each of them.
	
	Finally, since the profit guarantee holds for each sub-box separately, we know that the total profit packed into these sub-boxes is at least $(1-O(\epsilon))p(\I)$ and the total number of $\N$-boxes is at most $O(1/\epsilon^{d+2})$.
\end{proof}

Next, we show that if $d=1$, any box which is packed using NFDH can be transformed into at most $O_\epsilon(1)$ many $\N$- and $V$- boxes.
\begin{lemma}\label{lem:1DCubes_transform}
	Let $B_0$ be a $1$-dimensional box and let $\hat{\I}$ be the set of items packed into $B_0$ using NFDH and $\epsilon < 1/2^{3}$ be given. Then, $B_0$ can be split into $O(1/\epsilon)$ many $\S$- and $\N$-boxes such that a subset $\hat{\I}' \subseteq \hat{\I}$ can be packed into them with $p(\hat{\I}') \geq (1-O(\epsilon))p(\hat{\I})$.
\end{lemma}
\begin{proof}
	Let $\ell_1(B_0)$ be the side length of $B_0$. % (or in the 1-dimensional case the capacity of the knapsack).
	Let $G_1$ be the group of items of size strictly larger than $\epsilon \ell_1(B_0)$. If $p(G_1) \leq \epsilon p(\hat{\I})$, let $\hat{\I}' := \hat{\I} \setminus G_1$. Otherwise, let $B_1$ be the sub-box of $B_0$ containing all items of size at most $\epsilon \ell_1(B_0)$. We now repeat the procedure with $G_2$ being all items packed into $B_1$ of size strictly larger than $\epsilon \ell_1(B_1)$. Again, if $p(G_2) \leq \epsilon p(\hat{\I})$, we are done. Otherwise, we continue. This gives us boxes $B_0,\dots,B_g$ and $G_1,\dots G_{g},G_{g+1}$ such that:
	\begin{itemize}
		\item All items packed into $B_g$ have side lengths at most $\epsilon \ell_{\min}(B_g)$.
		\item For each $G_\gamma$ with $\gamma = 1,\dots,g$, we have that the profit of items packed into $G_\gamma$ is strictly more than $\epsilon p(B_0)$ and the side lengths of items packed into $G_\gamma$ are at most $\ell_{\min}(B_{\gamma-1})$ and at least $\epsilon \ell_{\min}(B_{\gamma-1})$.
		\item For $G_{g+1}$ we have that $p(G_{g+1})\leq \epsilon p(\hat{\I})$.
	\end{itemize}
	The second property implies that $g\leq 1/\epsilon$. We discard all items contained in $G_{g+1}$.To transform $B_g$ into a $\S$-box we apply Lemma~\ref{lem:vstarbox_selection}. For each $\gamma = 1,\dots,g$, we apply Lemma~\ref{lem:aux_rangebox_multi} to transform $G_\gamma$ into $1/\epsilon$ many $\N$-boxes.
	%	\awr{Can we say what $\hat{\I}}$ finally is? We omit one of the constructed sets, no?}
\end{proof}

We now use Lemmas~\ref{lem:vstarbox_selection}~-~\ref{lem:1DCubes_transform}, to show that for any dimension $d$, a $d$-dimensional box which is packed using NFDH can be split into at most $O_{\epsilon,d}(1)$ many $\N$- and $\S$-boxes.
\begin{lemma}\label{lem:Cubes_transform}
Let $B_0$ be a $d$-dimensional box and let $\hat{\I}$ be the set of items packed into $B_0$ using NFDH and $\epsilon < 1/2^{d+2}$ be given. Then, $B_0$ can be split into $O(1/\epsilon^{d+2})$ many $\S$- and $\N$-boxes such that a subset $\hat{\I}' \subseteq \hat{\I}$ can be packed into them with $p(\hat{\I} ') \geq (1-O(\epsilon))p(\hat{\I})$.
\end{lemma}
\begin{proof}
We prove this statement by induction. The base case of $d=1$ is given by Lemma~\ref{lem:1DCubes_transform}. Now, suppose that the statement is true for $d-1$ dimensions. To prove it for $d$ dimensions, we start by splitting $B_0$ into a constant number of smaller boxes such that these can later be transformed into $\S$-boxes and $\N$-boxes. This splitting procedure is visualized in Figure~\ref{fig:sq2D_splitgamma} (for $d=2$) and works as follows. Let $L_1,L_2,\dots,L_m$ be the shelves of the $d$-dimensional NFDH packing of $B_0$ and denote by $h(L_j)$ the height of shelf $L_j$ defined by the side length of the largest item packed into shelf $L_j$. By definition of NFDH, the shelves are sorted in non-increasing order of heights. Let $B_1$ be the collection of shelves of $B_0$ of height at most $\epsilon \ell_{\min}(B_0)$ and $G_1$ be the shelves of height strictly larger than $\epsilon \ell_{\min}(B_0)$. If the profit of items packed into $G_1$ is at most $\epsilon p(B_0)$, we delete all items packed into $G_1$ such that all remaining items packed into $B_0$ have side lengths at most $\epsilon \ell_{\min}(B_0)$ and proceed by transforming $B_0$ into a $\S$-box (see below). If, however, the profit of items packed into $G_1$ is more than $\epsilon p(B_0)$, we split $B_0$ into $B_1$ and $G_1$ and repeat the procedure with $B_1$. We continue the procedure (including the final deletion step) until we have a collection of boxes $B_0,B_1,\dots,B_g, G_1,\dots,G_g, G_{g+1}$ with the following properties:
\begin{itemize}
	\item All items packed into $B_g$ have side lengths at most $\epsilon \ell_{\min}(B_g)$.
	\item For each $G_\gamma$ with $\gamma = 1,\dots,g$, we have that the profit of items packed into $G_\gamma$ is strictly more than $\epsilon p(B_0)$, the height of every shelf of $G_\gamma$ is at least $\epsilon \ell_{\min}(B_{\gamma-1})$ and the side lengths of items packed into $G_\gamma$ are at most $\ell_{\min}(B_{\gamma-1})$ and at least $\epsilon \ell_{\min}(B_{\gamma-1})$ except for some items packed into the top shelf.
	\item For $G_{g+1}$ we have that $p(G_{g+1})\leq \epsilon p(\hat{\I})$.
\end{itemize}
In the following, we use the boxes $B_g, G_1,\dots,G_g$ to construct $\N$- and $\S$-boxes, while boxes $B_1,\dots,B_{g-1}$ are crucial for the analysis and all items in $G_{g+1}$ are discarded. The second property above implies that $g \leq 1/\epsilon$ and since we delete only items in the final iteration of this procedure we lose a total profit of at most $\epsilon p(B_0)$.
% 	\begin{equation}\label{eq:cubes_pl1}
	% 	\end{equation}
Therefore, it remains to show how to transform $B_g$ into a $\S$-box and how to transform each $G_\gamma$ into a constant number of $\N$- and $\S$-boxes. First, consider $B_g$ and let $\I(B_g)$ be the remaining items packed into $B_g$ after deleting all items in $G_{g+1}$.
To finalize the transformation of $B_g$ into a $\S$-box, we use Lemma~\ref{lem:vstarbox_selection} to find a subset of items $\I'(B_g) \subseteq \I(B_g)$ such that $\sum_{i \in \I'(B_g)} \lceil s_i \rceil^d_{1+\epsilon} \leq \mathrm{VOL}_d(B_g)-2\cdot d \cdot \epsilon \mathrm{VOL}_d(B_g)$ and the lost profit is at most $O(\epsilon)p(B_g)$.
% 	\begin{equation}\label{eq:cubes_pl2}
	%
	% 	\end{equation}
\begin{figure}[h!]
	\centering
	\includegraphics[scale = 0.425,page = 3]{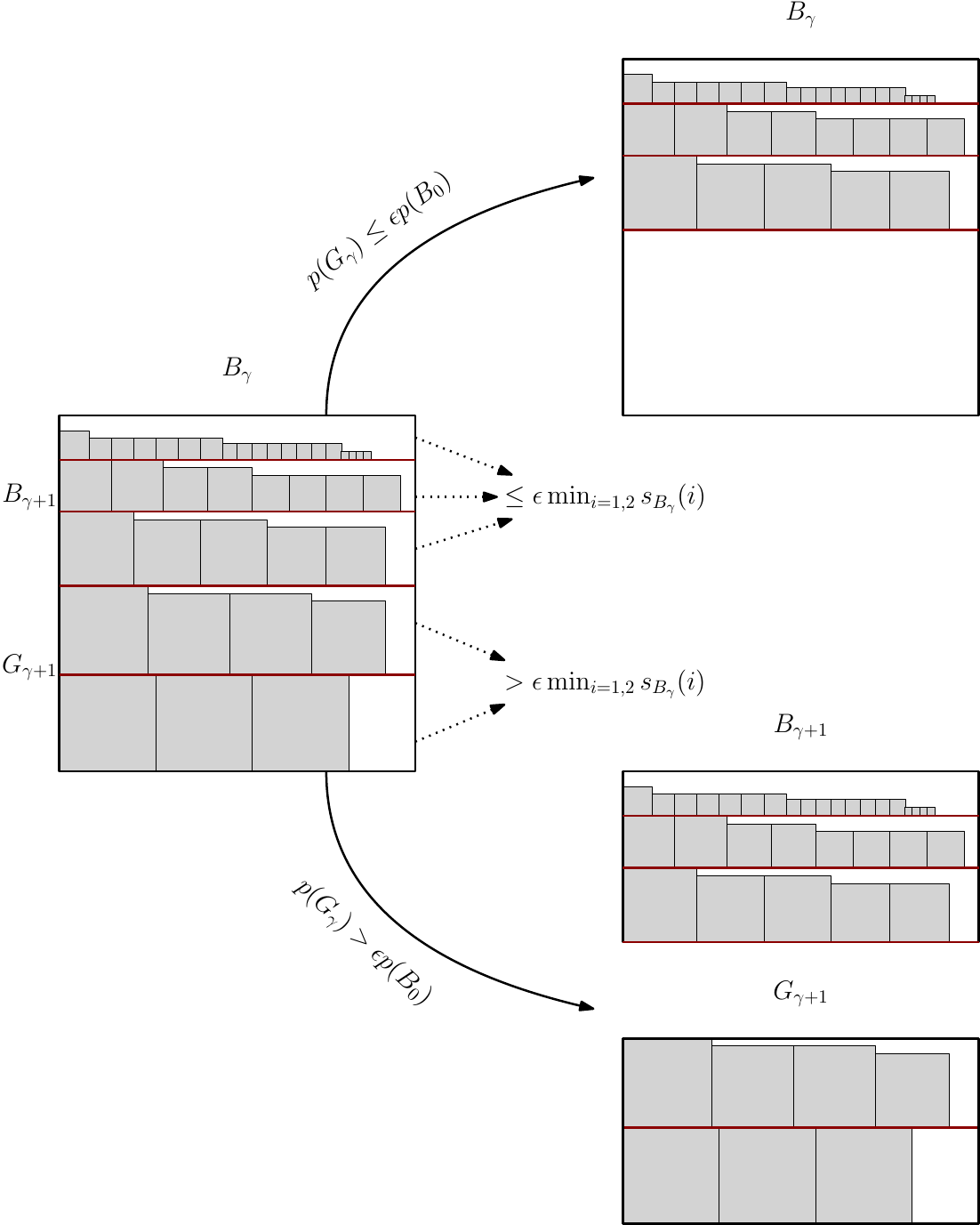}
	\caption{Iteration $\gamma$ of the splitting procedure}
	\label{fig:sq2D_splitgamma}
\end{figure}
% 	\\
% 	\\
Next, consider some box $G_\gamma$ with $\gamma =1,\dots,g$. We transform $G_\gamma$ into a constant number of $\S$- and $\N$-boxes. To do so, we first split $G_\gamma$ into at most $d$ sub-boxes as follows: The sub-box $G_\gamma(1)$ contains all shelves in dimension $d$ of $G_\gamma$ such that for every item packed into these shelves has side length at least $\epsilon \ell_{\min}(B_{\gamma-1})$ and at most $\ell_{\min}(B_{\gamma-1})$. If
$G_\gamma(1)$ contains all shelves in dimension $d$ of $G_\gamma$ we are done. Otherwise the last shelf of $G_\gamma$ is not contained in $G_\gamma(1)$. So we define $G_\gamma(2)$  as the sub-box containing all shelves in dimension $d-1$ of the last shelf in dimension $d$ of $G_\gamma$ such that every item packed into these shelves has side length at least $\epsilon^2 \ell_{\min}(B_{\gamma-1})$ and at most $\ell_{\min}(B_{\gamma-1})$.
If all other shelves in dimension $d-1$ of the last shelf of $G_\gamma$ in dimension $d$ contain items of size at most $\epsilon^2\ell_{\min}(B_{\gamma-1})$ we define an additional sub-box containing these items and are done. Otherwise, we proceed in a similar fashion in dimension $d-2$. Repeating this procedure leads to at most $r\leq d+1$ sub-boxes which we now transform into $\N$- and $\S$-boxes as follows.

We first consider box $G_\gamma(1)$ which contains items with side lengths of at least $\epsilon \ell_{\min}(B_{\gamma-1})$ and at most $\ell_{\min}(B_{\gamma-1})$. We split $G_\gamma(1)$ into at most $O_{\epsilon,1}(1)$ many $\N$-boxes losing a total profit of at most
$O(\epsilon)p(G_\gamma(1)).$
Using Lemma~\ref{lem:aux_rangebox_multi} we do the same for all sub-boxes  $G_\gamma(2),\dots,G_\gamma(r-1)$ we created for which the items have side lengths at least $\epsilon^2  \ell_{min}(B_{\gamma-1})$. For each $r' =2,\dots,r-1$, we lose a total profit of at most 
$O(\epsilon)p(G_\gamma(r')).$

Finally, consider $G_\gamma(r)$ which contains items of side lengths less than $\epsilon^2 \ell_{min}(B_{\gamma-1})$. Observe that since $G_\gamma(r)$ is a sub-box of the last shelf (in dimension $d$) packed into $G_\gamma$, we know that $\ell_d(G_\gamma(r)) \geq \epsilon \ell_{\min}(B_{\gamma-1})$. We now use an inductive argument to show that $G_\gamma(r)$ can be split into $O_{\epsilon,d}$ many $\N$- and $\S$-boxes. To this end, we ignore dimension $d$ and apply the inductive step to $G_\gamma(r)$ in dimensions $1$ to $d-1$. Let $B'$ be a resulting $\S$-box with the properties that
\begin{itemize}
	\item for every item $i \in \I(B')$, $s_i \leq \epsilon \ell_{\min}(B')$ {and}
	\item the total volume of the rounded up items packed into $B'$ is
	\[\sum_{i \in \I} \lceil s_i \rceil^{d-1}_{1+\epsilon} \leq \mathrm{VOL}_{d-1}(B')-2\cdot (d-1)\cdot \epsilon \mathrm{VOL}_{d-1}(B').\]
\end{itemize}
We now consider what happens when we take into account the $d$-th dimension of $B'$. Since $\ell_d(B') = \ell_d(G_\gamma(r)) \geq \epsilon \ell_{\min}(B_{\gamma-1})$ and every item packed into $B'$ has side length at most $\epsilon^2 \ell_{\min}(B_{\gamma-1})$, and $B'$ still satisfies the first property in $d$ dimensions. Furthermore, the total volume of the items packed into $B'$ rounded up in $d$ dimensions is at most
\[\sum_{i \in \I} \lceil s_i \rceil^{d}_{1+\epsilon} \leq (1+\epsilon)\epsilon^2 \ell_{\min}(B_{\gamma-1}) \mathrm{VOL_{d-1}(B')} \leq (1+\epsilon)\epsilon \mathrm{VOL_{d}(B')} < 2\epsilon\mathrm{VOL_{d}(B')}.\]
Since $\epsilon < 1/2^{d+2}$, this implies that 
\[\sum_{i \in \I} \lceil s_i \rceil^{d}_{1+\epsilon} < \mathrm{VOL}_d(B')-2\cdot d \cdot \epsilon \mathrm{VOL}_d(B').\]
Hence, $B'$ is a $d$-dimensional $\S$-box. 

Next, let $B''$ be a $d-1$-dimensional $\N$-box resulting from the transformation of $G_\gamma(r)$ with $s_{\max}(B'')$ being the side length of the largest item packed into $B''$. If we take $\ell_d(B'') = s_{\max}(B'')$ and $n_d(B'')= 1$, then $B''$ is a $\N$-box in $d$-dimensions. Observe that $s_{\max}(B'') < \ell_d(G_\gamma(r))$ which implies that its feasible to extend $B''$ into the $d$-th dimension this way.

Therefore, by induction we know that each of these sub-boxes can be transformed into $O(1/\epsilon^{d+1})$ many $\N$- and $\S$-boxes. In the end, this whole procedure leads to $O(1/\epsilon^{d+2})$ many $\N$- and $\S$-boxes constructed out of the original box $B_0$ with a profit of at least $(1-\epsilon)p(B_0)$.
\end{proof}

We will now prove Lemma~\ref{lem:struc_hypercubes} using Lemma~\ref{lem:Cubes_transform}.
\begin{proof}[Proof of Lemma~\ref{lem:struc_hypercubes}]
Consider an instance $\I$ and let $\epsilon < 1/2^{d+2}$. We start with the structural packing proven by Jansen et al.~\cite{jansen2022ptas} (see Theorem~\ref{thm:dD_cubes_jansen}). Let $n_b = C_{boxes}(d,\epsilon) + C_{large}(d,\epsilon)$ be the total number of boxes and large items (which we will treat as ${\mathcal{N}}$-boxes) used in this packing and denote these boxes by $B_1,B_2,\dots,B_{n_b}$. By Lemma~\ref{lem:Cubes_transform} we can transform any ${\V}$-box $B_h$ into $O_{\epsilon,d}(1)$ many $\S$- and $\N$-boxes
{and pack items from $B_h$ with a}
profit of at least $(1-O(\epsilon))p(B_h)$. Thus, the total profit is at least $(1-O(\epsilon))\sum_{i=1}^{n_b} p(B_i) \geq (1-O(\epsilon))(1-2^{d+2}\epsilon)OPT(\I) \geq (1-{2^{O(d)}}\epsilon)OPT(\I)$. {Denote by $n'_{b}$ the number of boxes after these transformations.}

To finish the proof, we still need to show that there exist values $k_1,k_2,\dots, k_r \in \mathbb{Z}_{\geq 0}$ with $r \in O_{\epsilon,d}(1)$ such that if $B$ is a $\N$-box there exists a $j \in \{1,2,\dots,r\}$ such that $s_{\max}(B) \leq k_{j}$ and $s_{\min}(B) \geq k_{j-1}$.

To this end, we are interested in the distinct values of $s_{\max}(B_i)$. Let $r \leq n'_b$ be the number of $\N$-boxes with distinct values of $s_{\max}(B_i)$. Clearly, $r \in O_{\epsilon,d}(1)$. Let $B_1,\dots,B_{n_b'}$ be $\N$-boxes with distinct $s_{\max}$-values such that $s_{\max}(B_i) > s_{\max}(B_{i+1})$ for every $i=1,\dots,r-1$. Let $k_i = s_{\max}(B_i)$ for $i=1,\dots,r$ and $k_0 = 0$.
%	\awr{maybe we want to omit $n'_b$ here and replace it with $r$?}

We can now split any $\N$-box $B$ into a constant number of $\N$-boxes such that each of these $\N$-boxes contains items of size $s_i \in (k_{j-1},k_j]$ for some $j=1,\dots, r$. To do so, let $G_j(B)$ be all shelves in dimension $d$ of $B$ containing items of size $s_i \in (k_{j-1},k_j]$ for each $j=1,\dots, r$. For each $j = 1,\dots,r$, $G_j(B)$ is a $\N$-box. Now, consider shelves in dimension $d$ of $B$ which contain items from more than one of the intervals. For such a shelf, we repeat the procedure above in dimension $d-1$ and set $n_d = 1$ for each of the resulting $\N$-boxes. Note that since we may have to repeat this procedure in each dimension and these shelves contain at most items in at most $r$ different intervals, the total number of $\N$-boxes formed from $B$ is at most $r^d$. We repeat this procedure for each original $\N$-box and end up with $O_{\epsilon,d}(1)$ many new $\N$-boxes since $r\in O_{\epsilon,d}(1)$. Hence, we have constructed a packing satisfying properties i), ii) and v). It remains to modify the packing such that it also satisfies properties iii) and iv) while maintaining the other properties. First, consider a $\V$-box $B$. We now round down $\ell_{d'}(B)$ to $\lfloor \ell_{d'}(B) \rfloor_{1+\epsilon}$. Observe, that in this way we lose at most a factor of $(1+\epsilon)^d$ of the volume of $B$.  Applying a similar LP-argument as used in the proof of Lemma~\ref{lem:vstarbox_selection}, we can find a subset of the items packed into $B$ that still can be packed into $B$ with the new side lengths such that the remaining profit is at least a factor of $(1+\epsilon)^d$ of the original profit in $B$. Applying this to all $\V$-boxes leads to a packing satisfying property iii). Next, observe that for an $\N$-box $B$ for which $n_{d'}(B) > 1/\epsilon$, rounding  $n_{d'}(B)$ to an integer power of $(1+\epsilon)$ implies that the number of items that we can pack into $B$ decreases by at most a factor of $(1+\epsilon)$. Thus, doing this in all dimensions leads to a packing satisfying property iv) while only losing a factor of  $(1+O(\epsilon))$ of the profit.
\end{proof}

\subsection{Computing a packing}
In the following, we {describe} details of Section~\ref{subsec:Computing-packing} {that were omitted above}. First, we prove that we can guess the side lengths of all $\S$-boxes in time $(\log_{1+\epsilon} n)^{O_{\epsilon,d}(1)}$. For the remainder of this subsection,  we disregard all items with profit less than $\epsilon/n \cdot p_{\max}$ while losing only an $\epsilon$-fraction of the optimal profit as $p_{\max} \leq \OPT$. We next show how to strengthen the statement of Lemma~\ref{lem:hypcub_guessing_basic} by reducing the total number of guesses to $\left( \log_{1+\epsilon}(n)\right)^{O_{\epsilon,d}(1)}$ such that all basic quantities can be guessed in time $(\log n)^{O_{\epsilon,d}(1)}$.

\begin{lemma}
Let $\B_\S$ be the set of $\S$-boxes in the packing due to Lemma~\ref{lem:struc_hypercubes}. By losing a factor $(1+\epsilon)^d$ in the profit of this packing, we can guess values $ \ell_{d'}(B) $ for all boxes $B \in \B_{{S}}$ and for all $d'=1,\dots, d$ in time $(\log_{1+\epsilon} n)^{O_{\epsilon,d}(1)}$,
	such that we can reduce the side lengths of each box $B \in \B_{{S}}$ to $ \ell_{1}(B) ,...,
	 \ell_{d}(B) $.
\end{lemma}

\begin{proof}
{It is straightforward to guess values the $ \ell_{1}(B) ,..., \ell_{d}(B) $ in time
	$(\log_{1+\epsilon} N)^{O_{\epsilon,d}(1)}$. To decrease the number of guesses to $(\log_{1+\epsilon} n)^{O_{\epsilon,d}(1)}$}, we will make use of the following useful fact{:} let $B$ be a $\S$-box, then we know that $\ell_{\min}(B) \geq \frac{1}{\epsilon}s_{\max}(\I(B))$.

{Now, let $B^*$ be an $\S$-box in $\B_{\S}$ such that $s_{\max}(\I(B^*)) = \max_{B \in \B_\S}\{s_{\max}(\I(B))\}$. We will now find the side length of the largest item packed into $B^*$ which will give us an estimate on the side lengths of $B^*$. We do this by first guessing the profit class of this item. By our preprocessing step of the profits, this can be done in time $O(\log_{1+\epsilon}n)$. Once we have the profit class of this item, we do the following. Losing only a factor of $1+\epsilon$, we can guess the number of items of this profit class in our structured packing in time $O(\log_{1+\epsilon}n)$. To be more precise, we guess how many items of this profit class are packed into each of the boxes. Let $t$ be the guessed profit class and let $n_{t,B}$ denote the guessed number of items of class $t$ packed into box $B$. We can find all values $n_{t,B}$ in time {$(\log_{1+\epsilon} n)^{O_{\epsilon,d}(1)}$}. We define $n_t := \sum_{B \in \B}n_{t,B}$ and using our item data structure and a binary search over all possible side lengths using the balanced binary search tree data structure which stores the side lengths, we find the $n_t$-th smallest item of profit class $t$ in time $O(\log^3 n)$. For this item, we guess the box in which it must be packed in time $O_{\epsilon,d}(1)$. We know that it can only be packed into an $\N$-box or $B^*$ since using that $\ell_{min}(B) \geq \frac{1}{\epsilon}s_{\max}(B)$ we have a contradiction for the definition of $B^*$. Hence, if the guess is $B^*$ {we are done since we obtain a value} $\hat{s}$ giving the size of the largest item packed into $B^*$. If the guessed box $B$, however, is an $\N$-box, then w.l.o.g. we may assume that this item and {also} the $n_{t,B}-1$ smaller items {of the same profit class} are all packed into $B$ and we update the value of $n_t$ to $n_t-n_{t,B}$. Repeating this procedure leads to a total amount of $O_{\epsilon,d}(1)$ iterations since there are at most $O_{\epsilon,d}(1)$ many $\N$-boxes. Thus, we can find the value $s_{\max}(B^*)$ in time {$(\log_{1+\epsilon} n)^{O_{\epsilon,d}(1)}$}. We know that $\ell_{d'}(B^*) \in [\frac{1}{\epsilon}s_{\max}(B^*),ns_{\max}(B^*)]$ for all $d' = 1,\dots,d$. Thus, we can guess all values $ \ell_{d'}(B^*)$ in time {$(\log_{1+\epsilon} n)^{O_{\epsilon,d}(1)}$}.
}

Next, we argue that for any other $\S$-box $B$ we either have $s_{\max}(B) \geq \frac{s_{\max}(B^*)}{n^{1/d}}$ or all items inside $B$ can be moved to $B^*$,
{while losing only a factor of $1+\epsilon$ in the profit of the items packed inside $B^*$}.
Consider all $\S$-boxes in $\B$ such that for each of them the largest item packed into them has size strictly less than
\[
\frac{s_{\max}(B^*)}{n^{1/d}}.
\]
Since there are at most $n$ such items and by definition of $s_{\max}(B^*)$, we know that their total volume is at most
\[
n \cdot \left(\frac{s_{\max}(B^*)}{n^{1/d}}\right)^d \leq n \cdot \left(\frac{\epsilon \ell_{\min}(B^*)}{n^{1/d}}\right)^d \leq \epsilon^d \mathrm{VOL}_d(B^*).
\]
Thus, by {reserving this volume in $B^*$ to pack all these items, we may assume that there are no such boxes.}
%allowing $B^*$ to contain an additive volume of $\epsilon^d \mathrm{VOL}_d(B^*)$, we may assume that there are no such boxes. 
From this it follows that for any other $\S$-box $B$, we have
\[
s_{\max}(B^*) \geq s_{\max}(B) \geq \frac{s_{\max}(B^*)}{n^{1/d}}.
\]
%We now prove the following upper bound for $s_{\max}(B)$
%\[
%s_{\max}(B) \leq \epsilon \cdot n \cdot s_{\max}(B^*).
%\]
%Suppose, towards contradiction that this upper bound does not hold. Then, we know
%\[
%\ell_{\min}(B) \geq \frac{s_{\max}(B)}{\epsilon} > \frac{\epsilon \cdot n \cdot s_{\max}(B^*)}{\epsilon} = n \cdot s_{\max}(B^*) \geq \ell_{\min}(B*).
%\]
%By choice of $B^*$, this yields a contradiction. Therefore, we have
%\[
%\frac{s_{\max}(B^*)}{n^{1/d}} \leq s_{\max}(B) \leq  \epsilon \cdot n \cdot s_{\max}(B^*).
%\]
This implies that for each $d'=1,\dots,d$ we have
\[
\frac{s_{\max}(B^*)}{\epsilon \cdot n^{1/d}}\leq \ell_{d'}(B) \leq \epsilon \cdot n^2\cdot s_{\max}(B*).
\]
Therefore, we can guess the value $ \ell_{d'}(B)$ using at most $O(\log_{1+\epsilon}n)$ guesses. Thus, we can guess the side lengths of all $\S$-boxes in time $\log_{1+\epsilon}^{O_{\epsilon,d}(1)}n$ losing only a factor of $(1-O(\epsilon))$ of the profit.

It remains to be shown that by reducing the side lengths of each box $B \in \B_{\S}$ to $ \ell_{1}(B) ,..., \ell_{d}(B) $ we only lose a factor of $(1+\epsilon)^d$. Hereto, consider some box $B \in \B_\S$ with the set of items $\I(B)$ packed into it. We know that by reducing the side lengths as described above we reduce the volume by at most $(1+\epsilon)^{-d}$. We now show that there exists a subset of $\I(B)$ such that the total profit of the items is at least
	\[
	(1+\epsilon)^{-d}p(\I(B)).
	\]
	We know that for $B$ and $\I(B)$ we have that $s_i \leq \epsilon \ell_{\min}(B)$ for all $i \in \I(B)$ and furthermore
	\[\sum_{i\in\I(B)}\lceil s_{i}\rceil_{1+\epsilon}^{d}\leq(1-2d\cdot\epsilon)\mathrm{VOL}(B).\]
	We will now use a similar argument as in the proof of Lemma~\ref{lem:vstarbox_selection} considering the following LP:

	\begin{alignat*}{3}
		& \text{minimize} & \sum_{j=1}^{m} &x_ip_i& \\
		& \text{subject to} \quad& \sum_{i \in \Q}&x_{i}\lceil s_i \rceil_{1+\epsilon}^d& \leq (1+\epsilon)^{-d}(1-2d\cdot \epsilon) \mathrm{VOL}_d(B) \\
		&&& x_{i} \geq 0, &  \forall i \in \I(B)
		\\
		&&& x_{i} \leq 1, &  \forall i \in \I(B)
	\end{alignat*}

	Consider the following solution $\hat{x}_i := (1+\epsilon)^{-d}$. Clearly, this solution is feasible. Furthermore, observe that $s_i \leq \epsilon  \ell_{\min}(B) $ and $\lceil s_i \rceil_{1+\epsilon}^d \leq 2\epsilon \prod_{d'=1}^d  \ell_{\min}(B) $. Let $x^*$ be the optimal solution to the LP. Then, by the rank lemma~\cite{lau2011iterative}, we know that there is at most one item $i' \in \I(B)$ for which $0 < x^*_i < 1$. We now define the set of items packed into $B$ with reduced side lengths as $\I'(B) :=\{i\in \I(B): x^*_i > 0\}$. By the candidate solution $\hat{x}$ we know that the obtained profit is at least $(1+\epsilon)^{-d}p(\I(B))$. Furthermore, we have the following
	\[\sum_{i\in\I(B')}\lceil s_{i}\rceil_{1+\epsilon}^{d}\leq(1+\epsilon)^{-d}(1-2d\cdot \epsilon) \mathrm{VOL}_d(B) + 2\epsilon\prod_{d'=1}^d  \ell_{\min}(B) .\]
	By Lemma~\ref{lem:NFDH}, we know that $\I'(B)$ can be packed into $B$ with reduced side lengths using NFDH.
\end{proof}

The proofs of Lemmas~\ref{lem:hypcub_guessing_pjs} and~\ref{lem:hypcub_guessing_abjs} follow from the arguments in Section~\ref{sec:hypercubes}. Therefore, all basic quantities can be guessed in time $(\log n)^{O_{\epsilon,d}(1)}$.

We will now move on to the details of our indirect guessing framework. Recall that for each $\ell \in \{1,...,r\}$, we are {concerned} with finding an estimate of $k_{\ell+1}$, denoted by $\tilde{k}_{\ell+1}$, assuming that we have already found {estimates} $\tilde{k}_{1},\dots, \tilde{k}_{\ell}$.
Consider a candidate $s$ for $\tilde{k}_{\ell+1}$.
We will first show how to reduce the number of variables considered in $(\mathrm{IP}(s))$ by only losing a factor of $1+\epsilon$ of the profit. To this end, we denote by $(\mathrm{LP}(s))$ the linear relaxation of $(\mathrm{IP}(s))$.

\begin{lemma}\label{lem:hypcub_reducevar}
Given a set of guessed boxes $\B$ and the set of input items $\tilde{\I}_{\ell+1}(s)$, there exists a subset of items of $\tilde{\I}_{\ell+1}(s)$ such that there are at most $O_\epsilon(\log n)$ many size and profit classes. The optimal solution to $(\mathrm{LP}(s))$ considering only this subset yields a profit of at least {a $(1-\epsilon)$-fraction} of the profit of the optimal solution to $(\mathrm{LP}(s))$ considering all items in $\tilde{\I}_{\ell+1}(s)$.
\end{lemma}
\begin{proof}%[Proof of Lemma~\ref{lem:proft_rounding}]
Recall that we disregarded all items with profit less than $\epsilon/n \cdot p_{\max}$.
% while losing only an $\epsilon$-fraction of the profit of the optimal solution to LP(s).
This implies that for all $i \in \tilde{\I}_{\ell+1}(s)$, we have that $p_i \in [\epsilon/n \cdot p_{\max}, p_{\max}]$ and {there are only} $O_\epsilon(\log n)$ profit classes.
In order to restrict the number of size classes consider the following two linear programs. The first one is $(\mathrm{LP}(s))$ and the second one is $(\mathrm{LP}(s))$ after disregarding every item $i \in \tilde{\I}_{\ell+1}(s)$
{that is \emph{small}, i.e., for which}
$s_i \leq \frac{\epsilon^{1/d} \mathrm{v}^{1/d}}{n^{1/d}}$ where $\mathrm{v}$ is the maximum volume of the guessed $\S$-boxes. We denote the latter linear program by $(\mathrm{LP}'(s))$. We will argue later that we can select all small items.

\begin{alignat*}{3}
	(\mathrm{LP}(s))\quad& \text{max} 	& \displaystyle \sum_{(t,t') \in \mathcal{T}}\sum_{B \in \B(\ell+1)} x_{t,t',B} p(t')			& 			& \quad & \\
	& \text{s.t.} & \displaystyle\sum_{(t,t') \in \mathcal{T}} x_{t,t',B} 	& \leq n(B)	& 		& \forall B \in \B_{\N}(\ell+1) \\
	& & \displaystyle\sum_{(t,t') \in \mathcal{T}} x_{t,t',B}s(t)^d	& \leq a_{B,\ell+1}\mathrm{VOL}_d(B)	& 		& \forall B \in \B_{\S}(\ell+1) \\
	&				& \displaystyle\sum_{B \in \B(\ell+1)} x_{t,t',B}								& \leq n_{t,t'}	& 		& \forall (t,t') \in \mathcal{T} \\
	&				& x_{t,t',B}								& \geq 0& 		&\forall (t,t') \in \mathcal{T}, B \in \B(\ell+1)
\end{alignat*}

\begin{alignat*}{3}
	(\mathrm{LP}'(s))\quad& \text{max} 	& \displaystyle \sum_{(t,t') \in \mathcal{T}}\sum_{B \in \B(\ell+1)} x_{t,t',B} p(t')			& 			& \quad & \\
	& \text{s.t.} & \displaystyle\sum_{(t,t') \in \mathcal{T}} x_{t,t',B} 	& \leq n(B)	& 		& \forall B \in \B_{\N}(\ell+1) \\
	& & \displaystyle\sum_{(t,t') \in \mathcal{T}} x_{t,t',B}s(t)^d	& \leq a_{B,\ell+1}(1-\epsilon)\mathrm{VOL}_d(B)	& 		& \forall B \in \B_{\S}(\ell+1) \\
	&				& \displaystyle\sum_{B \in \B(\ell+1)} x_{t,t',B}								& \leq n_{t,t'}	& 		& \forall (t,t') \in \mathcal{T} \\
	&				& x_{t,t',B}								& \geq 0& 		&\forall (t,t') \in \mathcal{T}, B \in \B(\ell+1)
\end{alignat*}

Let $x^*$ be an optimal solution to $(\mathrm{LP}(s))$, then $\hat{x}:= (1-\epsilon)x^*$ is a feasible solution to $(\mathrm{LP}'(s))$ with a profit of $(1-\epsilon)$ times the profit of $x^*$. Adding the small items will only increase the profit. We can add all small items since their total volume will be at most $\epsilon \mathrm{v}$ such that they can be packed into the $\S$-box with largest volume.
\end{proof}
We continue by presenting the proof of Lemma~\ref{lem:cubes_IP_sol}.
\begin{proof}[Proof of Lemma~\ref{lem:cubes_IP_sol}]
We first show how to find a $(1-\epsilon)$-approximate integral solution to $(\mathrm{IP}(s))$. Recall that $(\mathrm{IP}(s))$ is defined as follows.

\begin{alignat*}{3}
	(\mathrm{IP}(s))\quad& \text{max} 	& \displaystyle \sum_{(t,t') \in \mathcal{T}}\sum_{B \in \B(\ell+1)} x_{t,t',B} p(t')			& 			& \quad & \\
	& \text{s.t.} & \displaystyle\sum_{(t,t') \in \mathcal{T}} x_{t,t',B} 	& \leq n(B)	& 		& \forall B \in \B_{\N}(\ell+1) \\
	& & \displaystyle\sum_{(t,t') \in \mathcal{T}} x_{t,t',B}s(t)^d	& \leq a_{B,\ell+1}\mathrm{VOL}_d(B)	& 		& \forall B \in \B_{\S}(\ell+1) \\
	&				& \displaystyle\sum_{B \in \B(\ell+1)} x_{t,t',B}								& \leq n_{t,t'}	& 		& \forall (t,t') \in \mathcal{T} \\
	&				& x_{t,t',B}								& \in \mathbb{N}_{0}& 		&\forall (t,t') \in \mathcal{T}, B \in \B(\ell+1)
\end{alignat*}

We will now compute an approximate solution to $(\mathrm{IP}(s))$. We start by guessing the $2C_{boxes}(\epsilon,d)/\epsilon$ most profitable items in the solution to $(\mathrm{LP}'(s))$. Denote by $S_g$ the set of these items.
	% this solution by $S_g$.
	We obtain $S_g$ by guessing for each of the most profitable items the profit class it belongs to and then choose an item from the smallest size class for which at least one item of this profit class exists. Repeating this for each of the most profitable items gives a total of
	{$(\log n)^{O_{\epsilon,d}(1)}$}
	% $O_{\epsilon,d}(\log n)$
	guesses. Furthermore, we need to guess the correct box for each of these items which can be done in time $O_{\epsilon,d}(1)$. After obtaining $S_g$ in this way, we update the values of the right-hand side of the constraints in $(\mathrm{IP}(s))$. Hereto, consider a box $B \in \B$. If $B$ is a $\N$-box, we denote by $n^g(B)$ the number of guessed items packed into $B$. If $B$ is a $\S$-box, we denote by $\mathrm{VOL}_d^g(B)$ the total volume of the guessed items packed into $B$. Furthermore, for each pair $t,t'$, denote by $n_{t,t'}^g$ the number of items of size class $\S_t$ and profit class $\P_{t'}$ in $S_g$. Based on this we consider the following variant of $(\mathrm{IP}(s))$ which we denote by $(\mathrm{IP}_g(s))$.
	
	{
	\small
	\begin{alignat*}{3}
		(\mathrm{IP}_g(s))\quad& \text{max} 	& \displaystyle \sum_{(t,t') \in \mathcal{T}}\sum_{B \in \B(\ell+1)} x_{t,t',B} p(t')			& 			& \quad & \\
		& \text{s.t.} & \displaystyle\sum_{(t,t') \in \mathcal{T}} x_{t,t',B} 	& \leq n(B)-n^g(B)	& 		& \forall B \in \B_{\N}(\ell+1) \\
		& & \displaystyle\sum_{(t,t') \in \mathcal{T}} x_{t,t',B}s(t)^d	& \leq a_{B,\ell+1}\mathrm{VOL}_d(B) - \mathrm{VOL}_d^g(B)	& 		& \forall B \in \B_{\S}(\ell+1) \\
		&				& \displaystyle\sum_{B \in \B(\ell+1)} x_{t,t',B}								& \leq n_{t,t'} - n_{t,t'} ^g	& 		& \forall (t,t') \in \mathcal{T} \\
		&				& x_{t,t',B}								& \in \mathbb{N}_{0}& 		&\forall (t,t') \in \mathcal{T}, B \in \B(\ell+1)
	\end{alignat*}}

	The objective is now to find an approximate solution to $(\mathrm{IP}_g(s))$ via its LP-relaxation. 
	
	{
		\small
	\begin{alignat*}{3}
	(\mathrm{LP}_g(s))\quad& \text{max} 	& \displaystyle \sum_{(t,t') \in \mathcal{T}}\sum_{B \in \B(\ell+1)} x_{t,t',B} p(t')			& 			& \quad & \\
	& \text{s.t.} & \displaystyle\sum_{(t,t') \in \mathcal{T}} x_{t,t',B} 	& \leq n(B)-n^g(B)	& 		& \forall B \in \B_{\N}(\ell+1) \\
	& & \displaystyle\sum_{(t,t') \in \mathcal{T}} x_{t,t',B}s(t)^d	& \leq a_{B,\ell+1}\mathrm{VOL}_d(B) - \mathrm{VOL}_d^g(B)	& 		& \forall B \in \B_{\S}(\ell+1) \\
	&				& \displaystyle\sum_{B \in \B(\ell+1)} x_{t,t',B}								& \leq n_{t,t'} - n_{t,t'} ^g	& 		& \forall (t,t') \in \mathcal{T} \\
	&				& x_{t,t',B}								& \geq 0& 		&\forall (t,t') \in \mathcal{T}, B \in \B(\ell+1)
\end{alignat*}
}

	Observe, that we cannot solve $(\mathrm{LP}_g(s))$ in time $\mathrm{poly}(\log n)$. However, we can apply Lemma~\ref{lem:hypcub_reducevar} to $(\mathrm{LP}_g(s))$ and restrict the number of size classes to be considered. This gives the following LP denoted by $(\mathrm{LP}'_g(s))$. Here, $\mathcal{T}$ is the set of pairs $(t,t')$ considering only the size classes that are left after applying Lemma~\ref{lem:hypcub_reducevar}. By our assumption that we we only consider items with profit at least $\frac{\epsilon p_{\max}}{n}$, we have that $|\mathcal{T}| \in O(\log^2_{1+\epsilon}n)$. Observe that an optimal fractional solution to $(\mathrm{LP}'_g(s))$ yields at least $(1-\epsilon)$ times an optimal fractional solution to $(\mathrm{LP}_g(s))$.
	
	\begin{alignat*}{3}
	(\mathrm{LP}'_g(s))\quad& \text{max} 	& \displaystyle \sum_{(t,t') \in \mathcal{T}'}\sum_{B \in \B(\ell+1)} x_{t,t',B} p(t')			& 			& \quad & \\
	& \text{s.t.} & \displaystyle\sum_{(t,t') \in \mathcal{T}'} x_{t,t',B} 	& \leq \leq n(B)-n^g(B)	& 		& \forall B \in \B_{\N}(\ell+1) \\
	& & \displaystyle\sum_{(t,t') \in \mathcal{T}'} x_{t,t',B}s(t)^d	& \leq a_{B,\ell+1}\mathrm{VOL}_d(B) - \mathrm{VOL}_d^g(B)	& 		& \forall B \in \B_{\S}(\ell+1) \\
	&				& \displaystyle\sum_{B \in \B(\ell+1)} x_{t,t',B}								& \leq n_{t,t'} - n_{t,t'} ^g	& 		& \forall (t,t') \in \mathcal{T}' \\
	&				& x_{t,t',B}								& \geq 0& 		&\forall (t,t') \in \mathcal{T}', B \in \B(\ell+1)
\end{alignat*}

	Now, we obtain an optimal fractional solution $S_f$ to $(\mathrm{LP}'_g(s))$ in time $(\log_{1+\epsilon}(n))^{O(1)}$ (using the ellipsoid method or for example~\cite{cohen2021solving}). Next, we argue that there are at most $2C_{boxes}(\epsilon,d)$ many fractional non-zero variables.  By the rank lemma (see e.g.~\cite{lau2011iterative}) we know for an optimal vertex solution $\tilde{x}$ of $(\mathrm{LP}'(s))$, the number of fractional non-zeros in $\tilde{x}$ is upper bounded by the number of linearly independent tight constraints of $(\mathrm{LP}'(s))$. We now argue that there can be at most $2C_{boxes}(\epsilon,d)$ such constraints out of the   $O_{\epsilon,d}(\log^2 n)$ many. To this end, consider the third group of constraints
	\begin{align*}
		\displaystyle\sum_{B \in \B} x_{t,t',B}								& \leq n_{t,t'}	& 		& \forall  t\in \{0,\dots,\lceil \log_{1+\epsilon}(s)\rceil\}, t'\in \{0,\dots, \log n\}
	\end{align*}
	Suppose there is at least one tight constraint of the third type. Observe, it holds that $n_{t,t'} \in \mathbb{N}$ and, therefore, if there is a fractional variable in a {tight} constraint of this type in an optimal vertex solution, there has to be at least one other fractional variable in this constraint such that their sum is an integer. Assume there are $k$ fractional variables in this constraint. Then, there must be $k-1$ constraints of either the first or second type that are tight. However, the number of constraints of the first two types is at most $C_{boxes}(\epsilon,d)$. Therefore, the number of fractional variables is at most $2C_{boxes}(\epsilon,d)$. 
	We now obtain a $(1+\epsilon)$-approximate fractional solution to $(\mathrm{IP}(s))$ by combining $S_g$ and all integer variables of $S_f$ together with all small items. Since $S_g$ contains $2C_{boxes}(\epsilon,d)/\epsilon$ many integer variables and $S_f$ contains at most $2C_{boxes}(\epsilon,d)$ many fractional variables each having a smaller profit than the ones of $S_g$, we lose a total profit of at most $\epsilon p(S_{g})$.

	Finally, we need to know the total profit of the small items with $s_i \leq \frac{\epsilon^{1/d} \mathrm{v}^{1/d}}{n^{1/d}}$
	{where $\mathrm{v}$ is the maximum volume of the guessed $\S$-boxes}.
	By Lemma~\ref{lem:data-structure}, this can be done in time $O_{\epsilon}(\log_{1+\epsilon}^2 n)$ with one query for each profit class which counts the number of small items of this profit class. We can use this and the rounded profit of each profit class to compute a $(1+\epsilon)$-estimate of the total profit of these items.
	Finally, consider two values $s\leq s'$. Then, the approximate solution obtained to $(\mathrm{IP}(s))$ is also feasible for $(\mathrm{IP}(s'))$. Therefore, $q(s) \leq q(s')$.
\end{proof}
Next, we will prove Lemma~\ref{lem:cubes_induc_kr}.
\begin{proof}[Proof of Lemma~\ref{lem:cubes_induc_kr}]
Since we solve $(\mathrm{IP}(s))$ up to a factor of $1-\epsilon$,
we have that $q(k_{\ell+1})\ge (1-\epsilon)\hat{p}(\ell+1)$
since $\tilde{k}_{\ell}\le k_{\ell}$ and, therefore, $\tilde{\I}_{\ell+1}(k_{\ell})=\left\{ i\in\I:s_{i}\in[\tilde{k}_{\ell},k_{\ell})\right\} \subseteq\left\{ i\in\I:s_{i}\in[k_{\ell},k_{\ell})\right\} =\I_{\ell+1}$. Recall that we define $\tilde{k}_{\ell+1}$ to be the smallest value $s\in S$ for which $q(s)\ge (1-\epsilon)\hat{p}(\ell+1)$. Since $q(k_{\ell+1})\ge (1-\epsilon)\hat{p}(\ell+1)$ it must hold that  $\tilde{k}_{\ell+1}\le k_{\ell+1}$.
\end{proof}
Finally, we present the proof of Theorem~\ref{thm:cubes_ptas}.
\begin{proof}[Proof of Theorem~\ref{thm:cubes_ptas}]
% We set our accuracy to $(O_{\epsilon,d}(1))^{-1}\epsilon$ such that the structured packing achieves a profit of at least $(1-\epsilon)OPT(\I)$.
%r{commented out sentence (I did not know which ``accuracy'' it refers to}
Our approximation scheme proceeds in three stages:
\begin{enumerate}[(A)]
	\item  \textit{Guessing basic quantities:} The total number of guesses is $(\log n)^{O_{\epsilon,d}(1)}$. {Observe that due to our method of guessing the basic quantitities for each $\N$-box, the number of items we can pack into it is at least $(1+\epsilon)^{-d}$ times the number of items packed into it if these quantities are guessed exactly. Thus, our packing of $\N$-boxes is a $(1+O(\epsilon))$-approximate packing.}
	\item \textit{Indirect guesing framework:} For each guess we need $r \in O_{\epsilon,d}(1)$ iterations of the indirect guessing framework, leading to solutions $x^*(1),\dots,x^*(r)$ to the LP-relaxations of $(\mathrm{IP}(\tilde{k}_{1})),...,(\mathrm{IP}(\tilde{k}_{r}))$ {after guessing the mentioned most profitable items}. This takes time
	% $O_{\epsilon,d}(\log n)$.
	{$(\log n)^{O_{\epsilon,d}(1)}$}.
	Let $\hat{x}(1),\dots,\hat{x}(r)$ be the rounded solutions to $(\mathrm{IP}(\tilde{k}_{1})),...,(\mathrm{IP}(\tilde{k}_{r}))$
	\item \textit{Constructing final solution:} Observe that it is not relevant which item is assigned to which box but rather that each box is assigned the correct number of items of each combination of size class $\Q_t$ and profit class $\P_{t'}$ and we {obtain} at least the total profit of the combined solutions $\hat{x}(1),\dots,\hat{x}(r)$. For each combination of size class $\Q_t$ and profit class $\P_{t'}$ we use Lemma~\ref{lem:data-structure} to find the set of items with sizes of class $\Q_t$ and profits of class $\P_{t'}$ in time $O(n)$. Assigning the items one-by-one to boxes in $\B$ in non-decreasing order of side lengths can be done in time $O(n\log n)$. By making sure that each box $B$ receives the correct number of items for each pair  size class $\Q_t$ and profit class $\P_{t'}$ this guarantees a profit of at least {an $(1+\epsilon)^{-1}$-fraction} of the profit due {to} $\hat{x}(1),\dots,\hat{x}(r)$. Finally, it remains to find a packing of {the} items into boxes. Each $\N$-box $B$ can be filled one-by-one in time $O(n)$ placing each item into one of the cells of the grid of $B$. For the $\S$-boxes, we use NFDH which packs all of them in time $O(n\log n)$~\cite{bansal2006bin}. Due to the way we guess the basic quantities of our $\S$-boxes we know that for any $\S$-box $B$ we have that all items packed into it have side length at most $2\epsilon \ell_{\min}(B)$ and the total volume of the items rounded up is at most $(1-2d\epsilon)\mathrm{VOL}_d(B)+\epsilon \mathrm{VOL}_d(B)$. Hence, we know that by Lemma~\ref{lem:NFDH} all selected items can be packed into $B$
	{by NFDH}.
\end{enumerate}

This yields a total running time of $O(n \log^2 n) + (\log n)^{O_{\epsilon,d}(1)}$, where the first term is due to insertion of all items into the item data structure and the balanced binary search trees. By the nature of our guessing and Lemma~\ref{lem:struc_rectangles}, we know that the guessed boxes guarantee a profit of at least $(1-2^{O(d)}\epsilon)OPT$. The indirect guessing scheme guarantees that for each set $\I_j$ (with $j=1,\dots,r$) we find a set of items yielding a profit of at least $(1-\epsilon)\hat{p}_j$ and by definition of $\N$- and $\S$-boxes all of these items can be packed.
\end{proof}
\subsection{Dynamic Algorithm}\label{apx:cubes_dynamic}
In this section we present the proof of Theorem~\ref{thm:dyn_alg_cubes}.
\begin{proof}[Proof of Theorem~\ref{thm:dyn_alg_cubes}]
If an item $i$
is added to $\I$ or removed from $\I$, then we update our item data
structure in time $O(\log^2 n)$ (see Lemma~\ref{lem:data-structure}). Additionally, insertion or deletion of the corresponding side length from the balanced binary search tree takes $O(\log n)$~\cite{guibas1978dichromatic}.
If an $(1+\epsilon)$-estimate for $\OPT$ is queried, we execute
the indirect guessing framework above \emph{without }computing the
precise set of tasks in the computed solution. Instead, we compute
only the solutions $x^*(1),\dots,x^*(r)$ to the respective instance of $(\mathrm{LP}(s))$ {(after guessing
	the most profitable items of $(\mathrm{IP}(s))$)}.
We can do this in time $(\log n)^{O_{\epsilon,d}(1)}$. This suffices
to estimate the profit of the computed solutions, and hence the profit
of the optimal solution, up to a factor $(1+\epsilon)$. 
Suppose that for an item $i\in\I$ it is queried whether $i$ is contained
in the current solution. We compute the values $\tilde{k}_{1},\tilde{k}_{2},\dots,\tilde{k}_{r}$
and the corresponding solutions to $(\mathrm{IP}(\tilde{k}_{1})),...,(\mathrm{IP}(\tilde{k}_{r}))$
in time $(\log n)^{O_{\epsilon,d}(1)}$. Based on them, we specify implicitly
a fixed solution $\ALG$ for which we answer whether $i\in\ALG$ or
$i\notin\ALG$. We do this as follows. For each combination of a set
$\tilde{\I}_{j}$, a size class $\Q_{t}$, and a profit class $\P_{t'}$
we consider the set $\tilde{\I}_{j}\cap \Q_{t}\cap\P_{t'}$. Based
on the rounded LP-solutions, we define a value $z_{j,t,t'}\in\mathbb{N}_{0}$
that defines how many items from $\tilde{\I}_{j}\cap\Q_{t}\cap\P_{t'}$
are contained in $\ALG$. Using the way we construct our packing, we know that we picked the items in non-decreasing order of side lengths. For an item $i\in\tilde{\I}_{j}\cap\Q_{t}\cap\P_{t'}$,
we output ``$i\in\ALG$'' if $i$ is among the first $z_{j,t,t'}$
items, and ``$i\notin\ALG$'' otherwise. This can be done using a single query in time $O(\log^2 n)$ to count the number of items in $\tilde{\I}_{j}\cap \Q_{t}\cap\P_{t'}$ with side length smaller than $s_i$. This ensures that we give
consistent answers between two consecutive updates of the set $\I$.
Finally, once all values $z_{j,t,t'}$ are computed as described above,
we can easily output the whole solution $\ALG$ in time $|\ALG|\cdot(\log n) +(\log n)^{O_{\epsilon}(1)}$
if this is desired.
\end{proof}

\section{Details of Section~\ref{sec:rec_main}}\label{app:rec}
In this section we present the details of our (dynamic) algorithms for the two-dimensional geometric knapsack problem with rectangles. Recall that we are given a set $\I$.  For an item $i\in\I$, denote by $h_{i}$ its height, by $w_{i}$ is width, by $p_i$ its profit and by $d_i := p_i/w_i$ its density. Our knapsack $K$ has side length $N$ in each dimension for a given integer $N$. We assume for the remainder that we are given a fixed constant $\epsilon > 0$. Again, we
store our items $\I$ in a suitable data structure based on range counting/reporting data structures for points in three dimensions (corresponding to item heights, widths, profits and densities)~\cite{lee1984computational}. We refer to it
% call this
as our \emph{rectangle data structure}. See Appendix~\ref{app:data_struc} for details.
% which we will refer
%to as our \emph{rectangle data structure} based on range counting\aw{/reporting} data structures
%for points in three dimensions (\aw{corresponding to} item heights, widths and profits)~\cite{lee1984computational}.
% This data structure allows the following operations based on common techniques (see e.g.~\cite{lee1984computational}).
\begin{lemma}
	\label{lem:data-structure-rec}There is a data structure for the items
	$\I$ that allows the following operations:
	\begin{itemize}
		\item Insertion and deletion of an item in time $O(\log^4 n)$.
		\item Given eight values $a,b,c,d,e,f,g,h \in \mathbb{N}$, return the cardinality of the set $\I(a,b,c,d,e,f,g, h):= \{i \in I: a\leq h_i \leq b, c\leq w_i \leq d, e\leq p_i \leq f,g\leq d_i \leq h\}$ in time $O(\log^{3}n)$.
		\item Given eight values $a,b,c,d,e,f,g,h \in \mathbb{N}$, return the total profit of items in the set $\I(a,b,c,d,e,f,g, h):= \{i \in I: a\leq h_i \leq b, c\leq w_i \leq d, e\leq p_i \leq f,g\leq d_i \leq h\}$ in time $O(\log^{3}n)$.
		\item Given eight values $a,b,c,d,e,f,g,h \in \mathbb{N}$, return the total width of items in the set $\I(a,b,c,d,e,f,g, h):= \{i \in I: a\leq h_i \leq b, c\leq w_i \leq d, e\leq p_i \leq f,g\leq d_i \leq h\}$ in time $O(\log^{3}n)$.
		\item Given eight values $a,b,c,d,e,f,g,h \in \mathbb{N}$,, return the set $\I(a,b,c,d,e,f,g, h):= \{i \in I: a\leq h_i \leq b, c\leq w_i \leq d, e\leq p_i \leq f,g\leq d_i \leq h\}$ in time $O(\log^3 n +|\I(a,b,c,d,e,f,g,h)|)$. 
		%			\item Given four values $a,b,c,d \in \mathbb{N}$, the total profit of all items in the set $\I(a,b,c,d):= \{i \in I: a\leq s_i \leq b, c\leq p_i \leq d\}$ can be reported in time $O(\log^{3}n)$.
	\end{itemize}
\end{lemma}
Furthermore, we use balanced binary search trees (see e.g.~\cite{guibas1978dichromatic}) to store the set of item densities, profits, widths and heights such that an item can be inserted, deleted, and queried in time $O(\log n)$ in each of these additional data structures.% allowing for $O(\log n)$ insertion, deletion and lookups.}

We distinguish our items $\I$ into four types, depending on constants
$\el$ and $\es$ with $1\ge\epsilon\ge\el\ge\es>0$
such that $c(\epsilon)\es<\el$
for a value $c(\epsilon)$ (which we will define later).
% , depending only on $\epsilon$, that we
% will define later.
We say that an item $i\in\I$ is
\begin{itemize}
\item \emph{large }if $h(i)>\el N$ and $w(i)>\el N$,
\item \emph{horizontal} if $w(i)>\el N$ and $h(i)\le\es N$,
\item \emph{vertical} if $h(i)>\el N$ and $w(i)\le\es N$,
\item \emph{small} if $h(i)<\es N$ and $w(i)<\es N$,
\item \emph{intermediate }otherwise, i.e., $h(i)\in(\es N,\el N]$ or $w(i)\in(\es N,\el N]$. 
\end{itemize}
Following the ideas of~\cite{grandoni2021improved}, we can guess $\el$ and $\es$ in $O_\epsilon(1)$ time such that we may disregard all intermediate items and $\el$ and $\es$ are constants depending only on $\epsilon$. We state the lemma here for completeness.
\begin{lemma}[\cite{galvez2021approximating}]\label{lem:rec_noint}
	For any $\epsilon \geq 0$ and a positive increasing function $f(\cdot)$ (with $f(x) > x$), there exist constant values $\el$,~$\es$ with $\epsilon \geq \el \geq f(\es) \in \Omega_{\epsilon}(1)$ and $\es \in \Omega_\epsilon(1)$ such that the total profit of intermediate rectangles is bounded by $\epsilon p(\OPT)$. Furthermore, the pair $(\el,\es)$ is one pair from a set of $O_\epsilon(1)$ pairs and this set can be computed in polynomial time.
\end{lemma}

\begin{figure}[h!]
	\centering
	\includegraphics[width=.6\linewidth,page = 3]{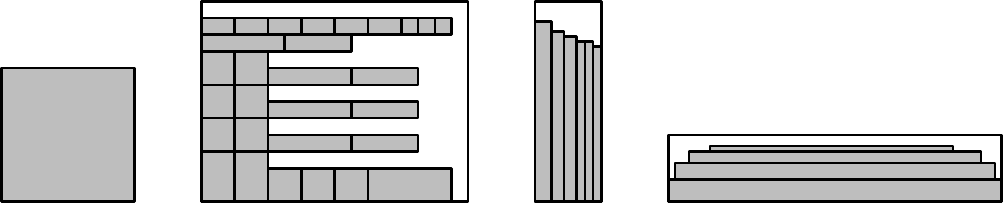}
	\caption{Visualization of $\L$-, $\S$-,$\V$- and $\H$-boxes}
	\label{fig:pack_LSVH}
\end{figure}
In the remainder of this section, we use the following notions. Consider some set of items $\I$. Denote by $h_{\min}(\I)$ ($w_{\min}(\I)$) and $h_{\max}(\I)$ ($w_{\max}(\I)$) the minimum and maximum height (width) among all items in $\I$, respectively.

\subsection{No rotations}\label{app:rec_nr}
In this subsection, we present the details for the setting without rotations leading to the proof of Theorem~\ref{thm:rec_nr}. We first give a high level overview. Similarly as before, let $\epsilon>0$ and we assume w.l.o.g.~that $1/\epsilon\in\mathbb{N}$ and that $\epsilon$ is sufficiently small.

\subsubsection{Structured packing of rectangles}
We use a packing based on rectangular boxes, similar to our packing
in Section~\ref{sec:hypercubes}. However, we distinguish the boxes
now into four different types (Figure~\ref{fig:pack_LSVH} shows a visualization of these types). The first type are $\L$-boxes which
contain only one large item each.
\begin{definition}[$\L$-boxes]
Let $B\subseteq K$ be an axis-parallel rectangle and suppose we
are given a packing of items inside $B$. We say that $B$ is an $\L$\emph{-box
}if $B$ contains exactly one large item and no other items.
\end{definition}
The next types are $\H$- and $\V$-boxes inside which the items are
intuitively horizontally and vertically stacked, respectively. See Figure~\ref{fig:pack_VH} for a visualization of these boxes.
\begin{definition}[$\H$- and $\V$-boxes]
Let $B$ be an axis-parallel rectangle and suppose we are given
a packing of items $\I'\subseteq\I$ inside $B$. We say that $B$
is a
\begin{itemize}
	\item $\H$\emph{-box} if $\I'$ contains only horizontal items which are
	stacked on top of each other inside $B$ and there are given values $w_{\min}(B),w_{\max}(B)$
	such that $w_{\min}(B)\le w_{i}\le w_{\max}(B)$ for each
	$i\in\I'$, 
	\item $\V$\emph{-box} if $\I'$ contains only vertical items which are
	stacked one next to the other inside $B$
	and there are given values
	$h_{\min}(B),h_{\max}(B)$
	such that $h_{\min}(B)\le h_{i}\le h_{\max}(B)$ for each
	$i\in\I'$.
\end{itemize}
\end{definition}
The third type are $\S$-boxes for which 
% which work similarly as the $\S$-boxes
%earlier: 
each item is small compared to the size of the box, and thus
NFDH can pack these items, unless they use almost the entire volume
of the box.
\begin{definition}[$\S$-boxes]
Let $B$ be an axis-parallel rectangle with height $h_{B}\in\mathbb{Z}_{\ge0}$
and width $w_{B}\in\mathbb{Z}_{\ge0}$ and let $\I'\subseteq\I$ be
a set of items packed into $B$. We say that $B$ is an $\S$\emph{-box}
if for each item $i\in\I'$ we have $h_{i}\le\epsilon\cdot h_{B}$
and $w_{i}\le\epsilon\cdot w_{B}$ and additionally $\sum_{i\in\I'}\lceil h_{i} \rceil_{1+\epsilon}\cdot \lceil w_{i}\rceil_{1+\epsilon}\leq(1-3\epsilon)h_{B}w_{B}$.
\end{definition}
We argue now that there always exists a $(2+\epsilon)$-approximate
\emph{easily guessable packing} with $\L$-, $\H$-, $\V$-, and $\S$-boxes that is suitable
for our indirect guessing framework.
% \newpage
\begin{lemma}[Near-optimal packing of rectangles]
\label{lem:struc_rectangles} There exists a packing with the following
properties: 

\begin{enumerate}[i)]
	\item it consists of $\L$-, $\H$-, $\V$- and $\S$-boxes whose total
	number is bounded by a value $C_{\mathrm{boxes}}(\epsilon)$ depending
	only on $\epsilon$. 
	\item there is a universal set of values $U(\epsilon$) with $|U(\epsilon)|=O_{\epsilon}(1)$
	such that for each box $B$ of the packing $h_{B}=\alpha_{B}\cdot N$
	for some $\alpha_{B}\in U(\epsilon)$,
	\item there exist values $k_{0},k_{1},k_{2},\dots,k_{r}\in\mathbb{Z}_{\geq0}$
	with $k_{0}=0$ and $r\in O_{\epsilon}(1)$ and a value $j_{B}\in\{1,2,\dots,r\}$
	for each $\H$- or $\V$-box $B$ such that
	\begin{enumerate}
		\item if $B$ is a $\H$-box, then $w_{\min}(B)=k_{j_{B}-1}$ and $w_{\max}(B)=k_{j_{B}}$, 
		\item if $B$ is a $\L$-box, then the unique item $i\in\I$ packed inside
		$B$ satisfies that $w_{i}=k_{j_{B}}$,
	\end{enumerate}
	\item let $h^{(1)},\dots, h^{(c)}$ be the distinct heights of $\V$-boxes in non-decreasing order such that for each $j = 1,\dots,c$ we have that
	\begin{enumerate}
		\item each $\V$-box of height $h^{(1)}$ only contains items of height at most $h^{(1)}$ and
		\item each $\V$-box of height $h^{(j)}$ only contains items of height at most $h^{(j)}$ and strictly larger than $h^{(j-1)}$ for $j=2,\dots,c$
	\end{enumerate}
	\item the total profit of the packing is at least $(\frac{1}{2}-O(\epsilon))\OPT$.
	% ,
	% where $\OPT$ is the profit of an optimal packing for instance
	% $\I$.
\end{enumerate}
\end{lemma}

To prove Lemma~\ref{lem:struc_rectangles} we need the \emph{resource augmentation packing lemma} introduced by Galvez et al.~\cite{galvez2021approximating}. In~\cite{galvez2021approximating}, three types of boxes are used. Let $B$ be a box with height $h(B)$ and width $w(B)$. Then, $B$ is a \emph{type 1} box if all items in $B$ are stacked on top of each other in non-increasing order of widths. Analogously, $B$ is a \emph{type 2} box if all items in $B$ are placed next to each other in non-increasing order of heights. Finally, we say $B$ is a \emph{type 3} box if for every item $i \in \I$ packed into $B$ we have $h_i  \leq \epsilon h(B)$ and $w_i \leq \epsilon w(B)$. The resource augmentation packing lemma states the following.
\begin{lemma}[\cite{galvez2021approximating}]\label{lem:rec_rs_augment}
	Let ${\hat{\I}}$ be a collection of items that can be packed into a box $B$ with height $h(B)$ and width $w(B)$, and let $\epsilon_{ra} > 0$ be a given constant. Then there exists a packing into sub-boxes of type 1, 2 and 3 of ${\hat{\I}}' \subseteq {\hat{\I}}$ inside a box $B'$ with height $(1+\epsilon_{ra})h(B)$ and width $w(B)$ such that:
	\begin{itemize}
		\item $p({\hat{\I}}') \geq (1-O(\epsilon_{ra}))p({\hat{\I}})$
		\item The number of sub-boxes used in the packing is $O_{\epsilon_{ra}}(1)$.
	\end{itemize}
\end{lemma}
Observe that the three types of boxes described above are different from our $\H$-, $\V$-, $\L$- and $\S$-boxes. But as it turns out in the upcoming proof of Lemma~\ref{lem:struc_rectangles}, we can modify the packing into type 1,2 and 3 boxes into a packing using $\H$-, $\V$-, $\L$- and $\S$-boxe while keeping a $1-O(\epsilon)$ fraction of the profit. 
\begin{proof}[Proof of Lemma~\ref{lem:struc_rectangles}]
	Let $\OPT$ be the optimal packing of $\I$ into $K$. We first apply a shifting argument to identify a horizontal strip
	$S:=[0,N)\times[a,a+\epsilon N)$ for
	some $a\in[0,(1-\epsilon)N)$ inside the knapsack $K$ such that
	\begin{itemize}
		\item the items that are horizontal or small and that intersect with $S$,
		together with
		\item the items that are large or vertical and for which at least one corner
		is contained in $S$
	\end{itemize}
	have a total profit of at most $2\epsilon\cdot\OPT$.
	To {show} this we can consider values $a_1,\dots, a_{1/\epsilon}$ where $a_j = (j-1)\epsilon$ as candidates for $a$ and pick one of them uniformly at random.
	% with probability $\epsilon$.
	Since $\epsilon \geq \el > \es$ we know that each horizontal and small item will intersect with at most two of the candidate strips. Hence, for each of them the probability of intersecting with the chosen strip is at most $2\epsilon$. For large and vertical items each corner is contained in at most one of the candidate strips. Therefore, the probability
	{that a given large or vertical item is contained in the set of items defined above is at most}
	$\epsilon$. Hence, the expected total profit of the two groups of items is at most $2\epsilon OPT$ which implies that there must also exist one particular strip such that the total profit of the two groups of items is at most $2\epsilon OPT$. We choose this strip and remove the items described above from $\OPT$, losing in total a profit of
	at most $O(\epsilon)\OPT$.
	% \\
	
	Now, consider the remaining packing. We partition the remaining items into
	two sets. The first set consists of all remaining items that are large
	or vertical and that intersect with $S$; we denote them by $\I_{1}$.
	Let $\I_{2}$ denote all other remaining items. 
	
	If $p(\I_{1})\ge p(\I_{2})$ then $p(\I_{1})\ge(\frac{1}{2}-O(\epsilon))\OPT$
	and we construct a packing for $\I_{1}$ as follows. Note that each
	item $i\in\I_{1}$ must ``cross'' $S$ since no corner of $i$ is
	contained in $S$. This implies that no two items are packed on top of each other and we can
	{imagine that the items are packed as in a one-dimensional knapsack.}
	% consider the problem as a $1D$-knapsack.
	By re-arranging the items from left to right, first starting with the vertical items and then {continuing} with the large items, we create a single $\V$-box and at most $O_\epsilon(1)$ many $\L$-boxes, each of height $N$, i.e., $\alpha_B=1$ for each of them. This packing satisfies properties i), ii) and v). We will now show that it also satisfies property iii). To do this, let $B_1,\dots, B_r$ be the $\L$-boxes with distinct values of $w_{\max}(\cdot)$ such that $w_{\max}(B_1) \leq \dots \leq w_{\max}(B_r)$ and {define} $k_j = w_{\max}(B_j)$ for each $j=1,\dots,r$. As the number of $\L$-boxes in the packing is at most $O_\epsilon(1)$, we also know that $r \in O_\epsilon(1)$. Finally, observe that we have only a single $\V$-box which automatically satisfies property iv).
	
	If $p(\I_{1})<p(\I_{2})$ then $p(\I_{2})\ge(\frac{1}{2}-O(\epsilon))\OPT$
	and we construct a packing for $\I_{2}$. Note that no item $i\in\I_{2}$
	intersects with $S$. Using the empty space $S$, we invoke the \emph{resource augmentation packing lemma} (see Lemma~\ref{lem:rec_rs_augment}) with $\epsilon_{ra} = \frac{\epsilon-\epsilon^2}{2-2\epsilon} \leq \epsilon$ which allows us to find a packing of a subset of items $\I'_{2}$ into $O_{\epsilon}(1)$ boxes,
	using only the area $[0,N]\times[0,(1-\epsilon/2-\epsilon^2)N]$ inside $K$ such that $p(\I'_2) \geq (1-O(\epsilon))p(\I_2)$. The boxes used in this packing, however, are boxes of type 1, 2 or 3. We will now show how to change the packing into a packing using $\H, \V, \L$ and $\S$-boxes. Let $B$ be a box of the current packing. We will now make a case distinction.
	
	\textbf{Case 1:} $B$ is a box of type 1 (all items in $B$ are stacked on top of each other in non-increasing order of heights). We first re-order the packing from bottom to top. First, we stack all horizontal items on top of each other. Then, we stack all large items on top of each other, followed by the vertical and, finally, the small items. Let $B_h$ be the sub-box of $B$ filled with horizontal items.
	{We define a $\H$-box as a sub-box of $B$ for all these horizontal items.}
	% By definition $B$ is a $\H$-box.
	Furthermore, since the height of $B$ is at most $(1-\epsilon/2)N$, there can be at most $1/\el$ many large and vertical items, respectively. For each large item,
	{we define an $\L$-box as a sub-box of $B$ that contains only this large item.}
	% the sub-box of $B$ containing only this item is a $\L$-box and,
	Similarly, {for each vertical item, we define a $\V$-box that contains only this vertical item.}
	% each sub-box containing a vertical item is a $\V$-box.
	Finally, consider the area of $B$ containing all small items. We will take care of these items later.
	
	\textbf{Case 2:} $B$ is a box of type 2 (all items in $B$ are stacked on next to each other in non-increasing order of widths). We first re-order the packing from left to right. First, we stack all vertical items on next to each other. Then, we stack all large items on next to each other, followed by the horizontal and, finally, the small items.
	{We define a $\V$-box $B_v$ as a sub-box of $B$ that contains all vertical items in $B$.}
	% Let $B_v$ be the sub-box of $B$ filled with vertical items. By definition $B$ is a $\V$-box.
	Furthermore, since the width of $B$ is at most $N$, there can be at most $1/\el$ many large and horizontal items, respectively. For each large item, {we define an $\L$-box as a} the sub-box of $B$ containing only this item. Similarly,
	{for each horizontal item we define a $\H$-box containing only this item.}
	% each sub-box containing a horizontal item is a $\H$-box.
	Finally, consider the area of $B$ containing all small items. We will take care of these items later.
	
	\textbf{Case 3:} $B$ is a box of type 3 (for every item $i \in \I(B)$ packed into $B$ we know that $h_i \leq \epsilon h_B$ and $w_i \leq \epsilon w_B$). We will now find a subset $\I'(B) \subseteq \I(B)$ such that
	\begin{itemize}
		\item $p(\I'(B)) \geq (1-O(\epsilon))p(\I(B))$ and
		\item $\sum_{i\in\I'(B)}\lceil h_{i} \rceil_{1+\epsilon}\cdot \lceil w_{i}\rceil_{1+\epsilon}\leq(1-3\epsilon)h_{B}w_{B}$.
	\end{itemize}
	To do this, consider the following LP-relaxation of the 1-D Knapsack problem with capacity $(1-4\epsilon)h_{B}w_{B}$ and items $\I(B)$, where the heigth and width are rounded up {to the next larger power of $1+\epsilon$} and the item size is given by the area of the items.
	
	\begin{alignat*}{3}
		& \text{minimize} & \sum_{i \in \I(B)} &x_ip_i& \\
		& \text{subject to} \quad& \sum_{i\in\I(B)} &x_i \lceil h_{i} \rceil_{1+\epsilon}\cdot \lceil w_{i}\rceil_{1+\epsilon}&\leq(1-4\epsilon)h_{B}w_{B} \\
		&&& x_{i} \geq 0, &  \forall i \in \I(B)
		\\
		&&& x_{i} \leq 1, &  \forall i \in \I(B)
	\end{alignat*}

	Observe that $\sum_{i\in\I(B)}\lceil h_{i} \rceil_{1+\epsilon}\cdot \lceil w_{i}\rceil_{1+\epsilon}\leq(1+\epsilon)^2h_{B}w_{B}$. Consider the following LP-solution $\hat{x}$ with $\hat{x}_i := (1+\epsilon)^{-2}(1-4\epsilon)$
	{for each item $i \in \I(B)$}. This solution yields a profit of at least $(1+\epsilon)^{-2}(1-4\epsilon)p(\I(B)) \geq (1-O(\epsilon))p(\I(B))$ and is feasible due to the observation above. Let $x^*$ be the optimal {extreme point} solution to the LP. By the rank lemma~\cite{lau2011iterative}, we know that there is at most one item $i' \in \I(B)$ such that $0< x^*_i < 1$. Let $\I'(B):= \{i \in \I(B): x^*_i > 0\}$. Then, we have that $p(\I'(B)) \geq (1-O(\epsilon))p(\I(B))$ and since we only pick one item for which $x^*_i < 1$ and this item has height $h_i \leq \epsilon h_B$ and width $w_i \leq \epsilon w_B$, we have that $\sum_{i\in\I'(B)}\lceil h_{i}\cdot w_{i}\rceil_{1+\epsilon}\leq(1-3\epsilon)h_{B}w_{B}$. Thus, $B$ with the set $\I'(B)$ packed into $B$ is an $\S$-box.
	
	Lastly, consider all small items which were packed into boxes of type 1 or type 2 or that are now packed into $\S$-boxes of width less than $\es N$. We will argue that these items can be packed into a box of height $\epsilon^2 N$ and width $N$ which can be placed at the top the free strip $S$. Observe that for each item $i\in \I$ in this group we have $h_i \leq \es N \leq \epsilon^3 N$ and $w_i \leq \es N \leq \epsilon N$ (by adequate choice of $f(\cdot)$ when applying Lemma~\ref{lem:rec_noint}).
	% Therefore, each individual item is small compared to the assigned box.
	Therefore, the total area of the small items packed into a box of type 1 (or type 2) is at most $(1-\epsilon)\es N^2$. Since there are at most $C_{boxes}(\epsilon)$ of these groups, the total area of all of them is at most $C_{boxes}(\epsilon)(1-\epsilon)\es N^2 \leq (1-\epsilon)\epsilon^2N^2$. Here, again we use that we choose an appropriate $f(\cdot)$ when removing intermediate items. Observe that this is possible since $f(\cdot)$ will only depend on $\epsilon$. In fact, we choose $f(x):= \frac{x}{c(\epsilon)}$ for some constant $c(\epsilon)\geq C_{boxes}(\epsilon)$. Using an LP-argument similar to the one used for the boxes of type 3, we can now transform this box into a single $\S$-box. Note that if one of the $\S$-boxes of width less than $\es N$ also contained vertical items we may now treat this box as a $\V$-box.
	
	Repeating these operations for each box of type 1, 2 or 3 results in a packing using $O_\epsilon(1)$ many $\H$-, $\V$-, $\L$- and $\S$-boxes containing a total profit of at least $(1/2-O(\epsilon))OPT(\I)$. Hence, this packing satisfies properties i), ii) and v) of the lemma. 
	
	We now proceed, to make sure that the packing also satisfies property iv). Consider a $\V$-box $B$ such that the property does not hold. Then, we can split $B$ into at most $C_{boxes}(\epsilon)$ many $\V$-boxes for which the property holds. To do this let $h^{(1)},\dots,h^{(c)}$ be the distinct heigths of $\V$-boxes in non-decreasing order.
	Then, let $B(h^{(1)})$ be the sub-area of $B$ containing items of height at most $h^{(1)}$, we now turn this sub-area into its own $\V$-box. Similarly, for each $j=2,\dots,c$ we look at the sub-area of $B$ containing items with heigth at most $h^{(j)}$ and strictly larger than $h^{(j-1)}$ and turn this sub-area into its own $\V$-box. Since $c \in O_\epsilon(1)$ this
	{increases the number of boxes by at most a factor $O_\epsilon(1)$.}
	
	Finally, we adapt the packing such that it also satisfies property iii). To this end, let $k_{1},\dots,k_r$ with $r \leq C_{boxes}(\epsilon)$ be the distinct widths of $\H$- and $\L$-boxes in non-decreasing order and set $k_0 := 0$. Now, for each $\H$-box $B$, we do the following. Let $B_j$ be the sub-box of $B$ containing only items with $w_i \in (k_{j-1},k_j]$. Then, $B_j$ is a $\H$-box as well. Repeating this for all horizontal boxes leads to at most $O_\epsilon(1)$ many $\H$-boxes satisfying property iii). For $\L$-boxes we make sure that each $\L$-box is exactly as wide as the item placed inside of it. Lastly, define $U(\epsilon):=\{i\cdot \frac{\epsilon}{2(C_{boxes}(\epsilon))}: i=1,\dots, \frac{2(C_{boxes}(\epsilon))}{\epsilon}\}$, where $C_{boxes}(\epsilon)$ is the number of boxes. For each box $B$, we round its height up to $\lceil h_B \cdot \frac{\epsilon}{2(C_{boxes}(\epsilon)+1)} \rceil N$. As there are at most $C_{boxes}(\epsilon)$ boxes on top of each other and we still have a free space with height $\epsilon/2 N$ at the top of $K$ the packing remains feasible after this rounding.
	
	Thus, we derived two packings satisfying properties i), ii), iii) and iv) such that one of them also satisfies property v).
\end{proof}

\subsubsection{Computing a packing}
We now describe how to compute a packing based on Lemma~\ref{lem:struc_rectangles}. We follow a similar structure as our algorithm for hypercubes. First, we guess some basic quantities of the packing.
Then, we find an implicit packing of vertical items into $\V$- and $\S$-boxes as well an implicit packing of small items into $\S$-boxes. Finally, we use {the} indirect guessing framework to find a packing of horizontal and large items into $\S$-,$\H$- and $\L$-boxes. Let $\B$ denote the set of boxes of the structured packing and $\B_\H$, $\B_\V$, $\B_\L$ and $\B_\S$ denote the set of $\H$-, $\V$-, $\L$- and $\S$-boxes in $\B$, respectively. {In the following, we assume that
	$\min\{\alpha_B: \alpha_B \in U(\epsilon)\} \ge \epsilon\cdot\epsilon_{\mathrm{small}}$. Since $U(\epsilon)$ is a universal set depending only on $\epsilon$, we can ensure this by choosing $c(\epsilon)$ accordingly.}

For the remainder of this section, we partition the set of items into profit classes, where a profit class is defined as $\P_{t}:= \{i \in \I: p_i \in [(1+\epsilon)^{t},(1+\epsilon)^{t+1})]\}$. Let $p_{\min}$ and $p_{\max}$ denote the smallest and largest profit of items in $\I$. We may disregard all items with profit less than $\epsilon \frac{p_{\max}}{n}$ (while only losing a profit of at most $\epsilon p_{\max} \leq \epsilon OPT$) such that $p_{\min} \geq \epsilon \frac{p_{\max}}{n}$. Hence, $t \in \mathcal{T}_P:=\{\lfloor p_{\min}\rfloor_{1+\epsilon}, \dots,\lceil p_{\max}\rceil_{1+\epsilon}\}$ with $|\mathcal{T}_P| \in O_\epsilon(\log n)$. Furthermore, we define $\widehat{p}(t)  = (1+\epsilon)^{t+1}$ {for each $t$ which is hence} the rounded profit of items in profit class $\P_{t}$.

\paragraph{Guessing basic quantities.}
First, we guess how many boxes of each type there are in $\B$. This amounts to a total of $O_{\epsilon}(1)$ many possibilities. For each of the boxes $B$ we guess its height. Since $|U(\epsilon)|\le O_\epsilon(1)$ there are only $O_\epsilon(1)$ possibilities for each of these heights. Note that each box has a height of at least $\frac{\el}{\epsilon}N$.

Since for a rectangle the height and the width can be different, guessing the widths of the $\S$-boxes is more difficult than it was in the setting of hypercubes. We will explain later how we do this.
% {Observe that in contrast to $\S$-boxes in the hypercube setting, $\S$-boxes in the rectangle setting may behave differently. Therefore, guessing their widths is not trivial and we will explain how to do this later.}

We do not know the values $k_{1},k_{2},\dots,k_{r}$; however, we
know that they yield a partition of $\I$ into sets $\I_{j}:=\{i\in\I:w_{i}\in(k_{j-1},k_{j}]\}$
for each $j\in[r]$ where for convenience we define $k_{0}:=0$. We
guess approximately the profit that each set $\I_{j}$ contributes
to $\OPT$. Formally, for each $j\in[r]$ we guess $\hat{p}(j):=\left\lfloor p(\I_{j}\cap\OPT)\right\rfloor _{1+\epsilon}$
if $p(\I_{j}\cap\OPT)\ge\frac{\epsilon}{r}\OPT$ and $\hat{p}(j):=0$
otherwise. Observe that for $\hat{p}(j)$ there are at most $O_{\epsilon,d}(\log n)$
possibilities since $\OPT\in[p_{\max},n\cdot p_{\max})$ and hence
$\hat{p}_{j}\in\{0\}\cup[\frac{\epsilon}{r}p_{\max},n\cdot p_{\max})$.
Also, one can show that $\sum_{j=1}^{r}\hat{p}(j)\ge(1-O(\epsilon))\OPT$.

Now, for each $\H$- or $\L$-box $B$, we guess the value $j_B$. This needs a total amount of $O_\epsilon(1)$ many guesses. 
Observe that each $\H$- or $\L$-box $B\in\B$ can contain only items from
$\I_{j_{B}}$. However, each $\S$-box $B\in\B$ might contain items
from more than one set $\I_{j}$. For each $\S$-box $B\in\B$ and
each set $\I_{j}$, we guess approximately the fraction of the area of
$B$ that is occupied by items from $\I_{j}$. Formally, for each
such pair we define the value $a_{B,j}:=\frac{\sum_{i\in\I(B)\cap\I_{j}}\lceil h_{i}\rceil_{1+\epsilon}\lceil w_{i}\rceil_{1+\epsilon}}{h_Bw_B}$
and guess the value $\hat{a}_{B,j}:=\left\lceil \frac{a_{B,j}}{\epsilon/(r+2)}\right\rceil \epsilon/(r+2)$,
i.e., the value $a_{B,j}$ rounded up to the next larger integral
multiple of $\epsilon/r$. Note that for each value $\hat{a}_{B,j}$
there are only $O_{\epsilon}(1)$ possibilities, and that there
are only $O_{\epsilon}(1)$ such values. Additionally, $\S$-boxes may contain small items and vertical items (which by definition are not part of any $\I_j)$. Therefore, we also guess values $\hat{a}^{\mathrm{small}}_{B}$ and $\hat{a}^{\mathrm{vert}}_{B}$ defined similarly as above. For the correct guesses we have that $\left(\sum_{j=1}^r \hat{a}_{B,j} + \hat{a}^{\mathrm{small}}_{B} + \hat{a}^{\mathrm{vert}}_{B} \right)h_Bw_B   \leq (1-3\epsilon)h_Bw_B$ for each $\S$-box $B$ to account for possibly needing more space later. Overall, the rounded values may lead to less profit than in the structured packing. However, we can still guarantee a profit of at least a $(1+O(\epsilon))$ ratio of the profit achieved by the structured packing. 

We now distinguish between two cases in order to {guess} the widths of $\S$-boxes. We first consider the case that there is at least one $\S$-box containing a horizontal or large item, i.e. $\hat{a}_{B,j} > 0$ for some $B \in \B_\S$ and some $j$. In this case we guess the width of this box {up to a factor of $1+\epsilon$} in time $O_{\epsilon}(1)$ since it is at most $N$ and at least $\frac{\el}{\epsilon} N$. Let $w^*$ denote this width and consider all small items
% and vertical items
% r{commented out ``and vertical items'' since they might not fit into $B$}
of width less than $\frac{\epsilon}{n}w^*$. Their total width is at most $\epsilon w^*$ {and hence they} can be placed next to each other in a shelf of height $\el N$ using a shifting argument we can replace other items in this box while losing only an $\epsilon$-fraction of the profit. Thus, all remaining {small} items have width at least $\frac{\epsilon}{n}w^*$ implying that {for} each remaining $\S$-box {its} width {is} in the interval $[\frac{\epsilon}{n}w^*,w^*]$; {hence, it } can be guessed in time $O(\log_{1+\epsilon} n)$ {up to a power of} $1+\epsilon$. We will explain later how to handle the other case.

Next, we will describe our indirect guessing framework which is used to pack horizontal and large items.

\paragraph{Indirect guessing framework: Packing horizontal and large items into $\H$-, $\L$- and $\S$-boxes.}
Next, we explain how to compute a packing of large and horizontal items into $\H$-, $\L$- and $\S$-boxes. We proceed similarly to the indirect guessing framework used in the hypercube setting based on the framework introduced in~\cite{heydrich2019faster}. We will describe the indirect guessing framework for the case that we guessed the approximate width of $\S$-boxes already using the trick explained above. If this is not the case, we can ignore the $\S$-boxes in the indirect guessing framework since they only contain small and vertical items; we will pack the latter items later.

The goal is to determine the values $k_{1},k_{2},\dots,k_{r}$.
Unfortunately, we cannot guess them directly in polylogarithmic time,
since there are $N$ options for each of them. Instead, we define $\tilde{k}_{0}:=0$ and compute values $\tilde{k}_{1},\tilde{k}_{2},\dots,\tilde{k}_{r}$
that we use instead of the values $k_{0},k_{1},k_{2},\dots,k_{r}$.
This yields a partition of $\I$ into sets $\tilde{\I}_{j}:=\{i\in\I:w_{i}\in(\tilde{k}_{j-1},\tilde{k}_{j}]\}$.
Intuitively, for each $j$ we want to pack items from $\tilde{\I}_{j}$
into the space that is used by items in $\I_{j}$ in the packing from
Lemma~\ref{lem:struc_rectangles}. We will choose the values $\tilde{k}_{1},\tilde{k}_{2},\dots,\tilde{k}_{r}$
such that in this way, we obtain almost the same profit. On the other
hand, we will ensure that $\tilde{k}_{j}\le k_{j}$ for each $j\in[r]$.

We work in $r$ iterations. We define $\tilde{k}_{0}:=0$. Suppose
inductively that we have determined $\ell$ values $\tilde{k}_{1},\tilde{k}_{2},\dots,\tilde{k}_{\ell}$
already for some $\ell\in\{0,1,...,r-1\}$ such that $\tilde{k}_{\ell}\le k_{\ell}$.
We want to compute $\tilde{k}_{\ell+1}$. We can assume w.l.o.g.~that
$k_{\ell+1}$ equals $w_{i}$ for some item $i\in\I$. We do binary
search on the set $W(\ell):=\{w_{i}:i\in\I\wedge w_{i}>\tilde{k}_{\ell}\}$,
using our rectangle data structure. For each candidate value $w \in W(\ell)$,
we estimate the possible profit due to items in $\tilde{\I}_{\ell+1}$
if we define $\tilde{k}_{\ell+1}:=w$. We want to find such a value
$w$ such that the obtained profit from the set $\tilde{\I}_{\ell+1}$
equals essentially $\hat{p}(\ell+1)$. In the following, we denote by $\B_{\H}$, $\B_{\L}$ and $\B_{\S}$, the set of $\H$-, $\L$- and $\S$-boxes in $\B$, respectively.

We describe now how we estimate the obtained profit for one specific
choice of $w\in W(\ell)$. We try to pack items from $\tilde{\I}_{\ell+1}(w):=\left\{ i\in\I:w_{i}\in(\tilde{k}_{\ell},w]\right\} $
into
\begin{itemize}
\item the $\H$-boxes $B\in\B_{\H}$ for which $j_{B}=\ell+1$,
\item the $\L$-boxes $B\in\B_{\L}$ for which $j_{B}=\ell+1$ and
\item the $\S$-boxes, where for each $\S$-box $B\in\B_{\S}$, we use an area
of $\widehat{a}_{B,\ell+1}\cdot h_Bw_B$ and ensure that we pack
only items $i\in\tilde{\I}_{\ell+1}(w)$ for which $w_{i}\leq 2\epsilon \widehat{w}_B$ and $h_i \leq \epsilon h_B$.
\end{itemize}
We solve this subproblem approximately via the following integer program
$(\mathrm{IP}(w))$. Intuitively, we treat items equally if they have
almost the same profit and almost the same height and width, up to a factor $1+\epsilon$,
respectively. To this end, we use the notion of height and width classes.  Formally, we define a height class $H_{t'}=\{i\in \tilde{\I}_{\ell+1}(w):h_{i}\in[(1+\epsilon)^{t'},(1+\epsilon)^{t'+1})\}$
for each $t'\in \mathcal{T}_H:= \{\lfloor \log_{1+\epsilon}(h_{\min})\rfloor,\dots,\lceil \log_{1+\epsilon}(h_{\max})\rceil\}$. Furthermore, we define a width class $W_{t^{''}}=\{i\in \tilde{\I}_{\ell+1}(w):w_{i}\in[(1+\epsilon)^{t^{''}},(1+\epsilon)^{t^{''}+1})\}$ for each  $t^{''} \in \mathcal{T}_W:= \{\lfloor \log_{1+\epsilon}(\tilde{k}_\ell)\rfloor,\dots,\lceil \log_{1+\epsilon}(w)\rceil\}$.
We denote by $\hat{h}(t'):=(1+\epsilon)^{t'+1}$ and $\hat{w}(t^{''}):=(1+\epsilon)^{t^{''}+1}$ the rounded height and width, respectively. Furthermore, we define a set of triplets $\mathcal{T}:=\{(t,t',t''): t \in \mathcal{T}_P \wedge t' \in \mathcal{T}_H \wedge t^{''} \in \mathcal{T}_W\}$. For each triplet $(t,t',t^{''}) \in \mathcal{T}$, $n_{t,t',t^{''}}$ denotes the number of items of profit class $\P_t$, height class $H_{t'}$ and width class $W_{t^{''}}$. Losing only a factor of $(1+O(\epsilon))$, the subproblem is equivalent to selecting how many items of each combination of profit, height and width class will be packed into each box which we can formulate as the following IP. We denote by $\B_\H(\ell+1)$ the $\H$-boxes for which $j_B = \ell+1$ and by $\B_\L(\ell+1)$ the $\L$-boxes for which $j_B = \ell+1$. Furthermore, by $\B(\ell+1)$ we denote the set of all relevant boxes for this iteration of the indirect guessing framework.

\begin{alignat*}{3}
	(\mathrm{IP}(w))\quad& \text{max} 	& \displaystyle \sum_{(t,t',t^{''}) \in \mathcal{T}}\sum_{B \in \B(\ell+1)} x_{t,t',t''',B} p(t)			& 			& \quad & \\
	& \text{s.t.} & \displaystyle  \sum_{(t,t',t^{''}) \in \mathcal{T}} x_{t,t',t'',B} \hat{h}(t')	& \leq h_B& 		& \forall B \in \B_{\H}(\ell+1) \\
	&  & \displaystyle  \sum_{(t,t',t^{''}) \in \mathcal{T}} x_{t,t',t'',B} 	& \leq 1& 		& \forall B \in \B_{\L}(\ell+1) \\
	& & \displaystyle \sum_{(t,t',t^{''}) \in \mathcal{T}} x_{t,t',t'',B} \hat{h}(t')\hat{w}(t'')	& \leq a_{B,\ell+1}h_B\widehat{w}_B	& 		& \forall B \in \B_{\S}\\
	&				& \displaystyle\sum_{B \in \B(\ell+1)} x_{t,t',t'',B}								& \leq n_{t,t',t''}	& 		& \forall  (t,t',t^{''}) \in \mathcal{T} \\
	&				& x_{t,t',t'',B}								& \in \mathbb{N}_{0}& 		&\forall (t,t',t^{''}) \in \mathcal{T}, B \in \B(\ell+1)\\
\end{alignat*}
We are now faced with the same challenges as in the case of hypercubes. We cannot solve $(\mathrm{IP}(w))$ or its LP-relaxation directly since we might have $\mathrm{poly}(\log N)$ many variables. However, we now show how to solve it approximately losing only a factor of $1+\epsilon$. We remark that the computed solution is not an explicit packing since for each triplet $(t,t',t'')$ we compute how many items of profit class $\P_t$, height class $H_{t'}$ and width class $W_{t^{''}}$ we want to select for each box but not the identities of the selected items. This would not be possible in time $O_{\epsilon}(\log(n))$ since the solution might consist of $\Omega(n)$ items

\begin{lemma}\label{lem:rec_IP_sol}
There is an algorithm with a running time of {$(\log_{1+\epsilon}(n))^{O(1)}$}
that computes an $(1{+}O(\epsilon))$-approximate solution for \textup{$(\mathrm{IP}(w))$};
we denote by $q(w)$ the value of this solution. For two values $w,w'$
with $w\le w'$ we have that $q(w)\le q(w')$.
\end{lemma}
\begin{proof}
We will compute an approximate solution $(\mathrm{IP}(w))$, by first guessing the $2C_{boxes}(\epsilon)/\epsilon$ most profitable items of set $\tilde{\I}_{\ell+1}(w)$ in the solution.To do this, we again first guess the profit class of each of these items with a total number of $\log_{1+\epsilon}^{O_{\epsilon}(1)}n$ many guesses. For each profit class $t$ this gives value $n_t^g \in O_{\epsilon}(1)$ indicating the number of items guessed from this profit class. Next, for each profit class we find the $n_t$ items of smallest height. We can do this in time $O(\log^4 n)$ using our rectangle data structure (see Lemma~\ref{lem:data-structure-rec}) and the balanced binary search tree used to store the distinct item heights. Observe that, we only lose a factor of $(1+\epsilon)$ of the profit by taking these items. Now it remains to guess for each of these items which box it must be packed in which can be done in time $O_\epsilon(1)$. We denote this guessed solution by $S^g$. Next, we need to update the right hand-sides of $(\mathrm{IP}(w))$.  For an $\H$-box $B$, let $h^g_B$ be the total height of all guessed items. Similarly, for an $\S$-box $B$ let $\mathrm{area}^g_B$ be the total area of the guessed items. For $\L$-boxes, we denote by $n^g_B$ the number of items guessed for $B$ (which is either $0$ or $1$). Then, the following IP remains.
\begin{alignat*}{3}
	(\mathrm{IP}'(w))\quad& \text{max} 	& \displaystyle \sum_{(t,t',t^{''}) \in \mathcal{T}}\sum_{B \in \B(\ell+1)} x_{t,t',t''',B} p(t)			& 			& \quad & \\
	& \text{s.t.} & \displaystyle  \sum_{(t,t',t^{''}) \in \mathcal{T}} x_{t,t',t'',B} \hat{h}(t')	& \leq h_B-h^{g}_B& 		& \forall B \in \B_{\H}(\ell+1) \\
	&  & \displaystyle  \sum_{(t,t',t^{''}) \in \mathcal{T}} x_{t,t',t'',B} 	& \leq 1-n^{g}_B& 		& \forall B \in \B_{\L}(\ell+1) \\
	& & \displaystyle \sum_{(t,t',t^{''}) \in \mathcal{T}} x_{t,t',t'',B} \hat{h}(t')\hat{w}(t'')	& \leq a_{B,\ell+1}h_B\widehat{w}_B	 - \mathrm{area}^{g}_B& 		& \forall B \in \B_{\S}\\
	&				& \displaystyle\sum_{B \in \B(\ell+1)} x_{t,t',t'',B}								& \leq n_{t,t',t''}	& 		& \forall  (t,t',t^{''}) \in \mathcal{T} \\
	&				& x_{t,t',t'',B}								& \in \mathbb{N}_{0}& 		&\forall (t,t',t^{''}) \in \mathcal{T}, B \in \B(\ell+1)\\
\end{alignat*}
Next, we restrict the number of variables by restricting the number of height and width classes. {To this end, observe that since we only consider horizontal and large items, we know that there are $O_{\epsilon}(1)$ many relevant width classes and for large items we also only have $O_{\epsilon}(1)$ many height classes. Thus, we need to restrict the number of height classes for for horizontal items. Therefore, let $B \in \B_\H$ be the horizontal box of maximum height, denoted by $h^*$.
	{Since the set $U(\epsilon)$ is a universal set that depends only on $\epsilon$ and we can require that $\el$ is small
		compared to the values in $U(\epsilon)$, we can ensure that $B$ has a height
		of}
	% 	We know that this box has height
	at least $\frac{\el}{\epsilon}N$. Thus, consider all horizontal items of height less than $\frac{\el}{n}N$. By reserving a space of of $\el N$ in each of the horizontal boxes we may pack all these very thin items while losing only a factor of $1+\epsilon$ in the profit.}

If we consider the appropriate height and width classes in $(\mathrm{IP'}(w))$, we now have an IP with $\log_{1+\epsilon}^{O(1)} n$ many variables. Let $S^f(w)$ be an optimal solution to the LP-relaxation of this IP. By the rank lemma~\cite{lau2011iterative} $S^f(w)$ has at most $2C_{boxes}(\epsilon)$ many fractional non-zero variables and the profit corresponding to each of them is smaller than the profit of any item in the guessed solution $S^g(w)$. Thus, by rounding down these values we lose at most a profit of $\epsilon(p(S^g(w)))$. Therefore, combining this rounded solution with $S^g(w)$ and all tiny items we discarded to reduce the number of variables, we can find an approximate solution to $(\mathrm{IP}(w))$. Finally, consider two values $w\leq w'$. Then, the approximate solution obtained to $(\mathrm{IP}(w))$ is also feasible for $(\mathrm{IP}(w'))$. Therefore, $q(w) \leq q(w')$.
\end{proof}

We define $\tilde{k}_{\ell+1}$ to be the smallest value $w\in W(\ell)$
for which $q(w)\ge (1-\epsilon)\hat{p}(\ell+1)$. Based on this and the inductive assumption that $\tilde{k}_{\ell}\le k_{\ell}$ we can prove the following statement.
\begin{lemma}\label{lem:rec_induc_kr}
We have that $\tilde{k}_{\ell+1}\le k_{\ell+1}$.
\end{lemma}
{We run our indirect guessing framework for $O_{\epsilon}(1)$ iterations, and hence obtain a packing of horizontal and large items into $\H$-, $\L$- (and possibly $\S$-boxes).}

\paragraph{Packing vertical items into $\V$- and $\S$-boxes.}
{Before explaining how to pack vertical and small items, we now show how to guess the width of $\S$-boxes in case this was not done before the indirect guessing framework. To this end, we first guess the widths of $\V$-boxes as follows. We start by guessing the $2C_{boxes}(\epsilon)/\epsilon$ most profitable vertical items in the packing (this will also be useful for charging profit of items we discard later). We do this following the same technique as before. First, we guess the profit type of each of these items in time $O(\log_{1+\epsilon}n)$ and {also the box that} it must be assigned to (this gives an indication on the height {of the item}). Then, for each profit class and each interval of heights $(h^{(j-1)},h^{(j)}]$ defined by the heights of $\V$-boxes, we now know how many items to pick. Losing only a factor of $(1+\epsilon)$ of the profit we may choose the items in non-decreasing order of widths. We next guess the total width and how it is split among the $\V$-boxes. To do this, we again guess the profit type of the item packed into the remaining space of $\V$-boxes whose width is maximal in time $O(\log_{1+\epsilon}n)$. Let $t$ be the this profit type and let $n_t^g(\I_{\V}(j))$ be the number of items from $\I_{\V}(j)$ of this profit type which should be in our final packing. We can guess this value as a power of $(1+\epsilon)$ in time $O(\log_{1+\epsilon}n)$. By losing only a profit of $(1+\epsilon)$ we can assume that we may choose the  $n_t^g(\I_{\V}(j))$ narrowest (lowest width) items of this profit class. Therefore, we find the width of the $n_t^g(\I_{\V}(j))$-th narrowest of these items in time $O(\log^4 n)$;
we denote by $w^*$ its width.
%\awr{please double-check this, I was not sure}
% of width $w^*$.
If this item is packed into a $\V$-box then we know that the total width of the remaining $\V$-boxes is within the range $[w^*,nw^*]$ and can be guessed up to a factor of $(1+\epsilon)$ in time $O(\log_{1+\epsilon}n)$. If, however, this item is packed into an $\S$-box we may assume that all items of width {at most} $\epsilon \frac{w^*}{n}$ can also be packed into this $\S$-box. Therefore, the total remaining width of the $\V$-boxes must be in the range $[\frac{\epsilon}{n} w^*,nw^*]$ and can be guessed up to a factor of $(1+\epsilon)$ in time $O(\log_{1+\epsilon}n)$. Now, using this guessed value of the total remaining width of the $\V$-boxes, we can guess the {approximately the individual width $\hat{w}_B$ of each $\V$-box $B$
	in time $(\log_{1+\epsilon}n)^{O_\epsilon(1)}$}. We are now ready to guess the widths of the $\S$-boxes as powers of $(1+\epsilon)$. Consider the $\S$-box $B^*$ of maximum width.
{Imagine that we make $B^*$ greedily wider as much as possible such that we possibly push other boxes on the left of $B^*$ to the left towards the left edge of the knapsack, and similarly we push boxes on the right of $B^*$ to the right towards the right edge of the knapsack. Once we cannot make $B^*$ wider anymore, there is a set of boxes
	$\B'$ with $B^* \in \B'$ such that $\sum_{B\in \B'}w_B =N$.
	Let $w(\mathrm{fixed})$ be the total width of $\L$-,$\H$- and $\V$-boxes in $\B'$ for which we already know the widths and let $n_\S$ the number of $\S$-boxes in $\B'$}. Then, we know that the width of $B^*$ is in the interval $[\frac{N-w(\mathrm{fixed})}{n_\S},N-w(\mathrm{fixed})]$ and hence it can be guessed in time $O(\log_{1+\epsilon}n)$ {up to a factor of $1+\epsilon$}. Let $\widehat{w}_{B^*}$ denote {the guessed width}. Then, by our assumption on the minimum height of the boxes, we know that all items of width less than $\frac{\epsilon}{n}\widehat{w}_{B^*}$ have a total area of at most $\epsilon$ times the area of $B^*$. Therefore, by reserving space for these items we can pack them all in $B^*$, {by removing some other items from $B^*$ such that we lose only a factor of $1+\epsilon$ in the profit}. From this it follows that for all other $\S$-boxes {$B'$, each item packed into $B'$} has width at least $\frac{\epsilon}{n}\widehat{w}_{B^*}$ and at most $\epsilon \widehat{w}_{B^*}$ and, thus, the widths {of the other $\S$-boxes} can be guessed in time $O(\log_{1+\epsilon}n)$.
}

Let $\I_{\V}$ denote the set of vertical items. It is important to observe that vertical items are only contained in $\V$-boxes and $\S$-boxes. We will now prove the following statement. Furthermore, observe that due to the structured packing and since we guessed the height of each $\V$-box we know a specific range for the height of items which we may pack into each $\V$-box. More specifically, let $\B_{\V}$ denote the set of $\V$-boxes in $\B$ and let $h^{(1)},\dots,h^{(c)}$ be the set of distinct heights of boxes in $\B_{\V}$ in non-decreasing order and set $h_0:=0$. Then, we know that each $\V$-box of heigth $h^{(j)}$ contains items with height in the interval $(h^{(j-1)},h^{(j)}]$ only. Denote by $\B_{\V}(j)$ the set of $\V$-boxes of height $h_j$ and let $\I_{\V}(j):= \{i \in \I_{\V}: h_i \in (h^{(j-1)},h^{(j)}]\}$. For each $\S$-box we can now guess a value $\hat{a}^{\mathrm{vert}}_{B,,j} := \left\lceil \frac{a_{B,\V,j}}{\epsilon/(c)}\right\rceil \epsilon/(c)$ where $a^{\mathrm{vert}}_{B,j}:=\frac{\sum_{i\in\I(B)\cap\I_{\V}(j)}\lceil h_{i}\rceil_{1+\epsilon}\lceil w_{i}\rceil_{1+\epsilon}}{h_Bw_B}$. We only consider guesses satisfying $\sum_{j=1}^c \hat{a}^{\mathrm{vert}}_{B,j} = \hat{a}^{\mathrm{vert}}_{B}$.
We will now show how to find a packing of vertical items for some range $(h^{j-1},h^j]$. We then apply this procedure to each range individually.
\begin{lemma}\label{lem:rec_lp_vert}
There is an algorithm with a running time of {$(\log_{1+\epsilon}n)^{O_{\epsilon}(1)}$}
that computes an $(1{+}O(\epsilon))$-approximate packing of items $\I_{\V}(j)$.
\end{lemma}
\begin{proof}
We have already guessed the widths of every box and the $2C_{boxes}(\epsilon)/\epsilon$ most profitable vertical items in the packing. We will proceed by finding an approximate solution for the remaining packing using similar ideas as before. First, we define a width class $W_{t'}:=\{i\in \I_{\V(j)}:w_{i}\in[(1+\epsilon)^{t'},(1+\epsilon)^{t'+1})\}$ with rounded-up width $\widehat{w}(t') =(1+\epsilon)^{t'+1}$. We denote by $\mathrm{W}$ the set of all width classes and define $\mathcal{T}_W :=\{t': \mathcal{W}_{t'} \in \mathrm{W}\}$. For each pair $(t,t') \in \mathcal{T}:=\{(t,t'):t \in \mathcal{T}_P \wedge t' \in \mathcal{T}_W\}$ we treat items of profit class $\P_t$ and width class $W_{t'}$. Using similar arguments as before we can show that we can restrict ourselves to $O(\log_{1+\epsilon}n)$ many width classes by either reserving an $\epsilon$-fraction of the width of the $\V$-boxes or an $\epsilon$-fraction of the area of $\S$-boxes. Thus, to find the remaining packing of vertical items, we can consider the following LP which finds a fractional implicit packing of the remaining items from. An optimal solution to this LP again loses a factor of $1+\epsilon$ of the profit and it leaves space to add all omitted items.
{
	\small	
	\begin{alignat*}{3}
		(\mathrm{LP}(j))\quad& \text{max} 	& \displaystyle \sum_{(t,t') \in \mathcal{T}}\sum_{B \in  \B_{\V}(j)} x_{t,t',B} \widehat{p}(t)			& 			& \quad & \\
		& \text{s.t.} & \displaystyle \sum_{(t,t') \in \mathcal{T}} x_{t,t',B} \widehat{w}(t')	& \leq (1-\epsilon)\widehat{w}_B& 		& \forall B \in \B_{\V}(j) \\
		& \text{s.t.} & \displaystyle \sum_{(t,t') \in \mathcal{T}} x_{t,t',B} \widehat{w}(t')	& \leq \widehat{a}_{B,\V,j}(1-\epsilon)h_Bw_B - \mathrm{area}(S^g(j))& 		& \forall B \in \B_{\S} \\
		&				& \displaystyle\sum_{B \in  \B_{\V}(j)} x_{t,t',B}								& \leq n_{t,t'}	& 		& \forall  (t,t') \in \mathcal{T}\\
		&				& x_{t,t',B}								& \geq 0& 		&\forall (t,t') \in \mathcal{T}, B \in \B_{\V}(j)\\
	\end{alignat*}
}
Let $S^f(j)$ be the optimal fractional solution to this remaining LP which we can find in time
$(\log_{1+\epsilon}n)^{O_\epsilon(1)}$~\cite{cohen2021solving}. We now combine $S^g(j)$ and $S^f(j)$ to yield a feasible fractional solution to $\mathrm{LP}(j))$.
% Observe that
% we are now guaranteed to have at least $2C_{boxes}(\epsilon)/\epsilon$ many integral variables and
Due to the rank lemma~\cite{lau2011iterative} {we have} at most $2C_{boxes}(\epsilon)$ fractional non-zero variables; {each item corresponding to such a fractional variable has less profit than any previously guessed item.}
% where each of them is less profitable than the integral variables.
Therefore, by rounding down these variables we only lose a profit of at most $\epsilon p(S^g(j))$ (where $p(S^g(j))$ is the profit due to the guessed items). Finally, combining the integral solution with all tiny items that were discarded to restrict our necessary number of guesses, we find an implicit $(1+O(\epsilon)$-{approximate}
packing of vertical items of the set $\I_{\V}(j)$. It is important to note that using our rectangle data structure (see Lemma~\ref{lem:data-structure-rec}.) we can find the number of tiny items as well as an $(1+\epsilon)$-approximation of the total profit of tiny items in time $O(\mathrm{poly}(\log n))$.
\end{proof}

As there are at most $O_{\epsilon}(1)$ many different values of $j$, we have the following result.
\begin{lemma}\label{lem:rec_lp_vert_all}
There is an algorithm with a running time of {$(\log_{1+\epsilon}(n))^{O_{\epsilon}(1)}$}
that computes an $(1{+}O(\epsilon))$-approximate packing of the items $\I_{\V}$ into the guessed $\V$-boxes $\B_\V$. 
\end{lemma}

\paragraph{Packing small items into $\S$-boxes.}
As small items are only packed into $\S$-boxes, we may also pack them separately. Hereto, we use the same ideas as for the vertical items. Let $\I_{s}$ denote the set of small items. {At this point, we already know the guessed width of the $\S$-boxes as well as that all remaining small items have width at least $\frac{\epsilon}{n}w^*_{\S}$ and at most $\epsilon w^*_{\S}$, where $w^*_{\S}$ is the maximum width of all $\S$-boxes. Thus, we can show the following result using the same techniques as before.}
\begin{lemma}\label{lem:rec_lp_small}
There is an algorithm with a running time of {$(\log_{1+\epsilon}(n))^{O_{\epsilon}(1)}$}
that computes an $(1{+}O(\epsilon))$-approximate packing of the items $\I_{s}$ into the guessed $\S$-boxes $\B_{\S}$.
\end{lemma}
\begin{proof}
Again, we start by guessing the $2C_{boxes}(\epsilon)/\epsilon$ most profitable items from 
$\I_{s}$ in our packing. We do this as follows. For each of these items we first guess the profit class which we can do in time $(\log_{1+\epsilon}n)^{O_\epsilon(1)}$. For each profit class $\P_t$ this gives us a value $n_t \in O_{\epsilon}(1)$. Now for each profit class we can find the $n_t$ items of lowest width in this profit class in time $O(\log^4 n)$ using our rectangle data structure (see~Lemma~\ref{lem:data-structure-rec}) in combination with the balanced binary search tree storing all possible width values. Note that we may assume that the structured packing also uses this set of items by only losing a factor of $(1+\epsilon)^{-1}$. Finally, for each of these items we guess which box it must be assigned to which can be done in time $O_{\epsilon}(1)$. This results in a preliminary packing $S^g$ of $C_{boxes}/\epsilon$ items of the set $\I_{\S}$. For each $\S$-box we also obtain a value $\mathrm{area}(S^g)$ which is now already occupied.

We will now show how to compute the remaining packing. To this end, we again use the notion of width and height classes (of the small items now). Formally, we define a height class $H_{t'}=\{i\in \I_s:h_{i}\in[(1+\epsilon)^{t'},(1+\epsilon)^{t'+1})\}$
for each $t'\in \mathcal{T}_H= \{\lfloor \log_{1+\epsilon}(h_{\min}(\I_s))\rfloor,\dots,\lceil \log_{1+\epsilon}(h_{\max}(\I_s))\rceil\}$. Furthermore, we define a width class $W_{t^{''}}=\{i\in \I_s:w_{i}\in[(1+\epsilon)^{t^{''}},(1+\epsilon)^{t^{''}+1})\}$ for each  $t^{''} \in \mathcal{T}_W = \{\lfloor \log_{1+\epsilon}(w_{\min}(\I_s))\rfloor,\dots,\lceil \log_{1+\epsilon}(w_{\max}(\I_s))\rceil\}$.
We denote by $\hat{h}(t'):=(1+\epsilon)^{t'+1}$ and $\hat{w}(t^{''}):=(1+\epsilon)^{t^{''}+1}$ the rounded height and width, respectively.  For each triplet $(t,t',t^{''}) \in \mathcal{T}:=\{(t,t',t^{''}):t \in \mathcal{T}_P \wedge t' \in \mathcal{T}_H\wedge t^{''} \in \mathcal{T}_W\}$ we treat items of profit class $\P_t$, height class $H_{t'}$ and width class $W_{t^{''}}$ as identical.
Lastly, we let $n_{t,t',t^{''}}$ denote the number of items of profit class $\P_t$, height class $H_{t'}$ and width class $W_{t^{''}}$ for each triplet $(t,t',t^{''}) \in \mathcal{T}$. Consider $B \in \B_\S$ such that $a^{\mathrm{small}}_Bh_Bw_B$ is maximal. We already know that the number of width classes is $O(\log_{1+\epsilon}n)$. Hence, we need to restrict the number of height classes. Which can be done using the same techniques used to restrict the number of width classes by discarding some tiny items and reserving space for them in each of the $\S$-boxes. For the remaining items we consider the following LP which gives a fractional $(1+O(\epsilon))$-approximation for the remaining packing.

{
\footnotesize
	\begin{alignat*}{3}
		(\mathrm{LP}(\S))\quad& \text{max} 	& \displaystyle \sum_{(t,t',t^{''}) \in \mathcal{T}}\sum_{B \in \B_\S} x_{t,t',t''',B} p(t)			& 			& \quad & \\
		& &\displaystyle\sum_{(t,t',t^{''}) \in \mathcal{T}} x_{t,t',t'',B} \hat{h}(t')\hat{w}(t'')	& \leq (1-\epsilon)a^{\mathrm{small}}_{B}h_Bw_B - \mathrm{area}(S^g)	& 		& \forall B \in \B_\S \\
		&				& \displaystyle\sum_{B \in  \B_\S} x_{t,t',t'',B}								& \leq n_{t,t',t''}	& 		& \forall (t,t',t^{''}) \in \mathcal{T} \\
		&				& x_{t,t',t'',B}								& \geq 0& 		&\forall (t,t',t^{''}) \in \mathcal{T}, B \in  \B_\S\\
	\end{alignat*}
}

We can now argue in the same direction as we did for the packing of vertical items. Let $S^f$ be the optimal fractional solution to this remaining LP which we can find in time $\log_{1+\epsilon}^{O_\epsilon(1)}n$~\cite{cohen2021solving}. We now combine $S^g$ and $S^f$ by rounding down all fractional non-zero variables which by the rank lemma~\cite{lau2011iterative} are at {most} $C_{boxes}(\epsilon)$ {many}.
Since we have at least $2C_{boxes}(\epsilon)/\epsilon$ integral variables, due to our guessing scheme and the fact that each of the fractional variables
{corresponds to a set of items for which each item is at most as profitable as each previously guessed item,}
% the total profit of the fractional variables is at most $\epsilon p(S^g)$ ($p(S^g)$ being the profit due to the guessed items.
the total profit of the fractional variables is at most $\epsilon p(S^g)$ ({where $p(S^g)$ is the profit due to the guessed items)}.
Finally, combining the integral solution with all tiny items that were discarded to restrict our necessary number of guesses, we find an implicit $(1+O(\epsilon)$-{approximate} packing of small items.
Again, by using our rectangle data structure (see Lemma~\ref{lem:data-structure-rec}) we can find the number of tiny items as well as an $(1+\epsilon)$-approximate guess of the total profit of tiny items in time $\mathrm{poly}(\log n)$. 
\end{proof}

Combining the algorithms to compute implicit packings of small and vertical items with the indirect guessing framework, we can prove Theorem~\ref{thm:rec_nr}.

\begin{proof}[Proof of Theorem~\ref{thm:rec_nr}]
We first prove the first part of the theorem.
% First, we set our accuracy $(O_{\epsilon}(1))^{-1}\epsilon$ such that the computed solution has profit at least $(1/2-\epsilon)OPT(\I)$.
Observe that since there is only a constant number of boxes, we can guess their packing into the knapsack in constant time, {once we have determined their heights and widths}.
Our approximation algorithm proceeds in two stages.
\begin{enumerate}[(A)]
\item  \textit{Guessing basic quantities:} The total number of guesses is $(\log n )^{O_{\epsilon}(1)}$. Additionally, since the number of boxes is a constant, we can guess {their relative arrangement in the knapsack}
% 	the packing of these boxes into our knapsack
in constant time.
\item \textit{Indirect guesing framework and construction of packing:} 
\begin{enumerate}[(i)]
	\item \textit{Packing of horizontal and large items into $\H$, $\L$- and $\S$-boxes}: We need $r \in O_{\epsilon,d}(1)$ iterations of the indirect guessing framework, leading to solutions $x^*(1),\dots,x^*(r)$ to the LP-relaxations of $(\mathrm{IP}(\tilde{k}_{1})),...,(\mathrm{IP}(\tilde{k}_{r}))$. This takes time $(\log_{1+\epsilon}n)^{O_\epsilon(1)}$.  Let $\hat{x}(1),\dots,\hat{x}(r)$ be the rounded solutions to $(\mathrm{IP}(\tilde{k}_{1})),...,(\mathrm{IP}(\tilde{k}_{r}))$. These can be found in time $(\log_{1+\epsilon}(n))^{O_{\epsilon}(1)}$ due to Lemma~\ref{lem:rec_IP_sol}. In order to compute an explicit packing of $\L$-boxes, observe that for each triplet $(t,t',t'')$ the preliminary packing only indicates the number of large items chosen of of profit class $\P_t$, height class $H_{t'}$ and width class $W_{t''}$. Again, we may choose these items in non-increasing order of profits in time  $O_{\epsilon}(n+\log^3 n)$ and assign them to the correct $\L$-boxes in time $O_\epsilon(n)$. For $\H$-boxes, we proceed similarly and stack items on top of each other in non-increasing order of widths. This can be done in time $O(n\log n)$ for each box. Finally, for $\S$-boxes, we need to remark that again we may choose any set of items such that for each triplet $(t,t',t'')$ the correct number of items is chosen. Thus, we again may choose the items in non-increasing order of profits and assign the correct number of items for each triplet to each box. This way we only lose a factor of $(1+\epsilon)$ of the profit for each $\S$-box and the selected items still satisfy the conditions necessary to pack them into $\S$ using NFDH. 
	\item \textit{Packing of vertical items:} 	Here, we compute packings of vertical items into $\V$- and $\S$-boxes using Lemma~\ref{lem:rec_lp_vert_all}. For each combination of profit, width and height class, indicated by {a tuple} $(t,t',t'')$, this gives us a value $z_{t,t',t''}$ indicating the number of vertical items of profit class $\P_t$, height class $H_{t'}$ and width class $W_{t''}$ we must choose for our packing. We may choose them in non-increasing order of profits. Afterwards, we must assign the correct number of items to each of the boxes. We may now construct the packing of vertical boxes by packing items from left to right in non-increasing order of heights. The correct set of items can be found in time $O(n+\log^3 n)$ using our rectangle data structure. The packing of $\V$-boxes can be done in time~$O(n)$.
	\item \textit{Packing of small items:} 	Here, we compute the packing of small items into $\S$-boxes using Lemma~\ref{lem:rec_lp_small}. Similar as before, all that is left is to choose the correct number of items from each profit class. Again we select the in non-increasing order of profits and assign the correct number to each box. This process can be done in time $O(n +\log^3(n))$ using our rectangle data structure. Finally, we can pack each $\S$-box in time $O(n \log n)$ using NFDH.
\end{enumerate}
\end{enumerate}

This yields a total running time of $O(n\cdot(\log^4 n))+(\log n)^{O_{\epsilon}(1)}$ since all items can be inserted in time $O(n \log^4 n)$ to initialize the data structure.

Next, we will prove the second statement of the theorem regarding the dynamic algorithm. The insertion and deletion of items in time $O(\log^4 n)$ is due to our rectangle data structure. To output a $(2+\epsilon)$-approximate solution $|ALG|$ in time $O(|ALG|\cdot(\log n))+(\log n)^{O_{\epsilon}(1)}$, we use the algorithm described above with a refinement of the running time since we will choose at most $|ALG|$ items. If one queries for a $(2+\epsilon)$-estimate of the optimal value, we invoke the algorithm above without the construction of the actual packings, i.e., we only compute the preliminary packings using profit, height and width classes, which can be done in time $(\log n)^{O_{\epsilon}(1)}$. Observe that for the tiny items we now need to compute an estimate of the total profit, which we can do in time $O_\epsilon(\log^3n)$ using our rectangle data structure
by querying for the number of items and then multiplying this with the rounded profit. Finally, if one wishes to query whether an item $i \in \I$ is contained in the solution $ALG$, we compute the preliminary packings in time $(\log n)^{O_{\epsilon}(1)}$. Let $(t,t',t'')$ be such that $i$ is of profit type $\P_t$, height type $H_{t'}$ and width type $W_{t''}$ and let $z_{t,t',t''}$ be the total number of items of these types in the preliminary packing. Note that the choice of the triplet $(t,t',t'')$ may depend on which type of item $i$ is.
Then by the construction of our packing we only must check whether $i$ is among the first $z_{t,t',t''}$ when items are ordered in non-increasing order of profits. This can be done using our rectangle data structure by counting the number of items of classes $(t,t',t'')$ that have smaller profit than $i$ which we can do with a single range counting query in time $O(\log n)$ If this value is at least $z_{t,t',t''}$, we answer the query with ``no" and otherwise with ``yes''. This ensures that we give consistent answers between two consecutive updates of the set $\I$.
\end{proof}

\subsection{Rotations allowed}\label{app:rec_r}
We consider now the case where it is allowed to rotate the input rectangles
by 90 degrees and first give a high-level overview of our results. For this case, we present even a $(17/9+\epsilon)$-approximation.
As for the case without rotations, we argue that there is an \emph{easily guessable packing}
based on $\L$-, $\H$-, $\V$- and $\S$-boxes.
Since the items can be rotated, we allow now both horizontal and vertical items
to be assigned to $\H$- and $\V$-boxes.
There might be
be a special (intuitively very large) $\L$-box $B^{*}$ of width $N$.
Also, a special case arises when the packing consists of at most three items.
\begin{lemma}[Near-optimal packing of rectangles with rotations]
	\label{lem:struc_rectangles-rotation} There exists a packing
	with a total profit of at least $(9/17-O(\epsilon))\OPT$,
	which consists of at most three items or which satisfies all of
	the following properties: 
	\begin{enumerate}[i)]
		\item it consists of a set of boxes $\B$ where $\B$ is bounded
		by a value $C_{\mathrm{boxes}}(\epsilon)$ depending only on~$\epsilon$,
		\item each box in $\B$ is a $\L$-, $\H$-, $\V$- or $\S$-box,
		\item possibly, there is an $\L$-box $B^{*}$ that is declared as \emph{special} and that
		satisfies
		%$h_{B^{*}}\ge(1-\epsilon)N$
		% and $w_{B^{*}}\ge(1-\epsilon)N$,
		$w_{B^{*}}=N$; if there is no special box we define for convenience $h_B^* :=0$,
		\item there is a universal set of values $U(\epsilon)$ with $|U(\epsilon)|=O_{\epsilon}(1)$,
		values $k_{0},k_{1},k_{2},\dots,k_{r}\in\mathbb{Z}_{\geq0}$ with
		$k_{0}=0$ and $r\in O_{\epsilon}(1)$, and a value $j_{B}\in\{1,2,\dots,r\}$
		for each box $B\in\B\setminus\{B^{*}\}$ such that
		\begin{enumerate}
			\item if $B$ is a $\H$-box, then 
			% \begin{enumerate}
				% % \item
				$h_{B}=\alpha_{B}\cdot (N-h_{B^*}) $ for some $\alpha_{B}\in U(\epsilon)$,
				$w_{\min}(B)=k_{j_{B}-1}$, and \mbox{$w_{\max}(B)=k_{j_{B}}$},
				% \item $w_{B}=\alpha_{B}\cdot N$ for some $\alpha_{B}\in U(\epsilon)$  \awr{can we omit this case since we can rotate everything anyway?}
				% \end{enumerate}
			\item if $B$ is a $\V$-box, then 
			% \begin{enumerate}
				% \item $w_{B}=\alpha_{B}\cdot N$ for some $\alpha_{B}\in U(\epsilon)$,
				% $h_{\min}(B)=k_{j_{B}-1}$, and $h_{\max}(B)=k_{j_{B}}$ or
				% \item
				$h_{B}=\alpha_{B}\cdot (N-h_{B^*}) $ for some $\alpha_{B}\in U(\epsilon)$, 
				
				%$h_{\min}(B)=k_{j_{B}-1}$, and \mbox{$h_{\max}(B)=k_{j_{B}}$},
				% \end{enumerate}
			\item if $B$ is an $\L$-box and $B \ne B^*$, then $h_{B}=\alpha_{B}\cdot (N-h_{B^*}) $ %or $w_{B}=\alpha_{B}\cdot N$
			for some $\alpha_{B}\in U(\epsilon)$ and the unique item packed inside $B$ satisfies \mbox{$w_{i}=k_{j_{B}}$}.
		\end{enumerate}
		\item let $h^{(1)},\dots, h^{(c)}$ be the distinct heigths of $\V$-boxes in non-decreasing order such that for each $j = 1,\dots,c$ we have that
		\begin{enumerate}
			\item each $\V$-box of height $h^{(1)}$ only contains items of heigth at most $h^{(1)}$ and
			\item each $\V$-box of height $h^{(j)}$ only contains items of height at most $h^{(j)}$ and strictly larger than $h^{(j-1)}$ for $j=2,\dots,c$
			\item for each $\H$-box, $k_{j_B} \leq h^{(j)}$ and $k_{j_B-1} \geq h^{(j-1)}$ for some $j = 1,\dots,c$,
		\end{enumerate}
		% \item the total profit of the packing is at least $(???-O(\epsilon))\OPT$,
		% where $\OPT$ is the profit of an optimal packing for instance
		% $\I$.
		\item if $B$ is an $\S$-box, then $w_B \in \Omega(\es(N-h_{B^*}))$.
	\end{enumerate}
\end{lemma}

Using Lemma~\ref{lem:struc_rectangles-rotation},
we compute an $(17/9+O(\epsilon))$-approximate packing using our indirect
guessing framework. Our algorithm is technically more involved
than our algorithm in the setting without rotations. For example,
to the  $\H$- and $\V$-boxes we can assign both horizontal and vertical items;
for such
% an
input items, we do not know into which box type we need to assign it.
% since
% possibly for some $\H$-box $B$, case (a)i. of Lemma~\ref{lem:struc_rectangles-rotation}
% applies and for some other $\H$-box $B'$ case (a)ii. applies.
Also,
there is possibly the special $\L$-box~$B^{*}$. For $B^{*}$ we
cannot estimate its height via the
% can estimate neither the height nor the width for $B^{*}$ via the
values in $U(\epsilon)$.
% or the values $k_{0},k_{1},k_{2},\dots,k_{r}$.
Instead, we first estimate the sizes of all other boxes, guess their
arrangement inside of the knapsack, and then place $B^{*}$ into the
remaining space. Finally, we devise an extra routine for the special case that the packing consists of at most three items.

\theoremrecrot*

In the following, we present the detailed proofs of Lemma~\ref{lem:struc_rectangles-rotation} and Theorem~\ref{thm:rectangles_rot}. To do this, we need some useful results from the literature. One of these results is the resource augmentation packing lemma (see Lemma~\ref{lem:rec_rs_augment}) which we also used when rotations are not allowed. The next result we need is the so-called resource contraction lemma introduced by Galvez et al.~\cite{galvez2021approximating}. We say that an item $i \in \I$ is \emph{massive} if $h_i \geq (1-\epsilon)N$ and $w_i \geq (1-\epsilon N)$. The resource contraction lemma states the following in the absence of such a massive item.

\begin{lemma}[\cite{galvez2021approximating}]\label{lem:rec_resource_contr}
If a set of items $\I$ does not contain a massive item and can be packed in a box of size $N \times N$, then it is possible to pack a set $\I'$ of profit at least $\frac{1}{2}p(\I)$ into a box of size $N \times (1-\frac{\epsilon}{2})N$ (or into a box of size $(1-\frac{\epsilon}{2})N \times N$) if rotations are allowed.
\end{lemma}

The third result we make use of later is due to Steinberg~\cite{steinberg1997strip} and gives an area based packing guarantee of rectangles.
\begin{theorem}
We are given a set of rectangles $\I$ and a box of size $w \times h$. Let $h_{\max}\leq h$ and $w_{\max}\leq w$ denote be the maximum width and height among the items in $\I$, respectively. Then all items in $\I$ can be packed into the box if
\[
2\sum_{i \in \I}h_iw_i \leq hw - \max\{2h_{\max}-h,0\}\max\{2w_{\max}-w,0\}.
\]
\end{theorem}
We will now prove Lemma~\ref{lem:struc_rectangles-rotation}. Again, we use Lemma~\ref{lem:rec_noint} such that $c(\epsilon)\cdot \es \leq \epsilon^2$ for a large constant $c(\epsilon) \geq 1/\epsilon$.

\begin{proof}[Proof of Lemma~\ref{lem:struc_rectangles-rotation}]
Let $\OPT$ be an optimal packing. Assume first that there is \emph{no massive item} in $\OPT$. In the following, we will construct three candidate packings and later argue that one of them contains a profit of at least $(\frac{7}{19}-O(\epsilon))\OPT$.

\textbf{Candidate packing $\mathrm{A}$:} Consider a ring of width $\frac{\epsilon}{32}N-\frac{\epsilon^2}{64}N$ in the knapsack. Let $\OPT_{ring}$ be the items contained in this ring and $\OPT_{inner}$ be all other items. We now apply the resource contraction lemma to $\OPT_{inner}$. More precisely, we use it to find a subset of items of $\OPT_{inner}$ which can be packed into box of height $(1-\epsilon/2)N$ and width $N$ such that their profit is at least $1/2p(\OPT_{inner})$. Similar as in the case without rotation, we now use the resource augmentation lemma (Lemma~\ref{lem:rec_rs_augment}) with $\epsilon_{ra} = \frac{\epsilon - 2\epsilon^2}{2-\epsilon} \leq \epsilon$. This, allows us to find a packing with profit $(1/2-O(\epsilon))\OPT_{inner}$ inside a box of height $(1-\epsilon/4-2\epsilon^2)$ and width $N$. {The additional space will help us later to round up the heights of boxes and pack all items packed into $\S$-boxes of width at most $\es N$ into a single $\S$-box at the top of the knapsack.} Using the same arguments as in the proof of Lemma~\ref{lem:struc_rectangles} we can transform this packing into one using a constant number of $\V$-, $\L$-, $\H$ and $\S$-boxes satisfying the properties of our structured packing inside the area $[0,N]\times [0,(1-\epsilon/4)N]$. Now, we consider $\OPT_{ring}$. Observe that the total area of the items packed into the outer ring is at most the area of the ring which is
$\frac{\epsilon}{8}N^2 - \frac{\epsilon^2}{256}N^2.$ Furthermore, we can rotate all items such that $h_{\max} \leq \frac{\epsilon}{32}N-\frac{\epsilon^2}{64}N < 
\frac{\epsilon}{4}N-\frac{\epsilon^2}{128}N$ and $w_{\max} \leq N$. Therefore, by Steinberg's theorem all items which were packed into the outer ring can be packed into the area $[0,N]\times [(1-\epsilon/4)N,(1-\epsilon^2/128)N]$. Finally, using the resource augmentation packing lemma with an appropriate choice of $\epsilon_{ra}$ and our transformation described in the proof of Lemma~\ref{lem:struc_rectangles}, we can find a packing of profit $(1-O(\epsilon))\OPT_{ring}$ using a constant number of $\V$-, $\L$-,$\H$- and $\S$-boxes satisfying properties i), ii), iii) and iv) of our structured packing. Observe that the $\V$-boxes already satisfy property v). Thus, we only need to modify the $\H$-boxes such that they satisfy property v) as well. To this end, we split each $\H$-box into a constant number of $\H$-boxes and update the $k_r$ values. To ensure property vi), as described in the proof of Lemma~\ref{lem:struc_rectangles}, we consider all $\S$-boxes of width less than $\es N$. The total volume of these items is at most $C_{boxes}(\epsilon) \es N^2$ and, hence, by an appropriate choice of $\es$ we know that these can be packed into the area at the top of the knapsack of height $\epsilon^2 N$ such that this gives one additional $\S$-box. In this way, we ensure that for each $\S$-box its height and width are at least $\es N$.
Thus, we find a packing with profit at least
\[
\left(1-O(\epsilon)\right)\left(\frac{1}{2}\OPT_{inner}+\OPT_{ring}\right).
\]

\textbf{Candidate packing $\mathrm{B}$:} {The reader may imagine that} $\OPT_{ring} = 0$, implying that there are no items contained inside the outer ring of the knapsack. Let $S_{left}$, $S_{right}$, $S_{top}$ and $S_{bottom}$ be strips of width or height $\frac{\epsilon}{32}-\frac{\epsilon^2}{64}$ on the left, right, top and bottom of the knapsack, respectively. We first define crossing items which are items that touch both $S_{left}$ and $S_{right}$ or $S_{top}$ and $S_{bottom}$ {but which do not intersect with more than two of the strips}. We may assume w.l.o.g. that all crossing items are packed horizontally and we may move these items to the bottom of the knapsack. They can be packed into an $\H$-box of width $N$ since there are no items packed to the left and right of them (recall that $\OPT_{ring} = 0$).
{Thus, the remaining items are packed into the area $[0,N] \times [N',N]$ where $N'$ is the total height of the moved crossing items.
Accordingly, the strip $S_{bottom}$ corresponds now to the area $[0,N] \times [N',N' + (\frac{\epsilon}{32}-\frac{\epsilon^2}{64})\cdot N]$, but it is still disjoint from $S_{top}$.}

Now, we consider the remaining items such that each of them intersects at most two of the four strips and there at most {four} items that intersect two of the strips (at the corners defined by the intersections of the strip boundaries). Let $\OPT_{corner}$ denote these corner items and $\OPT_{rest}$ denote the items which intersect with at most one strip. We now choose one of the strips uniformly at random and delete all items intersecting it. The remaining items can now be packed into the knapsack with a free space of width (height) $\epsilon/32 \cdot N$
{since we deleted the items intersecting one of the strips. This guarantees a}
% and height (width) $N$ and guarantee a
profit of
\[
\frac{3}{4} \OPT_{rest} + \frac{1}{2}\OPT_{corner}.
\]
Using the resource augmentation packing lemma (with an appropriate choice of $\epsilon_{ra}$ {to enable us to round up the heights of the boxes
{and still keep a strip of width $N$ and height $\Omega(\epsilon N)$ empty)}}
we transform
% and form a single $\S$-box containing items from all very small $\S$-boxes}) and transforming
the resulting containers into $\V$-, $\L$-,$\H$- and $\S$-boxes results in a packing satisfying all properties of our structured packing except for property vi.). To ensure that property v) holds for all $\H$-boxes we again split them according to the intervals given by the distinct heights of $\V$-boxes. Our obtained profit is at least
\[
\left(1-O(\epsilon)\right) \left(  \frac{3}{4} \OPT_{rest} + \frac{1}{2}\OPT_{corner} \right).
\]
Again, in a final step, similar to what we described in the proof of Lemma~\ref{lem:struc_rectangles}, we remove all $\S$-boxes of height or width less than $\es N$ and pack their items into a single $\S$-box at the top of the knapsack of height $\Omega(\epsilon N)$.

\textbf{Candidate packing $\mathrm{C}$:} {The reader may imagine that $\OPT_{rest} = 0$ and that hence there are only the corner items}. Then, we either keep all of them or three out of four. This yields a profit of at least
\[
\frac{3}{4}\OPT_{corner}.
\]
% {Possibly rotating the items, we can show that we can }
% This way we can round up the height of at least one of the items to be $\alpha_B N$ for some $\alpha_B \in U(\epsilon)$.
%\awr{Omitted argumentation here: not necessary since we have at most three items}

We now choose the best out of the three candidate packings. Observe that $\OPT_{inner}+\OPT_{ring}=1$ and $\OPT_{rest}+\OPT_{corner} = \OPT_{inner}$.
{Therefore, in any case we obtain a
% Minimizing the possible
profit of at least $(\frac{7}{19}-O(\epsilon))\OPT$.}

Next, assume that there is a single massive item in $\OPT$. We now construct two candidate packings. In the following, we denote $\OPT_{massive}$ as the solution containing only the massive item and $\OPT_{rest}$ as the solution containing all other items.

\textbf{Candidate Packing $\mathrm{Massive A}$}: In this packing, we take out the massive item. Observe that the remaining items (up to rotation) can be packed into a knapsack of width $N$ and height at most $4\epsilon N$. Thus, we can apply the resource augmentation packing lemma to find a packing with profit at least
\[
(1-O(\epsilon))\OPT_{rest}.
\]
Accounting for the fact that we need to round up the height of boxes, we choose $\epsilon_{ra}$ such that the resulting packing (before rounding) can be packed into 
a knapsack of height $(1-(\epsilon/2-2\epsilon^2)N$ and width $N$. We can then use the same arguments as in the proof of Lemma~\ref{lem:struc_rectangles} to find the desired packing. In addition to these arguments we use the final step to remove all small $\S$-boxes similar to this procedure in \textbf{Candidate Packing $\mathrm{A}$} or  \textbf{Candidate Packing $\mathrm{B}$} and modify the $\H$-boxes such that they satisfy property v).

\textbf{Candidate Packing $\mathrm{Massive B}$}: Next, we find a packing which combines the massive item with a subset of the other items. Again let $S_{left}$, $S_{right}$, $S_{top}$ and $S_{bottom}$ be the strips of the knapsack on the left, right, top or bottom, respectively, which are not filled by the massive item. We will now construct a packing in a similar way as in~\cite{galvez2021approximating} with slight modifications necessary to have a packing using our four types of boxes as well as taking into account the rounding of the height of each box. We consider the strip with the largest amount of profit among the four strips (note that strips may share some profit). But there is one strip such that the profit of all items contained inside this strip is at least
\[
\OPT_{strip} \geq \frac{1}{4}\OPT_{rest}.
\]
We may assume w.l.o.g. that this is one of the {horizontal} strips ($S_{top}$ or $S_{bottom}$) as otherwise we may rotate the massive item and the chosen strip such that the strip is {horizontally} positioned {above or below} the massive item. In particular,
{we can assume that}
the massive item $i$ is now packed into the area $[0,N]\times [N-h_i,N]$ and the items that were previously packed into the chosen strip are packed into the area $[0,N]\times [0,N]$. We now treat this area as a single knapsack and use the proof of Lemma~\ref{lem:struc_rectangles} to find a packing into constantly many $\V$-,$\H$-,$\L$- and $\S$-boxes satisfying the properties of our structured packing such that the profit is at least
\[
\left(\frac{1}{2}-O(\epsilon)\right) \OPT_{strip} \geq  \left(\frac{1}{8}-O(\epsilon)\right)\OPT_{rest}.
\] 
{When invoking Lemma~\ref{lem:struc_rectangles}, we make sure that the height of the resulting boxes are rounded up to values {of the form} $\alpha_B (N-h_i)$; {also, we make sure that
there is enough space at the top of the knapsack to move all items in very small $\S$-boxes (i.e., boxes $B $ with width
$w_B \in O({\es(N-h_{B^*}}))$}
%  $w_B \in O_{\epsilon}(\aw{N-h_{B^*}})$})
there.}
Hence, this packing in combination with the massive item which we place in a $\L$-box of height $N$ yields a profit of at least
\[
\OPT_{massive}+\left(\frac{1}{8}-O(\epsilon)\right)\OPT_{rest}.
\]
Using the fact that $\OPT_{massive}+\OPT_{rest} = 1$, chosing the most profitable of the two candidate packings yields a guarantee of approximately $(\frac{8}{15}-O(\epsilon))\OPT$.
\end{proof}

\subsubsection{Computing a packing}
In the following, we explain how to compute a packing corresponding to Lemma~\ref{lem:struc_rectangles-rotation}, {the goal being to prove} Theorem~\ref{thm:rectangles_rot}. In the following, we assume w.l.o.g. (due to rotation) that there are no vertical items such that we only have small, horizontal and large items. A lot of the steps needed to compute our packing are similar to the ideas used in the non-rotational case. In the following, we will focus on the key differences.
First, we will explain how to compute a packing of at most three items.
\begin{lemma}
If the structured packing with profit at least $(9/17-O(\epsilon))\OPT$ consists of at most $3$ items, it can be computed in time $(\log_{1+\epsilon}n)^{{O(1)}}$.
\end{lemma}
\begin{proof}
Let $i_1, i _2$ and $i_3$ denote the items in the packing. {Since we have at most three items, one can show that
there is one item that we can we can assign to an $\L$-box of height $N$ such that this box does not intersect any of the other two items.
W.l.o.g. suppose that $i_1$ is this item.} Then, we guess in time $O(\log_{1+\epsilon}n)$ the profit type of $i_1$. Losing only a factor of $1+\epsilon$ of the profit we now choose the item of this profit type with the smallest width or height (due to rotation) which we can find in time $O(\log^3 n)$ using our rectangle data structure and a binary search. This will leave enough space for the other two items. For each of these items we can again guess the profit class in time $O(\log_{1+\epsilon}n)$ and the orientation of the item which will then indicate whether we take the item with the smallest width or smallest height in this profit class. Using our rectangle data structure we can do this in time $O(\log^3 n)$.
\end{proof}
We now proceed with how to compute a packing corresponding to Lemma~\ref{lem:struc_rectangles-rotation} if there are more than three items. In this case, we use the following structure resembling our results for hypercubes and rectangles without rotation. We first guess some basic quantities. Then, we pack the small items into $\S$-boxes. Then, we use our indirect guessing framework to pack horizontal and large items into $\H$-$\V$-,$\L$- and $\S$-boxes.

\paragraph{Guessing basic quantities.} We start by guessing similar quantities as in the setting without rotation. For completeness, we give the details here. We again disregard all items with profit less than $\epsilon \frac{p_{\max}}{n}$ such that we have $O(\log_{1+\epsilon}n)$ profit classes of the form $\P_{t}:= \{i \in \I: p_i \in [(1+\epsilon)^{t},(1+\epsilon)^{t+1})]\}$.
% Losing only a factor of $(1+\epsilon)^{-1}$ of the profit, we may guess for each $t$ and $B$ the number of items used in the optimal packing as powers of $(1+\epsilon)$. Denote one such value by $n^g_{t,B}$. Guessing all these values can be done in time $O_\epsilon(\log^{C_{boxes}(\epsilon)} n)$.
% \awr{Commented out guessing $n^g_{t,B}$ that I think we cannot guess fast enough. But it is never uses.}

Next, we guess how many boxes of each type there are in $\B$. This amounts to a total of $O_{\epsilon}(1)$ many possibilities. For each of the boxes $B$ we guess its height by guessing $\alpha_B$, using that $|U(\epsilon)|\le O_{\epsilon}(1)$.
%
% for which we use $U(\epsilon):=\{i\cdot \frac{\epsilon}{2C_{boxes}(\epsilon)}: i=1,\dots, \frac{2C_{boxes}(\epsilon)}{\epsilon}\}$ such that there are $O_\epsilon(1)$ many possibilities for each box.

The remaining quantities are the same as in the non-rotational setting with a few exceptions. We do not need to guess values $\hat{a}_{B,\V}$ for an $\S$-box $B$ since we now treat all vertical items as horizontal items which will be packed in the indirect guessing framework. For each $\S$-box $B\in\B$, we guess an estimate of the width $w_B$. Formally, for each $\S$-box $B\in\B$, we guess the value $\left\lfloor w_B \right\rfloor _{1+\epsilon}$. Observe that there are $O_\epsilon(1)$ many since for each $\S$-box we know that its width is in the interval $[\es N, N]$ and $\es$ is bounded from {below by some constant depending only on $\epsilon$. Furthermore, consider the distinct heights of $\V$-boxes denoted by $h^{(1)},\dots, h^{(c)}$ and $h^{(0)}:=0$. For each range $(h^{(j-1)},h^{(j)}]$ we guess the total profit of horizontal items with width in this range used in $\V$-boxes in the structured packing. More, precisely for each $j=1,\dots,c$ we guess a value $\hat{p}_{\V}(j)$ as a power of $1+\epsilon$ which we can do in time $O_\epsilon(1)$.

\paragraph{Packing small items into $\S$-boxes.} Observe that in contrast to the non-rotational setting, we now know the guessed widths of $\S$-boxes in advance and can, therefore, start by packing small items into $\S$-boxes. For this step, we proceed in a similar fashion as in the non-rotation case with the only exception that now a small item $i$ may be placed into an $\S$-box $B$ if $w_i \leq \epsilon w_B$ and $h_i \leq \epsilon h_B$ or $h_i \leq \epsilon w_B$ and $w_i \leq \epsilon h_B$. Let $\I_s$ be the set of small items and $\B_\S$ be the set of $\S$-boxes. Again, we use the notion of height and width classes. We define a height class $H_{t'}=\{i\in \I_s:h_{i}\in[(1+\epsilon)^{t'},(1+\epsilon)^{t'+1})\}$
for each $t'\in \mathcal{T}' = \{\lfloor \log_{1+\epsilon}(h_{\min}(\I_s))\rfloor,\dots,\lceil \log_{1+\epsilon}(h_{\max}(\I_s))\rceil\}$. Furthermore, we define a width class $W_{t^{''}}=\{i\in \I_s:w_{i}\in[(1+\epsilon)^{t^{''}},(1+\epsilon)^{t^{''}+1})\}$ for each  $t^{''} \in \mathcal{T}^{''} = \{\lfloor \log_{1+\epsilon}(w_{\min}(\I_s))\rfloor,\dots,\lceil \log_{1+\epsilon}(w_{\max}(\I_s))\rceil\}$.

We denote by $\hat{h}(t'):=(1+\epsilon)^{t'+1}$ and $\hat{w}(t^{''}):=(1+\epsilon)^{t^{''}+1}$ the rounded height and width, respectively. Let $\mathcal{T}:= \{(t,t',t''): t \in \mathcal{T_P} \wedge t' \in \mathcal{T}_H \wedge t'' \in \mathcal{T}_W\}$. We denote by $n_{t,t',t^{''}}$ the number of items corresponding to the tuple $(t,t',t^{''})$. Now, we compute a packing based on the following IP. We denote by {$\I(B_\S)$} the set of small items which may be packed into box $B$. In the following, IP we only consider variables for which the height class and width class is compatible with a box $B$. In particular, we only consider a variable $x_{t,t',t'',B}$ if $\hat{h}_{t'} \leq 2\epsilon h_B$ and $\hat{w}_{t''} \leq 2\epsilon w_B$ or $\hat{h}_{t'} \leq 2\epsilon w_B$ and $\hat{w}_{t''} \leq 2\epsilon h_B$.

{
\small
\begin{alignat*}{3}
(\mathrm{IP}(\S))\quad& \text{max} 	& \displaystyle \sum_{(t,t',t'') \in \mathcal{T}}\sum_{B \in \B_\S} x_{t,t',t'',B} p(t)			& 			& \quad & \\
& &\displaystyle \sum_{(t,t',t'') \in \mathcal{T}} x_{t,t',t'',B} \hat{h}(t')\hat{w}(t'')	& \leq a^{\mathrm{small}}_{B}h_Bw_B & 		& \forall B \in \B_\S(\ell+1) \\
&				& \displaystyle\sum_{B \in  \B_\S} x_{t,t',t'',B}								& \leq n_{t,t',t''}	& 		& \forall  (t,t',t'') \in \mathcal{T} \\
&				& x_{t,t',t'',B}								& \in \mathbb{N}_{0}& 		&\forall t  (t,t',t'') \in \mathcal{T}, B \in  \B_\S
\end{alignat*}
}
We now show how to find a $(1+O(\epsilon))$-approximate implicit solution to the problem of packing small items into $\S$-boxes using the same techniques as used in the non-rotational setting. More specifically, we first guess the most profitable items to gain a partial integral solution to $(\mathrm{IP}(\S))$. Then, we show how we can restrict the number of relevant width and heigth classes to $O(\log_{1+\epsilon}n)$ each such that we can solve the remaining LP in time $(\log_{1+\epsilon}n)^{O(1)}$~\cite{cohen2021solving}. {Here, we again use the fact that all
{$\S$-boxes} are sufficiently high such that all small items whose height is much less than that of the {highest $\S$-box} can be packed into this box, while only losing a {factor of $1+\epsilon$ with respect to the
profit of a fractional packing of items into this box.}} Lastly, we take {an optimal solution to this LP} and round all fractional non-zero variables down such that in combination with the guessed solution and the discarded items we obtain a feasible solution to $(\mathrm{IP}(\S))$. Finally, we use the rank lemma~\cite{lau2011iterative} to argue that {due to this we lose only a factor of $1+\epsilon$} of the profit due the rounding. {Finally,} we add the profit of {the} tiny items {that we had} discarded to restrict the number of height and width classes.
{Overall, our solution is a $(1+O(\epsilon))$-approximate solution to $(\mathrm{IP}(\S))$.} 
% }
% %
% We now show how to find a $(1+O(\epsilon))$-approximate implicit solution to the problem of packing small items into $\S$-boxes using the same techniques as used in the non-rotational setting. More specifically, we first guess the most profitable items to gain a partial integral solution to $(\mathrm{IP}(\S))$. Then, we show how we can restrict the number of relevant width and heigth classes to $O(\log_{1+\epsilon}n)$ each such that we can solve the remaining LP in time $\log_{1+\epsilon}^{O(1)}n$~\cite{cohen2021solving}. {Here, we again use the fact that all boxes are sufficiently high such that all small items whose hight is much less than that of the largest box can be packed into this box while only losing an $\epsilon$-fraction of the profit of a fractional packing into this box.} Lastly, we take {an optimal solution to this LP} and round all fractional non-zero variables down such that in combination with the guessed solution and the discarded items we obtain a feasible solution to $(\mathrm{IP}(\S))$. Finally, we use the rank lemma~\cite{lau2011iterative} to argue that we only lose an $\epsilon$-fraction of the profit due the rounding and add the profit of tiny items discarded to restrict the number of height and width classes such that the total solution is a $(1+O(\epsilon))$-approximate solution to $(\mathrm{IP}(\S))$.
\begin{lemma}\label{lem:rec_rot_lp_small}
There is an algorithm with a running time of {$(\log_{1+\epsilon}(n))^{O_{\epsilon}(1)}$}
that computes an {implicit} $(1{+}O(\epsilon))$-approximate solution to $(\mathrm{IP}(\S))$.
\end{lemma}

\paragraph{Indirect guessing framework to pack horizontal and large items.} In contrast to the setting without rotation, we now must consider all boxes when packing horizontal and large items since horizontal items may be packed into $\H$-,$\V$- and $\S$-boxes while large items may be packed into $\L$- and $\S$-boxes. In particular, the fact that horizontal items may be packed into $\H$- as well as $\V$-boxes requires a more involved and technical procedure to compute a packing compared to the indirect guessing framework applied in the non-rotational setting. Hereto, we make use of property v) of our structured packing (Lemma~\ref{lem:struc_rectangles-rotation}) which allows us to consider each interval $(h^{(j-1)},h^{(j)}]$, with $h^{(1)},\dots, h^{(c)}$ being the distinct heights of $\V$-boxes and $h^{(0)}:= 0$, separately in our indirect guessing framework.

We will now explain how our to use the indirect guessing framework for each of these intervals separately. To this end, consider the interval $(h^{(0)}, h^{(1)}]$. For all other intervals the procedure will be the same. Let $k_{1},k_{2},\dots,k_{r^{(1)}}$ be the $k_r$-values within the interval $(h^{(0)}, h^{(1)}]$. The objective is again to determine these values in polylogarithmic time which cannot be achieved by guessing them {exactly} since there are $N$ options for each of them. A key difference between the rotational setting and the non-rotational setting is that $\V$- and $\H$-boxes must both be considered for horizontal items and, therefore, $\V$-boxes are relevant for the indirect guessing framework. Since we do not know the widths of the $\V$-boxes this requires a step before the indirect guessing framework.

We start by guessing the $C_{boxes}(\epsilon)/\epsilon$ most profitable horizontal items packed into the $\V$-boxes of height $h^{(j)}$. To do this we guess the profit class for each of these items in time $\log_{1+\epsilon}^{O_{\epsilon}(1)}n$. Let $n^g_{t}$ be the number of guessed items for each profit class. Using our rectangle data structure we consider the $n^g_{t}$ items of smallest width of profit class $t$ which can be found in time $O(\log^4 n)$. For each of these items, we guess which box it is assigned to in the structured packing. If an item is assigned to a $\V$-box, $\H$-box or $\S$-box in the structured packing, we also assign this item to this box. If, however, an item is not used in the structured packing then we assign it to a $\V$-box losing only a factor of $(1+\epsilon)$ of the profit while making sure that the picked item is at most as wide as the item used in the structured packing.

After finding this partial packing, we distinguish between two cases based on the widths of $\V$-boxes of height $h^{(j)}$. Let $h_{B}^{min}$ be the minimum height of all $\H$-boxes relevant for the current iteration of the algorithm. Consider the $\V$-box of height $h^{(j)}$ with largest width. If this is width at least $\frac{\epsilon}{C_{boxes}(\epsilon)}h_{B}^{min}$, we can guess it as a power of $(1+\epsilon)$ in time $O_{\epsilon}(1)$ since  $h_{B}^{min} \in \Omega(\es N)$. Let $B^*$ be this box and denote by $\hat{w}_{B^*}$ the guessed width. We may assume now that all other $\V$-boxes of height $h^{(j)}$ have width at least $\frac{\epsilon}{C_{boxes}(\epsilon)}\hat{w}_{B^*}$ by reserving an $\epsilon$-fraction of the width of $B^*$ for all items packed into $\V$-boxes of width less than $\frac{\epsilon}{C_{boxes}(\epsilon)}\hat{w}_{B^*}$. Thus, we can also guess the width of the remaining boxes each in time $O_\epsilon(1)$. Now, suppose the maximum width of all $\V$-boxes of height $h^{(j)}$ is less than $\frac{\epsilon}{C_{boxes}(\epsilon)}h_{B}^{min}$. In this case, we reserve an $\epsilon$-fraction in each $\H$-box for all items packed into $\V$-boxes which can also fit into $\H$-boxes. Such that all remaining items packed into the $\V$-boxes do not fit into any of the $\H$-boxes and, therefore, we can pack the $\V$-boxes after packing the $\H$-boxes using our indirect guessing framework.

We now explain how to use the indirect guessing framework. We will assume here that we must also consider $\V$-boxes with the guessed widths following the above procedure. Note that if we are in the second case described above this is not necessary and we will later describe how to compute a packing of the $\V$-boxes in this case. We start by defining $\tilde{k}_{0}:=0$ and will compute values $\tilde{k}_{1},\tilde{k}_{2},\dots,\tilde{k}_{r^{(1)}}$ to use instead of the values $k_{0},k_{1},k_{2},\dots,k_{r^{(1)}}$. With these values we get a partition of $\I$ into sets $\tilde{\I}_{j}:=\{i\in\I:w_{i}\in(\tilde{k}_{j-1},\tilde{k}_{j}]\}$. Now, for each $j$ we want to pack items from $\tilde{\I}_{j}$ into the space reserved for these items in the packing from Lemma~\ref{lem:struc_rectangles-rotation}.We will choose the values $\tilde{k}_{1},\tilde{k}_{2},\dots,\tilde{k}_{r^{(1)}}$ such that in this way, we obtain almost the same profit. On the other hand, we will ensure that $\tilde{k}_{j}\le k_{j}$ for each $j\in[r^{(1)}]$.

We work in $r^{(1)}$ iterations. We define $\tilde{k}_{0}:=0$. Suppose
inductively that we have determined $\ell$ values $\tilde{k}_{1},\tilde{k}_{2},\dots,\tilde{k}_{\ell}$
already for some $\ell\in\{0,1,...,r^{(1)}-1\}$ such that $\tilde{k}_{\ell}\le k_{\ell}$.
We want to compute $\tilde{k}_{\ell+1}$. We can assume w.l.o.g.~that
$k_{\ell+1}$ equals $w_{i}$ for some item $i\in\I$. We do {a} binary
search on the set $W(\ell):=\{w_{i}:i\in\I\wedge w_{i}>\tilde{k}_{\ell}\}$,
using our rectangle data structure. For each candidate value $w \in W(\ell)$,
we estimate the possible profit due to items in $\tilde{\I}_{\ell+1}$
if we define $\tilde{k}_{\ell+1}:=w$. We want to find such a value
$w$ such that the obtained profit from the set $\tilde{\I}_{\ell+1}$
equals essentially $\hat{p}(\ell+1)$. In the following, we denote by $\B_{\H}(\ell+1)$ the set of $\H$-boxes for which $j_B = \ell+1$, by $\B_{\V}(\ell+1)$ the set of $\V$-boxes for which $j_B = \ell+1$, by $\B_{\V}(\ell+1)$ the set of $\V$-boxes for which $j_B = \ell+1$ and by $\B_\S$ the set of $\S$-boxes.

We describe now how we estimate the obtained profit for one specific
choice of $w\in W(\ell)$. We try to pack items from $\tilde{\I}_{\ell+1}(w):=\left\{ i\in\I:w_{i}\in(\tilde{k}_{\ell},w]\right\} $
into
\begin{itemize}
\item the $\H$-boxes $B\in\B_{\H}(\ell+1)$, 
\item the $\V$-boxes $B\in\B_{\V}(\ell+1)$, 
\item the $\L$-boxes $B\in\B_{\L}(\ell+1)$ and 
\item the $\S$-boxes, where for each $\S$-box $B\in\B_{\S}$, we use an area
of $\widehat{a}_{B,\ell+1}\cdot h_Bw_B$ and ensure that we pack
only items $i\in\tilde{\I}_{\ell+1}(w)$ for which $w_{i}\leq \epsilon w_B$ and $h_i \leq \epsilon h_B$ or $h_{i}\leq \epsilon w_B$ and $w_i \leq \epsilon h_B$ (due to rotation).
\end{itemize}
We do this via an IP where we again use the notion of width and height types. Denote by $\B(\ell+1)$. Then, losing only a factor of $(1+\epsilon)$, we can formulate the problem as the following IP. Note that for each box $B$ we must only consider height and width classes that are relevant, e.g. for large boxes only those corresponding to large items. Here, $\hat{w}_B$ denote the guessed widths from the previous step. Furthermore, we assume that the right hand-side for $\H$- and $\S$-boxes already take into account that we potentially assigned items to these boxes in our partial solution above. We again use the notation of profit, width and height class introduced in the non-rotational setting.

{
\small
\begin{alignat*}{3}
(\mathrm{IP}(w))\quad& \text{max} 	& \displaystyle  \sum_{(t,t',t'') \in \mathcal{T}}\sum_{B \in \B(\ell+1)} x_{t,t',t'',B} p(t)			& 			& \quad & \\
& \text{s.t.} & \displaystyle \sum_{(t,t',t'') \in \mathcal{T}} x_{t,t',t'',B} \hat{h}(t')	& \leq h_B& 		& \forall B \in \B_{\H}(\ell+1) \\
& & \displaystyle \sum_{(t,t',t'') \in \mathcal{T}} x_{t,t',t'',B} \hat{w}(t'')	& \leq \hat{w}_B& 		& \forall B \in \B_{\V}(\ell+1) \\
&  & \displaystyle  \sum_{(t,t',t'') \in \mathcal{T}} x_{t,t',t'',B} 	& \leq 1& 		& \forall B \in \B_{\L}(\ell+1) \\
& & \displaystyle  \sum_{(t,t',t'') \in \mathcal{T}} x_{t,t',t'',B} \hat{h}(t')\hat{w}(t'')	& \leq a_{B,\ell+1}h_Bw_B	& 		& \forall B \in \B_{\S}\\
&				& \displaystyle\sum_{B \in \B(\ell+1)} x_{t,t',t'',B}								& \leq n_{t,t',t''}	& 		& \forall (t,t',t'') \in \mathcal{T} \\
&				& x_{t,t',t'',B}								& \in \mathbb{N}_{0}& 		&\forall (t,t',t'') \in \mathcal{T}, B \in \B(\ell+1)\\
\end{alignat*}
}
We now describe how to find an implicit $(1+O(\epsilon))$-approximate solution to this IP following the same ideas as used in the setting of hypercubes and rectangles without rotations.

\begin{lemma}\label{lem:rec_IP_sol_rot}
There is an algorithm with a running time of {$(\log_{1+\epsilon}(n))^{O(1)}$}
that computes an {implicit} $(1{+}O(\epsilon))$-approximate solution for \textup{$(\mathrm{IP}(w))$};
we denote by $q(w)$ the value of this solution. For two values $w,w'$
with $w\le w'$ we have that $q(w)\le q(w')$.
\end{lemma}
\begin{proof}
We start by guessing the $C_{boxes}(\epsilon)/\epsilon$ most profitable large {and} horizontal items in the solution. First, consider the large items. For these, we first guess the profit, height and width class which can be done in {time} $O(\log_{1+\epsilon}n)$, $O_{\epsilon}(1)$ and $O_{\epsilon}(1)$, respectively. For each triplet $(t,t',t'')$ this gives us a value $n_{t,t',t''}^g$ indicating the number of items for this combination of classes. We now find the $n_{t,t',t''}^g$ items corresponding to $(t,t',t'')$ with lowest height. This can be done in time $O(\log^2 n)$ using our rectangle data structure. For each of these items we guess which box it must be assigned to (among the $\L$- and $\S$-boxes) which takes another $O_{\epsilon}(1)$guesses. Now, we do the same for horizontal items as follows. We guess the profit class of each of the $C_{boxes}(\epsilon)/\epsilon$ most profitable items as well as the width class which can both be done in $O(\log_{1+\epsilon}n)$. Then, for each combination of profit and width class we get a value $n_{t,t''}^g$ and look at the  $n_{t,t''}^g$ items of lowest height and for each of them guess which box it must be assigned to. Observe that in this way we make sure that we only lose a factor $1+\epsilon$ of the profit while also making sure that the picked items are at most as high (in case of $\H$-boxes) or at most as wide (in case of $\V$-boxes) as the correct item used in the structured packing. This again takes time $O(\log^2 n)$ using our rectangle data structure. After these procedures, we have a partial solution to $(\mathrm{IP}(w))$ denoted by $S^g$. For each $\L$-box, let $n^g_B$ be the number of items assigned to $B$ by $S^g$. For each $\S$-box let $\mathrm{area}^g_B$ be the total area of the items assigned to $B$ by $S^g$ (using the rounded width and height). Similarly, for $\H$- and $\V$-boxes let $h^g_B$ and $w^g_B$ be the height of the items assigned to $B$. Note that also for $\V$-boxes we are interested in the height due to rotation. We now consider an LP-relaxation of $(\mathrm{IP}(w))$ with updated right hand-sides.

{
\footnotesize
\begin{alignat*}{3}
(\mathrm{LP}(w))\quad& \text{max} 	& \displaystyle \sum_{(t,t',t'') \in \mathcal{T}}\sum_{B \in \B(\ell+1)} x_{t,t',t''',B} p(t)			& 			& \quad & \\
& \text{s.t.} & \displaystyle  \sum_{(t,t',t'') \in \mathcal{T}} x_{t,t',t'',B} \hat{h}(t')	& \leq h_B-h^g_B& 		& \forall B \in \B_{\H}(\ell+1) \\
& & \displaystyle \sum_{(t,t',t'') \in \mathcal{T}} x_{t,t',t'',B} \hat{w}(t'')	& \leq \hat{w}_B - w^g_B& 		& \forall B \in \B_{\V}(\ell+1) \\
&  & \displaystyle  \sum_{(t,t',t'') \in \mathcal{T}} x_{t,t',t'',B} 	& \leq 1-n^g_B& 		& \forall B \in \B_{\L}(\ell+1) \\
& & \displaystyle  \sum_{(t,t',t'') \in \mathcal{T}} x_{t,t',t'',B} \hat{h}(t')\hat{w}(t'')	& \leq a_{B,\ell+1}h_Bw_B - \mathrm{area}^g_B	& 		& \forall B \in \B_{\S}\\
&				& \displaystyle\sum_{B \in \B(\ell+1)} x_{t,t',t'',B}								& \leq n_{t,t',t''}-n_{t,t',t''}^g& 		& \forall   (t,t',t'') \in \mathcal{T} \\
&				& x_{t,t',t'',B}								& \geq 0& 		&\forall  (t,t',t'') \in \mathcal{T}, B \in \B(\ell+1)
\end{alignat*}
}
The goal is now to solve and round $(\mathrm{LP}(w))$ in polylogarithmic time. Observe that for large items we have at most $O_{\epsilon}(1)$ height and width classes and for horizontal items we have at most $O_{\epsilon}(1)$ width classes. Thus, we only need to restrict the number of height classes. To this end, let ${\hat{b}}_{max}$ be the maximum value among all {values} $h_B-h_B^g$ and $\widehat{w}_B$. We know that all horizontal items we consider must have height at most $\hat{b}_{max}$.
%\awr{$hw_{max}$ is not defined, no? Do you mean $\hat{b}_{max}$ (no changed to ${\hat{b}}$)?} \todo{Changed :)}
Reserving a space of all items of height less than $\frac{\epsilon \hat{b}_{\max}}{n}$ we can restrict the number of height classes to $O(\log_{1+\epsilon}n)$ many. Hereto, we reduce the available height of all $\H$-boxes and width of all $\V$-boxes by a factor $1-\epsilon$. Let $S^f$ be a fractional optimal solution to $(\mathrm{LP}(w))$ taking into account the reduction of space for height and width boxes. {We can compute this solution in time
$(\log n)^{O_{\epsilon}(1)}$~\cite{cohen2021solving}}.
{Due to the rank lemma, we know that we have at most $C_{boxes}(\epsilon)$ fractional variables in our solution.
Due to our guessed solution $S^g$, we know that}
% We now use the fact that we have $C_{boxes}(\epsilon)$ integer variables due to $S^g$ to argue (by the rank lemma~\cite{lau2011iterative}) that
rounding down all fractional non-zero variables of $S^f$ {to the next smaller integer} leads to a loss of profit of at most $\epsilon p(S^g)$. Thus, taking $S^g$ together with the integral variables of $S^f$ as well as the discarded items used to reduce the number of height classes for horizontal items, we obtain an $(1+O(\epsilon))$-approximate solution to $(\mathrm{IP}(w))$ in time $(\log_{1+\epsilon} n)^{O_\epsilon(1)}$. The second part of the lemma follows from the fact that a solution for $w$ is also feasible for $w'$.
\end{proof}
Again, we define $\tilde{k}_{\ell+1}$ as the smallest value $w\in W(\ell)$
for which $q(w)\ge (1-\epsilon)\hat{p}(\ell+1)$. From  this and the inductive assumption that $\tilde{k}_{\ell}\le k_{\ell}$ the following statement holds. 
\begin{lemma}\label{lem:rec_induc_kr_rot}
We have that $\tilde{k}_{\ell+1}\le k_{\ell+1}$.
\end{lemma}
We now repeat the indirect guessing framework for each $j=1,\dots, r^{(1)}$. Finally, it could be that we still need to compute the remaining implicit packing of $\V$-boxes (in case we could not guess their widths before our indirect guessing framework). Here, we will use that we guessed the profit obtained by items packed into these boxes upfront. Recall, that this profit is given by $\hat{p}_{\V}(1)$ for the current range $(h^{(0)},h^{(1)}]$. To find this implicit packing, we will find the minimum total width necessary to achieve the desired profit. To this end, we use our rectangle data structure as well as the balanced binary search tree of item densities to find the smallest density $d^*$ such that packing all horizontal items with height in $(\tilde{k}_{r^{(1)}},h^{(1)}]$ up to this density yield a total profit of at least  $\hat{p}_{\V}(1)$. This can be done in time $O(\log^4 n)$. Hence, we now know the total width $w^*$ of all considered $\V$-boxes (with height in $(h^{(0)},h^{(1)}]$) as well as an implicit packing stating that all horizontal items with height in $(\tilde{k}_{r^{(1)}},h^{(1)}]$ and density of at most $d^*$ are packed. It remains to guess the width of each individual $\V$-boxes as well as implicit packings of these boxes. In the analysis as well as the construction of the implicit packing of the individual boxes we assume w.l.o.g. that the items are packed into the auxilliary $\V$-box of width $w^*$ in non-increasing order of profits. Let $B_1,\dots,B_{\ell}$ be the considered $\V$-boxes. We start by guessing the values $\hat{w}_B := \lfloor \frac{w_B}{\frac{\epsilon}{\ell} w^*} \rfloor \frac{\epsilon}{\ell}w^*$ for each box $B$. First, observe that $\sum_{j=1}^\ell \hat{w}_{B_j} \geq (1-\epsilon)w^*$ and since the items are packed in non-increasing order of profits all items that are packed into the width of  $\sum_{j=1}^\ell \hat{w}_{B_j}$ yield a profit of at least $(1-\epsilon)$. This packing might include a single fractional item. We now split this box into $B_1,\dots,B_{\ell}$ according to the guessed widths. This way, we lose at most one item per box. However, since initially (before the indirect guessing framework) we guess the $C_{boxes}(\epsilon)/\epsilon$ most profitable items packed into $\V$-boxes and each of the items we cannot pick has a profit strictly less than these guessed items, the total profit we lose is at most an $\epsilon$-fraction of  $\hat{p}_{\V}(1)$. Thus, we can construct an implicit packing of the individual boxes as follows. For box $B_1$, we use our rectangle data structure in combination with the binary search trees to find a value $p^{(1)}$ such that all items with profit at least $p$, density at most $d^*$ and height in $(\tilde{k}_{r^{(1)}},h^{(1)}]$ have total width at most $\hat{w}_{B_1}$ while increasing the profit would violate the width of the box. Then, we repeat this to find values $p^{(2)}, \dots,p^{(\ell)}$. This total procedure can be done in time $O_{\epsilon}(\log^4 n)$.

Repeating this procedure for all intervals $(h^{(j-1)},h^{(j)}]$, gives a complete implicit packing in polylogarithmic time.

Finally, let $i^*$ be the item that is packed in the box $B^*$ in the packing due to Lemma~\ref{lem:struc_rectangles-rotation}. If there are at least $1/\epsilon$ items from the same profit class as $i^*$ in the packing due to Lemma~\ref{lem:struc_rectangles-rotation}, then we can simply omit $i^*$ while losing only a factor of $1+\epsilon$ in the profit.
Otherwise, we guess the profit class of $i^*$ and identify the $1/\epsilon$ input items of this profit class of smallest width (assuming w.l.o.g. that in the input they are rotated such that their heights are not smaller than their widths).
If all of them have already been selected previously by our routines above, then we are done, losing at most a factor of $1+\epsilon$.
Otherwise, we can assume w.l.o.g. that $i^*$ is among these items, we guess $i^*$ in time $1/\epsilon$, and place $i^*$
in our packing.

% we assign the remaining space to the box $B^*$ such that $w_{B^*}=N$. We choose the height $h_{B^*}$ of $B^*$ as high as the other items allo}
Together with the algorithm to compute implicit packings of small items and the algorithm to compute implicit packings of horizontal and large items, we can prove Theorem~\ref{thm:rectangles_rot}. Note that we again proceed in a similar fashion as in the proof of Theorem~\ref{thm:rec_nr}.

\begin{proof}[Proof of Theorem~\ref{thm:rectangles_rot}]
We first prove the first part of the theorem.
% First, we set our accuracy $(O_{\epsilon}(1))^{-1}\epsilon$ such that the computed solution has profit at least $(9/17-\epsilon)OPT(\I)$.
Observe that since there is only a constant number of boxes, we can guess their {relative arrangement inside} the knapsack in constant time.
Our approximation algorithm proceeds in three stages.
\begin{enumerate}[A]
\item  \textit{Guessing basic quantities:} The total number of guesses is $(\log n )^{O_{\epsilon}(1)}$. Additionally, since the number of boxes is a constant, we can guess their {relative arrangement inside} in our knapsack in constant time.
\item \textit{Indirect guesing framework and construction of packing:} 
\begin{enumerate}[i]
\item \textit{Packing of small items:} Here, we compute the implicit solution due to Lemma~\ref{lem:rec_rot_lp_small}. This gives us for each profit, width and height class of small items {(indicated by a tuple $(t,t',t'')$)} a value $z_{t,t',t''}$ which tells us how many items of this class combination we must pack. Assigning these items in non-incresing order of profits to the $\S$-boxes such that each $\S$-box receives the correct number of items can be done in time $O(\log^3(n) + n)$ using our rectangle data structure and the balanced binary search tree for the profits. Furthermore, we can find the set of tiny items which we discarded in time $O(\log^2(n) + n)$, these will be assigned to the largest $\S$-box.
\item \textit{Packing of horizontal and large items:}  For each guess we need $r \in O_{\epsilon,d}(1)$ iterations of the indirect guessing framework, leading to solutions $x^*(1),\dots,x^*(r)$ to the LP-relaxations of $(\mathrm{IP}(\tilde{k}_{1})),...,(\mathrm{IP}(\tilde{k}_{r}))$. This takes time $(\log_{1+\epsilon}n)^{O_\epsilon(1)}$.  Let $\hat{x}(1),\dots,\hat{x}(r)$ be the rounded solutions to $(\mathrm{IP}(\tilde{k}_{1})),...,(\mathrm{IP}(\tilde{k}_{r}))$.
In order to compute a packing of $\L$-boxes, observe that for each triplet $(t,t',t'')$ the implicit packing indicates the number of large items chosen of of profit class $\P_t$, height class $H_{t'}$ and width class $W_{t''}$. Again, we may choose these items in non-increasing order of profits in time  $O_{\epsilon}(\log^3 n + n)$ and assign them to the correct $\L$-boxes in time $O_\epsilon(n)$. For $\H$- and $\V$-boxes, we proceed similarly.  Finally, for $\S$-boxes, we need to remark that again we may choose any set of items such that for each triplet $(t,t',t'')$ the correct number of items is chosen. Thus, we again may choose the items in non-increasing order of profits and assign the correct number of items for each triplet to each box. This way we only lose a factor of $1+\epsilon$ of the profit for each $\S$-box and the selected items still satisfy the conditions necessary to pack them into $\S$ using NFDH. This takes time $O(n\log n)$ for each box.
\end{enumerate}
\end{enumerate}
This yields a total running time of $O(n\cdot(\log n))+(\log n)^{O_{\epsilon}(1)}$.

Next, we will prove the second statement of the theorem regarding the dynamic algorithm. The insertion and deletion of items in time $O(\log^3 n)$ is due to our rectangle data structure. To output a $({17/9}+\epsilon)$-approximate solution $|ALG|$ in time $O(|ALG|\cdot(\log n))+(\log n)^{O_{\epsilon}(1)}$, we use the algorithm described above with a refinement of the running time since we will choose at most $|ALG|$ items. If one queries for a $({17/9}+\epsilon)$-estimate of the optimal {objective function} value, the algorithm above without the construction of the actual packings, i.e., only the computation of the implicit packings can be {executed} in time $(\log n)^{O_{\epsilon}(1)}$. Observe that for the tiny items we now need to compute an estimate of the total profit, which we can do in time $O_{\epsilon}(\log^3n)$ using our rectangle data structure
by querying for the number of tiny items for each profit class and then multiplying this with the rounded profit.
%\awr{The tiny items do not all have the same rounded profit, no? I guess you want to make a query for the total profit of all items in a rectangular area within the item data structure?}
Finally, if one wishes to query whether an item $i \in \I$ is contained in the solution $ALG$, we compute the implicit solutions in time $(\log n)^{O_{\epsilon}(1)}$. Let $(t,t',t'')$ be such that $i$ is of profit type $\P_t$, height type $H_{t'}$ and width type $W_{t''}$ and let $z_{t,t',t''}$ be the total number of items of these types in the implicit solution. Note that the choice of the triplet $(t,t',t'')$ may depend on {the type of item~$i$}.
Then by the construction of our packing {it is sufficient to} check whether $i$ is among the first $z_{t,t',t''}$ {items} when items are ordered in non-increasing order of profits. This can be done using our rectangle data structure by counting the number of items of classes $(t,t',t'')$ that have smaller profit than $i$ which we can do with a single range counting query in time $O(\log^2 n)$. If this value is at least $z_{t,t',t''}$, we answer the query with ``no" and otherwise with ``yes''. This ensures that we give consistent answers between two consecutive updates of the set $\I$.

\end{proof}

\section{Details of Data Structures}\label{app:data_struc}
An important ingredient of our algorithms is the usage of range counting/reporting data structures. The goal is to construct a data structure whose input points are points in $\mathbb{R}^d$ characterised by points $x_1,\dots,x_n \in \mathbb{R}^d$ and weights $f_1,\dots,f_n$. As can be found in the survey by Lee and Preparata~\cite{lee1984computational} (and references therein), one can construct data structures which allow the following operations:
\begin{itemize}
	\item Insertion and deletion of a new point in time $O(\log^d n)$.
	\item Given a hyperrectangle $A= [a_1,b_1] \times \dots \times [a_d,b_d]$, all points $x_i \in A$ can be reported in time $O(\log^{d-1}n + k)$ where $k$ is the number of input points in $A$.
	\item Given a hyperrectangle $A= [a_1,b_1] \times \dots \times [a_d,b_d]$, the sum of the weights of all points $x_i \in A$ can be reported in time $O(\log^{d-1}n)$.
\end{itemize}
To construct the data structures underlying our algorithms, we proceed as follows. For the \emph{item data structure} used in Section~\ref{sec:hypercubes}, we use a $2$-dimensional data structure where each item is stored as a point $x_i \in \mathbb{R}^2$ and the first coordinate represents the side length of the item while the second coordinate represents the profit of the item. Furthermore, we set $f_i := 1$ for any item. This allows the operations described in Lemma~\ref{lem:data-structure}. 

For the \emph{rectangle data structure}, we consider points in $\mathbb{R}^4$ with a coordinate for the width, height, profit and density of each item. To allow our range counting queries, we use three data structures where one uses $f_i := 1$, one uses $f_i := p_i$ and the third uses $f_i := w_i$. This data structure allows all operations in the given query times mentioned in Lemma~\ref{lem:data-structure-rec}.

\end{document}